\documentclass[journal,10pt,onecolumn,draftclsnofoot,]{IEEEtran}
\usepackage[a4paper, margin=1in]{geometry}
\usepackage{amsmath,amssymb, amscd, amsthm, amsfonts}
\usepackage{mathtools}
\usepackage{graphicx}
\usepackage{multirow}
\usepackage{enumitem} 
\usepackage[colorlinks = true,
             linkcolor = blue,
             urlcolor  = blue,
             citecolor = blue,
             anchorcolor = blue]{hyperref}
\usepackage{url}
\usepackage{cleveref}
\usepackage{xcolor}
\usepackage{diagbox}
\usepackage{tikz}
\usepackage{longtable}
\usepackage{makecell}
\usepackage{colortbl}

\DeclareRobustCommand{\rchi}{{\mathpalette\irchi\relax}}
\newcommand{\irchi}[2]{\raisebox{\depth}{$#1\chi$}} 
\usetikzlibrary{shapes.geometric, arrows,calc}
\usetikzlibrary{positioning}
\usetikzlibrary{decorations.pathreplacing}
\tikzstyle{simple_box} = [rectangle, rounded corners, minimum height=1cm, text centered, draw=black,align=center,font=\small]
\tikzstyle{simple_box_b} = [rectangle, rounded corners, minimum height=1cm, text centered, align=center,font=\small]
\tikzstyle{arrow} = [thick,->,>=stealth]
\tikzstyle{darrow} = [<->,>=stealth]

\title{Extending fibre nonlinear interference power modelling to account for general dual-polarisation 4D modulation formats}

\author{Gabriele Liga,~\IEEEmembership{Member,~IEEE},
Astrid Barreiro,~\IEEEmembership{Member,~OSA},
Hami Rabbani, 
and Alex Alvarado,\\~\IEEEmembership{Senior Member,~IEEE}.

\thanks{The authors are with the Information and Communication Theory Lab, Signal Processing Systems Group, Department of Electrical Engineering, Eindhoven University of Technology, 5600 MB, Eindhoven, The Netherlands (e-mail: \url{g.liga@tue.nl}). H.~Rabbani is also with the EE Dept. of K. N. Toosi University of Technology.}
\thanks{G.~Liga is funded by the EuroTechPostdoc programme under the European Union's Horizon 2020 research and innovation programme (Marie Skłodowska-Curie grant agreement No 754462). This work has received funding from the European Research Council (ERC) under the European Union's Horizon 2020 research and innovation programme (grant agreement No. 757791).}
}

\newcounter{theo}
\newtheorem{theorem}[theo]{Theorem}
\newtheorem{proposition}[theo]{Proposition}
\newtheorem{corollary}[theo]{Corollary}
\newtheorem{lemma}[theo]{Lemma}
\newtheorem{example_a}{Example}

\renewcommand{\Re}{\operatorname{Re}}

\begin{document}
\maketitle
\begin{abstract}
In optical communications, four-dimensional (4D) modulation formats encode information onto the quadrature components of two arbitrary orthogonal states of polarisation of the optical field. These formats have recently regained attention due their potential power efficiency, nonlinearity tolerance, and ultimately to their still unexplored shaping gains. As in the fibre-optic channel the shaping gain is closely related to the nonlinearity tolerance of a given modulation format, predicting the effect of nonlinearity is key to effectively optimise the transmitted constellation. Many analytical models available in the optical communication literature allow, within a first-order perturbation framework, the computation of the average power of the nonlinear interference (NLI) accumulated in coherent fibre-optic transmission systems. However, all current models only operate under the assumption of a transmitted polarisation-multiplexed, two-dimensional (PM-2D) modulation format. PM-2D formats represent a limited subset of the possible dual-polarisation 4D formats, namely, only those where data transmitted on each polarisation channel are mutually independent and identically distributed. This document presents a step-by-step mathematical derivation of the extension of existing NLI models to the class of arbitrary dual-polarisation 4D modulation formats. In particular, the methodology adopted follows the one of the popular enhanced Gaussian noise model, albeit dropping most assumptions on the geometry and statistic of the transmitted 4D modulation format. The resulting expressions show that, whilst in the PM-2D case the NLI power depends only on different statistical high-order moments of each polarisation component, for a general 4D constellation also several others cross-polarisation correlations need to be taken into account. 
\end{abstract}
\newpage

\section{Introduction}
With the resurgence of polarisation-diverse, optical coherent detection, transmission of information over an optical fibre is typically performed exploiting four degrees of freedom of the optical field: two quadrature components over two orthogonal states of polarisation. The standard approach consists in encoding data independently over the two polarisation channels using the same two-dimensional (2D) modulation format. The resulting four-dimensional (4D) constellation is often referred to as a \emph{polarisation-multiplexed} 2D (PM-2D) modulation format. The strong point of PM-2D formats is their simplicity of generation and performance analysis: as the two polarisation channels are independent and under the assumption of data-independent cross-polarisation interference in the fibre channel, transmission performance can be evaluated using the 2D component format.  

Despite the popularity of PM-2D formats, a substantial amount of research work in the literature has been devoted to more general 4D formats, i.e. 4D constellations which are not necessarily generated as Cartesian products of a component 2D constellation \cite{Agrell09, Karlsson09}. The reason relies on the fact that, by exploiting the full 4D space, constellation sensitivity and other relevant performance metrics such as mutual information or generalised mutual information can be improved compared to traditional PM-2D formats \cite{Alvarado2015, Eriksson2016, Kojima2017, Chen2019, Chen2020}.   
Previous works on optimised 4D modulation formats have either operated under an additive white Gaussian noise channel hypothesis \cite{Agrell09, Karlsson09, Alvarado2015}, or exploited some heuristic approaches to derive nonlinearly tolerant formats in the fibre-optic channel \cite{Kojima2017,Chen2019,Chen2020}. However, accurately predicting the amount of nonlinear interference generated by transmission of a given constellation in an optical fibre is key to optimise its shape in multiple dimensions. 
    
Modelling of nonlinear interference (NLI) in optical fibre transmission is quite a mature field of research where an impressive amount of progress was made in the first half of the 2010s, e.g., in  \cite{Poggiolini2012, Carena2014, Mecozzi2012, Dar2013}. In particular, \cite{Mecozzi2012, Dar2013} introduced for the first time the possibility of predicting the dependency of the nonlinear interference power as a function of the modulation format features, i.e. geometrical shape and statistical properties. Among other assumptions, one underlying key point of all previous models is the transmission of PM-2D modulation formats, where data on the two polarisation channels are assumed to be independent and identically distributed. Under this constraint, one can predict the NLI power using the statistical properties of the 2D component modulation format. It is clear, however, that this approach ceases to be applicable to general dual-polarisation 4D formats, where a single 2D component format might not even exist. 

In this work, we extend the existing analytical expressions for the NLI power to account for dual-polarisation 4D constellations where the two 2D polarisation components are not identically distributed or when, due to its properties in 4D (geometry and probability distribution), there is statistical dependency between them. The undertaken approach is the same as in \cite{Carena2014}, i.e. a frequency-domain, first-order perturbation study. However, unlike \cite{Carena2014}, no assumptions are made on either the marginal or joint statistics of the two polarisation components of the transmitted 4D constellation (besides being zero-mean). The final expressions reveal the impact of several cross-polarisation statistics on the NLI power.

The formulas presented in this work enable an accurate computation of the NLI power for \emph{all possible dual-polarisation formats in optical fibre transmission}. As a result, a reliable optimisation of both geometry and symbol probability of occurrence of such 4D formats is also enabled for the optical fibre channel.            
\section{Organisation of the document and notation}\label{sec:preamble}
The document is organised as follows: i) in Sec.~\ref{sec:model assumptions} the investigated system model is described and the model assumptions are presented; ii) Secs.~\ref{sec:PSD_periodic} to \ref{sec:final result} are devoted to a step-by-step analytical derivation of the model; iii) ultimately, the main model expression is presented in Sec.~\ref{sec:final result} (see Theorem \ref{th:keyresult}).  
In particular: in Sec.~\ref{sec:PSD_periodic}, the regular perturbation (RP) solution to the frequency domain Manakov equation is derived for a multi-span fibre system and its power spectral density (PSD) is evaluated, in the case of a transmitted periodic signal; in Sec.~\ref{sec:subset_classification}, the contributions of the different high-order moments and cross-polarisation correlations of the transmitted 4D modulation format are highlighted; in Sec.~\ref{sec:evaluation}, these contributions are separately evaluated; in Sec.~\ref{sec:sumallcontr}, all contributions are added together and the final expression for a periodic transmitted signal is derived; finally, Sec.~\ref{sec:final result} introduces the final result for general aperiodic signals.

Throughout this manuscript, we denote 2D (column) vectors with boldface letters (e.g., $\boldsymbol{a}$), whereas 2D column vector functions are indicated with boldface capital letters (e.g., $\boldsymbol{E}(f,z),\tilde{\boldsymbol{E}}(t,z)$, etc.). 
$\mathcal{F}\{\cdot\}$, $\mathbb{E}\{ \cdot \}$, and $\Re\{\cdot\}$ indicate the Fourier transform, the statistical expectation, and the real part operators, respectively. The delta distribution is indicated by $\delta(\cdot)$, whereas $\delta_k$ denotes the Kronecker delta defined as  
\begin{equation*}
    \delta_{k}\triangleq 
    \left\{\begin{array}{c}1 \qquad \text{for} \; k=0, \\ 0 \qquad \text{elsewhere}. \end{array}\right.
\end{equation*}
Finally, $\mathbb{Z}$ and $\mathbb{C}$ denote the  integer and complex fields, respectively, and $j$ is the imaginary unit. 
\section{Model assumptions}\label{sec:model assumptions}

\subsection{System model}\label{sec:system_model}
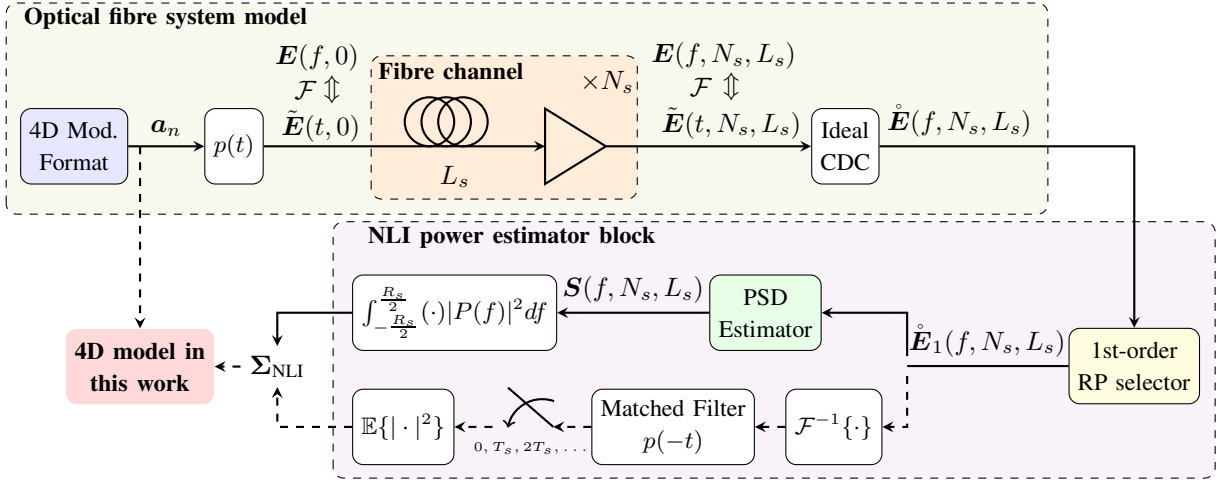
\begin{figure}[tbp]
	\centering
	 \begin{tikzpicture}
    \draw[dashed, rounded corners,fill=olive!5] (-2.3,-0.9) rectangle ++ (13.7,2.8);
	\draw[dashed, fill=orange!15,rounded corners] (2.5,-.7) rectangle ++ (3.5,1.9);
	\draw[dashed, rounded corners,fill=violet!5] (2,-4.4) rectangle ++ (11.6,3.4);
	
	\node (modformat) [simple_box,align=center,fill=blue!10] at (-40pt,0pt) {4D Mod.\\ Format};
	\node (shaping) [simple_box, right=1cm of modformat ,align=center,minimum width=0.1cm,fill=white] {$p(t)$};
	
	\node(Fibre)  [coordinate, right=2.4cm of shaping] {};
	\draw (Fibre) node [thick, anchor=south,circle,minimum size=0.65cm,outer sep=0, draw,xshift=-0.22cm] {};
	\draw (Fibre) node [thick, anchor=south,circle,minimum size=0.65cm,outer sep=0,draw] {};
	\draw (Fibre) node [thick, anchor=south,circle,minimum size=0.65cm,outer sep=0,draw,xshift=+0.22cm] {};
	\node(EDFA_left) [coordinate, right=1.3cm of Fibre] {};
	\coordinate(EDFA_right) at ($(EDFA_left)+(.8,0)$);
	\draw[thick] ($(EDFA_left)+(0,-0.5)$) -- ($(EDFA_left)+(0,+0.5)$) --  (EDFA_right) -- cycle;

	\node (cd) [simple_box, right=2.7cm of EDFA_right,fill=white] {Ideal \\ CDC};
	\node (1st_select) [coordinate, right=0.5cm of cd] {};
	\node (1st_select_b) [simple_box, below=1.9cm of cd,xshift=3.8cm ,align=center,fill=yellow!15] {1st-order\\RP selector};
	\node(1st_select_b_right) [coordinate, left=-0.2 of 1st_select_b] {};
	\node(1st_select_b_left) [coordinate, left=0.7 of 1st_select_b] {};
	\node (psde) [simple_box, fill=green!10, below =1.7cm of 1st_select, xshift=-2cm] {PSD \\ Estimator};
	\node (ift)  [simple_box, below=0.5cm of psde, xshift=.9cm,fill=white] {$\mathcal{F}^{-1}\{\cdot\}$};
	\node (integ) [simple_box, left= 2cm of psde,fill=white] {$\int_{-\frac{R_s}{2}}^{\frac{R_s}{2}}(\cdot)|P(f)|^2df$};
	\node (mf)  [simple_box, left=0.4cm of ift,fill=white] {Matched Filter \\ $p(-t)$};
	\node (samp_r) [coordinate, left=.2cm of mf] {};
	\node (samp_l) [coordinate, left=1cm of samp_r] {};
	\node (var)  [simple_box, left=0.6cm of samp_l,fill=white] {$\mathbb{E}\{|\cdot|^2\}$};
	
	\draw [-,thick] (1st_select_b) --++(-55pt,0pt)node[above]() {$\mathring{\boldsymbol{E}}_1(f,N_s,L_s)$}--++(-30pt,0pt)node[](1st_select_b_left){}; 
	\draw [arrow] (1st_select_b_left) |- (psde);
	\draw [dashed,arrow] (1st_select_b_left) |- (ift);
	\draw [arrow] (psde) -- node[anchor=south] {$\boldsymbol{S}(f,N_s,L_s)$} (integ);

	\draw [arrow] (modformat)-- node[anchor=south]{$\boldsymbol{a}_n$} (shaping);
	\draw [thick] (shaping) -- node[anchor=south,xshift=-0.5cm,yshift=.87cm] {$\boldsymbol{E}(f,0)$} (Fibre);
	\node[below=0.1cm of Fibre,xshift=-1.5cm,yshift=1.15cm] (udarrow1){$\Updownarrow$};
	\node[left=-0.15cm of udarrow1] (F){$\mathcal{F}$};
	\node[below=-0.15cm of udarrow1,xshift=-0.15cm] (Ef){$\tilde{\boldsymbol{E}}(t,0)$};
	
	\draw [arrow] (Fibre) -- (EDFA_left);
	\draw [arrow] (EDFA_right) -- node[anchor=south,xshift=0.2cm,yshift=.9cm] {$\boldsymbol{E}(f,N_s,L_s)$}  (cd);
	\node[below=0.1cm of cd,xshift=-1.5cm,yshift=1.7cm] (udarrow2){$\Updownarrow$};
	\node[left=-0.1cm of udarrow2] (F){$\mathcal{F}$};
	\node[below=-0.15cm of udarrow2] (Ef){$\tilde{\boldsymbol{E}}(t,N_s,L_s)$};
	\draw [arrow] (cd) -| node[anchor=south,xshift=-2.3cm] {$\mathring{\boldsymbol{E}}(f,N_s,L_s)$} (1st_select_b);
	
    \node[left=0.5cm of var,yshift=.8cm] (sigma) {$\boldsymbol{\Sigma}_{\text{NLI}}$};
	\draw [arrow] (integ)--(integ-|sigma)--(sigma);
	\draw [dashed,arrow] (var)--(var-|sigma)--(sigma);
    \node (4Dmod) [simple_box_b,left=0.34cm of sigma,align=center,fill=red!15] {\textbf{4D model in} \\ \textbf{this work}};
	\draw [dashed,arrow] (sigma)--(4Dmod);
	\draw [dashed,arrow](modformat-|4Dmod)--(4Dmod);
	
	\node (S) at (5.6,0.85) {$\times N_s$};
	\node (L) at (3.6,-0.4) {$L_s$};
	\node (SM) at (-0.2,1.7) {\small\textbf{Optical fibre system model}};
	\node (NLIV) at (4.35,-1.2) {\small\textbf{NLI power estimator block}};
	\node (FC) at (3.55,1) {\small\textbf{Fibre channel}};
	
	\node (samp_r) [coordinate, left=0.5cm of mf] {};
	\node (samp_l) [coordinate, left=.8cm of samp_r] {};
	\draw [dashed,arrow] (mf) -- (samp_r);
	\draw [dashed,arrow] (ift)-- (mf);
	\draw [dashed,arrow] (samp_l) -- (var);
	
	\draw [thick] (samp_r) -- (4.3,-3.2);
	\draw [thick,<-] (4.3,-3.73) arc (170:70:13pt);
    \node (T) at (4.6,-4) {\tiny{$0,T_s,2T_s, \ldots$}};

\end{tikzpicture}
	\caption{System model under investigation in this document which consists of an optical fibre system model and an NLI variance estimation block. The two branches in the NLI variance estimator block indicate alternative ways of estimating $\boldsymbol{\Sigma}_{\text{NLI}}$.}
	\label{fig:system_model}
\end{figure}

The baseband equivalent model of the optical fibre system under investigation in this document is shown in Fig.~\ref{fig:system_model}. The fibre channel is a multi-span fibre system using Erbium-doped fibre amplification. In this manuscript, it is assumed that a single-channel signal is transmitted. The transmitter is assumed to generate for each symbol period $n$ the 4D symbol $\boldsymbol{a}_{n}=[a_{x,n}, a_{y,n}]^{T}$ where $a_{x,n},a_{y,n}\in \mathbb{C}$ are complex symbols modulated on two arbitrary orthogonal polarisation states $x$ and $y$, respectively. Linear modulation with a single, real, pulse $p(t)$ on $x$ and $y$ polarisation is adopted. 
The pulse $p(t)$ with spectrum $P(f)$ is assumed to be strictly band-limited within the range of frequencies 
$[-R_s/2, R_s/2]$.
As discussed in Sec.~\ref{subsec:sig_period}, the transmitted signal $\tilde{\boldsymbol{E}}(t,0)$\footnote{In this paper, the first variable of the optical field represents either the time or frequency variable, whereas the second one represents the fibre propagation section. An exception is made for the multi-span system case, where second and third variable are assigned to the number of spans and span length, resp. This highlights the joint dependence of the output optical field on these two variables, as shown later in the paper.} is assumed to be periodic with period $T$, such that 
\begin{equation}\label{eq:lin_mod}
\tilde{\boldsymbol{E}}(t,0)=\sum_{n=0}^{W-1}\boldsymbol{a}_{n}p(t-nT_s), \qquad \text{for} \;\; 0\leq t \leq T,   
\end{equation}
and $T_s=1/R_s=T/W$ represents the symbol period, and $R_s$ is the symbol rate. A schematic representation of the transmitted signal is shown in Fig.~\ref{fig:time_sketch}.

The signal is transmitted over $N_s$ (homogeneous) fibre spans, each of length $L_s$ and each followed by an ideal lumped optical amplifier whose gain exactly recovers from the span losses. Since in this document we are only concerned about the prediction of NLI arising from the signal-signal nonlinear interactions along the fibre propagation, the optical noise added by the amplifier plays no role in the model and will be entirely neglected. The signal at the channel output $\tilde{\boldsymbol{E}}(t,N_s,L_s)$ is ideally compensated for accumulated chromatic dispersion in the link (see Sec.~\ref{sec:PSD_periodic}). From the resulting frequency-domain signal $\mathring{\boldsymbol{E}}(f,N_s,L_s)$ (Fig.~\ref{fig:system_model}) we ideally isolate the first-order regular perturbation (RP) term $\boldsymbol{E}_1(f,N_s,L_s)$ (see Sec.~\ref{sec:PSD_periodic}) and we compute its PSD $\boldsymbol{S}(f,N_s,L_s)$. The vector of the NLI powers $\boldsymbol{\Sigma}_{\text{NLI}}\triangleq [\sigma^2_{\text{NLI},x},\sigma^2_{\text{NLI},y}]^{T}$ for both $x$ and $y$ polarisations, is obtained by integrating over the frequency interval $[-R_s/2, \; R_s/2]$ the NLI PSD weighted by the function $|P(f)|^2$, where $P^*(f)$ is the frequency response of a matched filter (MF) for the system under consideration. As shown in Fig.~\ref{fig:system_model}, this quantity is equivalent to the variance of the output of the MF followed by symbol-rate sampling, which more naturally arises when assessing the transmission performance of systems employing an MF at the receiver. The model in this manuscript provides an analytical relationship between the statistical features of the transmitted symbols $\boldsymbol{a}_n$ and $\boldsymbol{\Sigma}_{\text{NLI}}$.    

\begin{figure}[t!]
	\centering
	 \begin{tikzpicture}
		\draw[arrow] (-4,0) -- (7,0) node [below] {{\footnotesize$t$~[s]}};
		\draw[thick] (-4,.1) -- (-4,-.1) node [below] {{\footnotesize $0$}};
		\draw[thick] (-3.3,.1) -- (-3.3,-.1) node [below] {{\footnotesize $T_s$}};
		\draw[thick] (-2.6,.1) -- (-2.6,-.1) node [below] {{\footnotesize $2T_s$}};
		\draw[thick] (-1.9,.1) -- (-1.9,-.1) node [below] {{\footnotesize $3T_s$}};
		\draw[thick] (0,.1) -- (0,-.1) node [below,xshift=.5cm] {{\footnotesize $T=WT_s$}};
		\draw[thick] (4,.1) -- (4,-.1) node [below] {{\footnotesize $2T$}};
		\draw[darrow] (0,-.9) -- node[anchor=north] {\small Period $T=\frac{1}{\Delta_f}$} (4,-.9) ;
		
		\draw [decorate,decoration={brace,amplitude=5pt}](-4,.3) -- (-.05,.3) node [black,midway,yshift=0.4cm] {\footnotesize $W$ symbols, $\hat{E}(t,0)$};
		\draw [decorate,decoration={brace,amplitude=5pt}](0.05,.3) -- (4,.3) node [black,midway,yshift=0.4cm] {\footnotesize $W$ symbols, $\hat{E}(t-T,0)$};
		\node (D0) at (5.5,0.35) {$\dots$};
		
		\node (D1) at (5.5,-0.4) {$\dots$};
		\node (D2) at (-1,-0.4) {$\dots$};
		\node (Ts) at (-2,-1.2) {$T_s=\frac{1}{R_s}$};
		\node (Df) at (5.35,-.9) {$\Delta_f\rightarrow 0$};
		\node (A)  at (5.5,-1.2){$\Downarrow$};
		\node (Df) at (5.55,-1.5) {$T\rightarrow \infty \Leftrightarrow W\rightarrow \infty $};
\end{tikzpicture}
	\caption{Schematic representation of the periodic signal assumption where $W$ symbols are transmitted every $T$~[s], each symbol with a duration of $T_s$~[s]. The periodicity assumption will be lifted in Sec.~\ref{sec:final result} by making $\Delta_f\rightarrow 0$ .}
	\label{fig:time_sketch}
\end{figure}
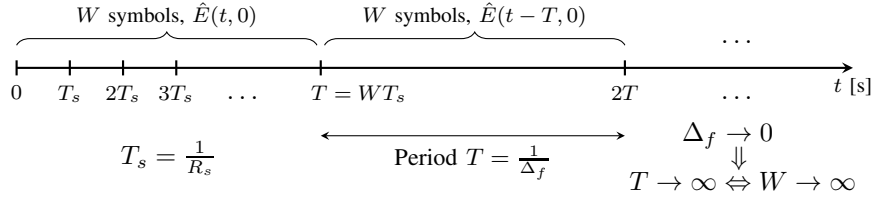

\subsection{Dual-polarisation 4D vs. PM-2D formats}\label{subsec:4Dformats}
The model presented in this document allows the prediction of the NLI for generic 4D real modulation formats. A 4D format is defined as a set

\begin{equation}
    \mathcal{A}\triangleq\{\boldsymbol{a}^{(i)}=[a^{(i)}_x,a^{(i)}_y]^{T}\in \mathbb{C}^2, \;\; i=1,2,\dots,M\},
    \label{eq:4Dsymbols}
\end{equation}

where $a_x$ and $a_y$ are the symbols modulated on two orthogonal polarisation states $x$ and $y$, respectively, and $M$ is the modulation cardinality. It can be seen that the elements in $\mathcal{A}$ are 2D vectors in $\mathbb{C}$ as opposed to 4D. This is only due the to baseband-equivalent representation of signals used throughout this paper, while it is common to refer to a modulation format dimensionality based on the real signal dimensions, which justifies the 4D format label. 

Two important particular cases of the formats in \eqref{eq:4Dsymbols} are: i) the so-called polarisation-multiplexed 2D (PM-2D) modulation formats, which are characterised by $\mathcal{A}=\mathcal{X}^2$, $\mathcal{X}\in\mathbb{C}$ where $\mathcal{X}$ represents the 2D component constellation; ii) polarisation-hybrid 2D modulation formats characterized by $\mathcal{A}=\mathcal{X}\times\mathcal{Y}$, with $\mathcal{X},\mathcal{Y}\in\mathbb{C}$, $\mathcal{X}\neq\mathcal{Y}$, where $\mathcal{X}$ and $\mathcal{Y}$ are two distinct component 2D formats in $x$ and $y$ polarisation, respectively. PM-2D formats are the most common ones in optical communications due to their generation's simplicity. Both PM-2D and polarisation-hybrid 2D formats are often analysed in terms of their 2D polarisation components. This is because $\mathcal{A}$ can be factorised in two-component formats. If the generic transmitted constellation point is regarded as a random variable, in a conventional PM-2D format the two polarisation components are statistically independent. In the remainder of this paper, no specific assumption on either the geometry or the statistic of the transmitted 4D symbols will be made, except the zero-mean feature $\mathbb{E}\{\boldsymbol{a}^{(i)}\}=\textbf{0}.$

\subsection{Transmitted signal form}\label{subsec:sig_period}
Let $\tilde{\boldsymbol{E}}(t,z)=\tilde{E}_x(t,z)\boldsymbol{i}_x+\tilde{E}_y(t,z)\boldsymbol{i}_y$ be the complex envelope of the optical field vector at time $t$ and fibre section $z$, and $\boldsymbol{i}_x$, $\boldsymbol{i}_y$ denote 2 linear orthonormal polarisations along direction $x$ and $y$, respectively of the transversal plane of propagation. Let also $\boldsymbol{E}(f,z)=E_x(f,z)\boldsymbol{i}_x+E_y(f,z)\boldsymbol{i}_y$ be
the (vector) Fourier transform of $\tilde{\boldsymbol{E}}(t,z)$ defined as 
\begin{equation*}
\boldsymbol{E}(f,z)=\mathcal{F}\{\tilde{\boldsymbol{E}}(t,z)\}\triangleq\int_{-\infty}^{\infty}\tilde{\boldsymbol{E}}(t,z)e^{-j2\pi ft}dt.
\label{eq:FTdef}    
\end{equation*}

Because of the periodicity assumption made in \eqref{eq:lin_mod} (see Fig.~\ref{fig:time_sketch}), $\tilde{\boldsymbol{E}}(t,0)$ we can write $\tilde{\boldsymbol{E}}(t,0)$ as 
\begin{equation}
\tilde{\boldsymbol{E}}(t,0)=\sum_{k=-\infty}^{\infty} \boldsymbol{C}_k\,e^{j2\pi k \Delta_ft}, 
\label{eq:FourierSeries}    
\end{equation}
where $\boldsymbol{C}_k=[C_{x,k},C_{y,k}]^T$, $C_{x/y,k}$ are the Fourier series coefficients of $\tilde{\boldsymbol{E}}(t,0)$, and $\Delta_f=1/T$ is the frequency spacing of the spectral lines in $E_{x/y}(f,z)$. Hence, $\boldsymbol{E}(f,0)$ can be then written as 
\begin{equation}
\begin{split}
\boldsymbol{E}(f,0)=\sum_{k=-\infty}^{\infty} \boldsymbol{C}_{k}\,\delta(f-k\Delta_f).
\end{split}
\label{eq:periodic_spectrum}    
\end{equation}
Since each component of $\tilde{\boldsymbol{E}}(t,0)$ is periodic with period $T$, we can write 

\begin{equation*}
\tilde{\boldsymbol{E}}(t,0)=\sum_{n=-\infty}^{\infty}\hat{\boldsymbol{E}}(t-nT,0),
\end{equation*}
where, as per assumption in \eqref{eq:lin_mod}, we have
\begin{equation*}
\hat{\boldsymbol{E}}(t,0)\triangleq
\begin{dcases}
\sum_{n=0}^{W-1} \boldsymbol{a}_{n} p(t-nT_s), & \text{for} \;\; 0\leq t\leq T \\
0,& \text{ otherwise}
\end{dcases}.
\end{equation*}

Under the above assumptions, the Fourier coefficients in \eqref{eq:FourierSeries}, for $k \in \mathbb{Z}$, are given by
\begin{subequations}
\begin{align}
\boldsymbol{C}_{k}&=\Delta_f \int_{0}^{T} \hat{\boldsymbol{E}}(t,0) e^{-j2\pi k \Delta_f t}dt\label{eq:FourierCoeff1}   \\
&=\Delta_f \int_{0}^{T} \sum_{n=0}^{W-1}\boldsymbol{a}_{n}p(t-nT_s)  e^{-j2\pi k \Delta_f t}dt \label{eq:FourierCoeff2} \\
&=\Delta_f \sum_{n=0}^{W-1}\boldsymbol{a}_{n}\int_{0}^{T} p(t-nT_s)  e^{-j2\pi k \Delta_f t}dt \label{eq:FourierCoeff3}\\
&\approx\Delta_f \sum_{n=0}^{W-1}\boldsymbol{a}_{n}P( k \Delta_f)e^{-j2\pi k \Delta_f n T_s} \label{eq:FourierCoeff4}\\
&=\Delta_f P( k \Delta_f)\sum_{n=0}^{W-1}\boldsymbol{a}_{n}e^{-j2\pi \frac{kn}{W}} \label{eq:FourierCoeff5}\\
&= \sqrt{\Delta_f} P( k \Delta_f)\boldsymbol{\nu_{k}},
\label{eq:FourierCoeff6}
\end{align}
\end{subequations}
where $P(f)\triangleq\mathcal{F}\{p(t)\}$ and 
\begin{equation}
\boldsymbol{\nu}_{k}=[\nu_{x,k},\nu_{y,k}]^T=\sqrt{\Delta_f} \sum_{n=0}^{W-1}\boldsymbol{a}_{n}e^{-j2\pi \frac{kn}{W}}, \qquad\qquad  \forall \; k \in \mathbb{Z},
\label{eq:symbol_fft}
\end{equation}
is the discrete Fourier transform of the sequence $\boldsymbol{a}_{n}$, $n=0,1,2,\ldots, W-1$. 
Note that the approximation in \eqref{eq:FourierCoeff3}--\eqref{eq:FourierCoeff4} is justified only for large enough values of $T$ as\footnote{Assuming without loss of generality that $p(t)$ is causal.}
\begin{equation*}
\lim_{T\rightarrow\infty}\int_{0}^{T} p(t-nT_s)  e^{-j2\pi k \Delta_f t}dt=\mathcal{F}\{p(t-nT_s)\}\bigr|_{f=k\Delta_f}, 
\end{equation*}
and letting $T\rightarrow\infty$ will be the approach taken at a later stage in this derivation.

Finally, combining \eqref{eq:periodic_spectrum} and \eqref{eq:FourierCoeff6} we obtain 
\begin{equation}
\boldsymbol{E}(f,0)=\sqrt{\Delta_f} \sum_{k=-\infty}^{\infty}P( k \Delta_f)\,\boldsymbol{\nu}_{k}\;\delta(f-k \Delta_f)\approx \sum_{k=-W/2}^{W/2}P( k \Delta_f)\,\boldsymbol{\nu}_{k}\;\delta(f-k \Delta_f). 
\label{eq:disc_spec_1}    
\end{equation}
where the approximate equality on the right-hand side of \eqref{eq:disc_spec_1} stems from the fact that $p(t)$ is assumed to be strictly or quasi strictly band-limited (see Sec.~\ref{sec:system_model}). Hence, $P(k \Delta_f)$ is effectively equal to zero for $k=-W/2,-W/2+1,\ldots,W/2$.\footnote{Here, $W$ is assumed even without loss of generality.}
\section{PSD of the first-order NLI for periodic transmitted signals}\label{sec:PSD_periodic}

To find an analytical expression for the NLI power, first a solution as explicit as possible to the Manakov equation \cite{Marcuse1997}
\begin{align}
    \frac{\partial \tilde{\boldsymbol{E}}(t,z)}{\partial z}&=-\frac{\alpha}{2}\tilde{\boldsymbol{E}}(t,z)-j\frac{\beta_2}{2}\frac{\partial^2 \tilde{\boldsymbol{E}}(t,z)}{\partial t^2}+j\frac{8}{9}\gamma |\tilde{\boldsymbol{E}}(t,z)|^2 \tilde{\boldsymbol{E}}(t,z),
    \label{eq:Manakov}
\end{align}
must be found. Eq.~\eqref{eq:Manakov} describes the propagation of the optical field $\tilde{\boldsymbol{E}}(t,z)$ in a single strand of fibre (e.g.,~a fibre span with no amplifier in the system in Fig.~\ref{fig:system_model}). In this case, $\alpha$, $\beta_2$ and $\gamma$ representing the attenuation, group velocity dispersion and nonlinearity coefficients, respectively, can be assumed to be spatially constant. As it is well known, general closed-form solutions are not available for \eqref{eq:Manakov}. Like most of the existing NLI power models in the literature, the model derived here operates within a first-order perturbative framework. In particular, a frequency-domain first-order regular perturbation (RP) approach in the $\gamma$ coefficient is performed \cite{Vannucci2002, Johannisson2013}, i.e., the Fourier transform of the solution in \eqref{eq:Manakov} expressed as

\begin{equation}\label{eq:RP_general}
\boldsymbol{E}(f,z)=\sum_{n=0}^{\infty}\gamma^n\boldsymbol{A}_n(f,z)\approx\boldsymbol{A}_0(f,z)+\gamma\boldsymbol{A}_1(f,z),
\end{equation}
where 
\begin{equation}
\boldsymbol{E}_n(f,z)=\gamma^n\boldsymbol{A}_n(f,z) \qquad \text{for} \;\; n=0,1,....,
\label{eq:AtoE}
\end{equation}
represents the so-called $n$th order term of the expansion.

In the following theorem, we present the expressions for $\boldsymbol{E}_0(f,z),$ and $\boldsymbol{E}_1(f,z),$ when a multiple fibre span system like the one in Fig.~\ref{fig:system_model} is considered. These expressions are well-known in the literature (see, e.g., \cite{Vannucci2002}). Nevertheless, we present the proof in Appendix \ref{app:theoRP} for completeness.

\begin{theorem}[First-order frequency-domain RP solution for a multi-span fibre system]\label{The:RP_NLSE}
Let $\boldsymbol{E}(f,z)$ be the solution in frequency-domain of the Manakov equation for the system in Fig.~\ref{fig:system_model} with initial condition at distance $z=0$ given by the transmitted signal $\boldsymbol{E}(f,0)$. Then, the first-order RP solution after $N_s$ spans $\boldsymbol{E}(f,N_s,L_s)$ is given by
\begin{equation*}
\boldsymbol{E}(f,N_s,L_s)\approx\boldsymbol{E}_0(f,N_s,L_s)+\boldsymbol{E}_1(f,N_s,L_s), 
\end{equation*}
where the zeroth-order term is given by 
\begin{equation*}
\boldsymbol{E}_0(f,N_s,L_s)= \boldsymbol{E}_0(f,0)e^{j2\pi^2f^2\beta_2N_sL_s}, 
\end{equation*}
and the first-order term is
\begin{equation}
\boldsymbol{E}_{1}(f,N_s,L_s)=-j\frac{8}{9}\gamma e^{j2\pi^2 f^2\beta_2 N_sL_s}\int_{-\infty}^{\infty}\int_{-\infty}^{\infty}\boldsymbol{E}^{T}(f_1,0)\boldsymbol{E}^*(f_2,0)\boldsymbol{E}(f-f_1+f_2,0) \eta(f_1,f_2,f,N_s,L_s)df_1 df_2,
\label{eq:1st_order_RP_solution_f}
\end{equation}
with
\begin{equation} \label{eq:fwm_efficiency}
 \eta(f_1,f_2,f,N_s,L_s)\triangleq\frac{1-e^{-\alpha L_s}e^{j4\pi^2\beta_2(f-f_1)(f_2-f_1)L_s}}{\alpha-j4\pi^2\beta_2(f-f_1)(f_2-f_1)}\sum_{l=1}^{N_s}e^{-j4\pi^2\beta_2 (l-1)(f-f_1)(f_2-f_1)L_s},
\end{equation}
where $N_s$ and $L_s$ are the number of spans and the span length of each span, respectively. 
\end{theorem}
\begin{proof}
See Appendix \ref{app:theoRP}.
\end{proof}

While Theorem \ref{The:RP_NLSE} gives an approximation for the field at the output of the fibre, we are interested in the field after ideal CDC (see Fig.~\ref{fig:system_model}). Ideal CDC ideally removes the exponential $e^{j2\beta_2\pi^2 f^2 N_sL_s}$ from \eqref{eq:1st_order_RP_solution_f}, leading to

a first-order term in the RP solution for the system in Fig.~\ref{fig:system_model} given by

\begin{equation}
\mathring{\boldsymbol{E}}_1(f,N_s,L_s)=[\mathring{E}_{1,x},\mathring{E}_{1,y}]^{T}=-j\frac{8}{9}\gamma\int_{-\infty}^{\infty}\int_{-\infty}^{\infty} \boldsymbol{E}^{T}(f_1,0)\boldsymbol{E}^*(f_2,0)\boldsymbol{E}(f-f_1+f_2,0)\eta(f_1,f_2,f,N_s,L_s)df_1df_2.
\label{eq:1st_order_RP_solution_RX}    
\end{equation}

Substituting the spectrum of the transmitted periodic signal \eqref{eq:disc_spec_1} in \eqref{eq:1st_order_RP_solution_RX}, we obtain, for instance for the $x$ component in \eqref{eq:1st_order_RP_solution_RX},

\begin{equation*}
\begin{split}
\mathring{E}_{1,x}(f,N_s,L_s)&=-j\frac{8}{9}\gamma\Delta_f^{3/2}\sum_{k=-\infty}^{\infty}\sum_{m=-\infty}^{\infty}\sum_{n=-\infty}^{\infty}P(k\Delta_f)P^*(m\Delta_f)P(n\Delta_f)\left(\nu_{x,k}\nu^*_{x,m}\nu_{x,n}+\nu_{y,k}\nu^*_{y,m}\nu_{x,n}\right)\\
&\cdot\int_{-\infty}^{\infty}\int_{-\infty}^{\infty}\delta(f_1-k\Delta_f)\delta(f_2-m\Delta_f)\delta(f-f_1+f_2-n\Delta_f)\eta(f_1,f_2,f,N_s,L_s)df_1 df_2.
\end{split}
\label{eq:1st_order_RP_periodic}
\end{equation*}

Integrating in $f_1$ and $f_2$, we obtain\footnote{The product of distributions is not well-defined in the standard distribution theory framework. However, in some cases, such products can be dealt with in the same way as products between distributions and smooth functions. This approach was formalised by Colombeau in his theory of product between distributions \cite{Colombeau1984}.}
\begin{equation}
\begin{split}
\mathring{E}_{1,x}&(f,N_s,L_s)=-j\frac{8}{9}\gamma\Delta_f^{3/2}\sum_{k=-\infty}^{\infty}\sum_{m=-\infty}^{\infty}\sum_{n=-\infty}^{\infty}P(k\Delta_f)P^*(m\Delta_f)P(n\Delta_f)\\
&\cdot\left(\nu_{x,k}\nu^*_{x,m}\nu_{x,n}+\nu_{y,k}\nu^*_{y,m}\nu_{x,n}\right)\eta(k\Delta_f,m\Delta_f,(k-m+n)\Delta_f,N_s,L_s)\delta(f-(k-m+n)\Delta_f).
\end{split}
\label{eq:1st_order_RP_periodic_2}
\end{equation}

Setting $i=k-m+n$ and defining 
\begin{align}
\begin{split}
\eta_{k,m,n}&\triangleq \eta(k\Delta_f,m\Delta_f,(k-m+n)\Delta_f,N_s,L_s)\\
&=\frac{1-e^{-\alpha L_s}e^{j4\pi^2\Delta_f^2\beta_2 (n-m)(m-k)L_s}}{\alpha-j4\pi^2\Delta_f^2\beta_2 (n-m)(m-k)}\sum_{l=1}^{N_s}e^{-j4\pi^2\Delta_f^2\beta_2 (l-1)(n-m)(m-k)L_s},
\end{split}
\label{eq:eta_def}
\end{align}
\eqref{eq:1st_order_RP_periodic_2} can be rewritten as 
\begin{equation}
\mathring{E}_{1,x}(f,N_s,L_s)=\sum_{i=-\infty}^{\infty}c_i\delta(f-i\Delta_f),
\label{eq:1st_order_RP_discrete}    
\end{equation}
where 
\begin{equation}
c_i\triangleq-j\frac{8}{9}\gamma\Delta_f^{3/2}\sum_{\substack{(k,m,n) \in \mathcal{S}_i}}P(k\Delta_f)P^*(m\Delta_f)P(n\Delta_f)\left(\nu_{x,k}\nu^*_{x,m}\nu_{x,n}+
\nu_{y,k}\nu^*_{y,m}\nu_{x,n}\right)\eta_{k,m,n},
\label{eq:ci_coeff}
\end{equation}
and
\begin{equation}\label{eq:setSi}
\mathcal{S}_i\triangleq \{(k,m,n)\in \mathbb{Z}^3:k-m+n=i\}.
\end{equation}

The PSD of the received nonlinear interference (to the 1st-order) is defined as 
\begin{equation}
\boldsymbol{S}(f,N_s,L_s)=[
     S_x(f,N_s,L_s), S_y(f,N_s,L_s)]^{T} 
\triangleq\left[
\mathbb{E}\left\{|\mathring{E}_{1,x}(f,N_s,L_s)|^2\right\},
\mathbb{E}\left\{|\mathring{E}_{1,y}(f,N_s,L_s)|^2\right\}\right]^{T}.
\label{eq:PSD_NLI}
\end{equation}
For periodic signals, which in the frequency domain can be expressed as in \eqref{eq:1st_order_RP_discrete}, the PSD can be expressed as \cite[Sec.~4.1.2]{ProakisDSP3rdEd}
\begin{equation}
S_x(f,N_s,L_s)=\sum_{i=-\infty}^{\infty}\mathbb{E}\{|c_i|^2\}\delta(f-i\Delta_f).
\label{eq:PSD_form}
\end{equation}

Substituting the expression \eqref{eq:ci_coeff} for $c_i$ in \eqref{eq:PSD_form} we obtain 
\begin{subequations}
\begin{align}
\begin{split}
S_x(f,N_s,L_s)&=\left(\frac{8}{9}\right)^2\gamma^2\Delta_f^3\sum_{i=-\infty}^\infty\delta(f-i\Delta_f)\mathbb{E}\biggl\{\sum_{(k,m,n) \in \mathcal{S}_i}P(k\Delta_f)P^*(m\Delta_f)P(n\Delta_f)\left(\nu_{x,k}\nu^*_{x,m}\nu_{x,n}\right.\\
&\left.+\nu_{y,k}\nu^*_{y,m}\nu_{x,n}\right)\eta_{k,m,n}\sum_{(k^{\prime},m^{\prime},n^{\prime}) \in \mathcal{S}_i}P^*(k^{\prime}\Delta_f)P(m^{\prime}\Delta_f)P^*(n^{\prime}\Delta_f)\left(\nu^*_{x,k^{\prime}}\nu_{x,m^{\prime}}\nu^*_{x,n^{\prime}}+\nu^*_{y,k^{\prime}}\nu_{y,m^{\prime}}\nu^*_{x,n^{\prime}}\right)\\
&\eta^*_{k^{\prime},m^{\prime},n^{\prime}}\biggl\}
\label{eq:PSD_x1}
\end{split}
\\
\begin{split}
&=\left(\frac{8}{9}\right)^2\gamma^2\Delta_f^3\sum_{i=-\infty}^\infty\delta(f-i\Delta_f)\mathbb{E}\biggl\{\sum_{\substack{(k,m,n) \in \mathcal{S}_i \\ (k^{\prime},m^{\prime},n^{\prime})\in \mathcal{S}_{i}}}\mathcal{P}_{k,m,n,k^{\prime},m^{\prime},n^{\prime}}\left(\nu_{x,k}\nu^*_{x,m}\nu_{x,n}\nu^*_{x,k^{\prime}}\nu_{x,m^{\prime}}\nu^*_{x,n^{\prime}}\right.\\
&+\nu_{x,k}\nu^*_{x,m}\nu_{x,n}\nu^*_{y,k^{\prime}}\nu_{y,m^{\prime}}\nu^*_{x,n^{\prime}}+\nu_{y,k}\nu^*_{y,m}\nu_{x,n}\nu^*_{x,k^{\prime}}\nu_{x,m^{\prime}}\nu^*_{x,n^{\prime}}\\
&+\left.\nu_{y,k}\nu^*_{y,m}\nu_{x,n}\nu^*_{y,k^{\prime}}\nu_{y,m^{\prime}}\nu^*_{x,n^{\prime}}\right)\eta_{k,m,n}\eta^*_{k^{\prime},m^{\prime},n^{\prime}}\biggl\}
\label{eq:PSD_x2} 
\end{split}
\end{align}
\end{subequations}
where we have defined 
\begin{equation}
\mathcal{P}_{k,m,n,k^{\prime},m^{\prime},n^{\prime}}\triangleq P(k\Delta_f)P^*(m\Delta_f)P(n\Delta_f)P^*(k^{\prime}\Delta_f)P(m^{\prime}\Delta_f)P^*(n^{\prime}\Delta_f).
\label{eq:P_def}  
\end{equation}

The following proposition can be used to make \eqref{eq:PSD_x2} more compact. In particular, we will  group the two inner correlation terms in \eqref{eq:PSD_x2} ( $\nu_{x,k}\nu^*_{x,m}\nu_{x,n}\nu^*_{y,k^{\prime}}\nu_{y,m^{\prime}}\nu^*_{x,n^{\prime}}$ and $\nu_{y,k}\nu^*_{y,m}\nu_{x,n}\nu^*_{x,k^{\prime}}\nu_{x,m^{\prime}}\nu^*_{x,n^{\prime}}$) using this proposition.
\begin{proposition}
For $\mathcal{P}_{k,m,n,k^{\prime},m^{\prime},n^{\prime}}$ in \eqref{eq:P_def}, we have
\begin{align}\label{eq:PSDconj}
&\sum_{\substack{(k,m,n) \in \mathcal{S}_i\\ (k^{\prime},m^{\prime},n^{\prime}) \in \mathcal{S}_i}}\mathcal{P}_{k,m,n,k^{\prime},m^{\prime},n^{\prime}}\nu_{x,k}\nu^*_{x,m}\nu_{x,n}\nu^*_{y,k^{\prime}}\nu_{y,m^{\prime}}\nu^*_{x,n^{\prime}}\eta_{k,m,n}\eta^*_{k^{\prime},m^{\prime},n^{\prime}}\nonumber\\
&=\Biggl(\sum_{\substack{(k,m,n) \in \mathcal{S}_i \\ (k^{\prime},m^{\prime},n^{\prime}) \in \mathcal{S}_i}}\mathcal{P}_{k,m,n,k^{\prime},m^{\prime},n^{\prime}}\nu_{y,k}\nu^*_{y,m}\nu_{x,n}\nu^*_{x,k^{\prime}}\nu_{x,m^{\prime}}\nu^*_{x,n^{\prime}}\eta_{k,m,n}\eta^*_{k^{\prime},m^{\prime},n^{\prime}}\Biggl)^*.
\end{align}
\label{prop:conj}
\end{proposition}
\begin{proof}
See Appendix \ref{app:B}.
\end{proof}
Using \eqref{eq:PSDconj}, \eqref{eq:PSD_x2} can be written as
\begin{align}
\begin{split}
S_x(f,N_s,L_s)&=\left(\frac{8}{9}\right)^2\gamma^2\Delta_f^3\sum^\infty_{i=-\infty}\delta(f-i\Delta_f)\mathbb{E}\Biggl\{\sum_{\substack{(k,m,n) \in \mathcal{S}_i \\ (k^{\prime},m^{\prime},n^{\prime}) \in \mathcal{S}_i}}               \mathcal{P}_{k,m,n,k^{\prime},m^{\prime},n^{\prime}}\\
&\cdot\left(\nu_{x,k}\nu^*_{x,m}\nu_{x,n}\nu^*_{x,k^{\prime}}\nu_{x,m^{\prime}}\nu^*_{x,n^{\prime}}+\nu_{y,k}\nu^*_{y,m}\nu_{x,n}\nu^*_{y,k^{\prime}}\nu_{y,m^{\prime}}\nu^*_{x,n^{\prime}}\right)\eta_{k,m,n}\eta^*_{k^{\prime},m^{\prime},n^{\prime}}\\
&+2\Re\{\mathcal{P}_{k,m,n,k^{\prime},m^{\prime},n^{\prime}}\nu_{x,k}\nu^*_{x,m}\nu_{x,n}\nu^*_{y,k^{\prime}}\nu_{y,m^{\prime}}\nu^*_{x,n^{\prime}}\eta_{k,n,m}\eta^*_{k^{\prime},n^{\prime},m^{\prime}}\}\Biggr\}\\
&=\left(\frac{8}{9}\right)^2\gamma^2\Delta_f^3\sum^\infty_{i=-\infty}\delta(f-i\Delta_f)\sum_{\substack{(k,m,n) \in \mathcal{S}_i \\ (k^{\prime},m^{\prime},n^{\prime}) \in \mathcal{S}_i}}               \mathcal{P}_{k,m,n,k^{\prime},m^{\prime},n^{\prime}}\left(\mathbb{E}\left\{\nu_{x,k}\nu^*_{x,m}\nu_{x,n}\nu^*_{x,k^{\prime}}\nu_{x,m^{\prime}}\nu^*_{x,n^{\prime}}\right\}\right.\\
&\left.+\mathbb{E}\left\{\nu_{y,k}\nu^*_{y,m}\nu_{x,n}\nu^*_{y,k^{\prime}}\nu_{y,m^{\prime}}\nu^*_{x,n^{\prime}}\right\}\right)\eta_{k,m,n}\eta^*_{k^{\prime},m^{\prime},n^{\prime}}\\
&+2\Re\{\mathcal{P}_{k,m,n,k^{\prime},m^{\prime},n^{\prime}}\mathbb{E}\left\{\nu_{x,k}\nu^*_{x,m}\nu_{x,n}\nu^*_{y,k^{\prime}}\nu_{y,m^{\prime}}\nu^*_{x,n^{\prime}}\right\}\eta_{k,n,m}\eta^*_{k^{\prime},n^{\prime},m^{\prime}}\}.
\end{split}
\label{eq:PSD_x_2}   
\end{align}

According to \eqref{eq:PSD_x_2}, the calculation of the PSD of the NLI reduces to the computation of a four-dimensional summation (per frequency component $i\Delta_f$) of three sixth-order correlations of the sequence of random variables $\nu_{x/y,n}, \; n=0,1,\dots,W-1$. The $y$-component  $S_y(f,N_s,L_s)$ of the PSD can be calculated once $S_x(f,N_s,L_s)$ is obtained, by simply swapping the polarisation label $x\rightarrow y$ and $y\rightarrow x$. This is due to the invariance of the Manakov equation in \eqref{eq:Manakov} to such a transformation.
\section{Classification of the modulation-dependent contributions in the 6th-order frequency-domain correlation}\label{sec:subset_classification}
In this section, we will break down the frequency-domain sixth-order correlation terms in \eqref{eq:PSD_x_2} to highlight different contributions in terms of 4D modulation-dependent cross-polarisation correlations.

\subsection{Expansion in terms of the stochastic moments of the transmitted modulation format}\label{sec:6th-order_freq_to_time}
To relate the PSD in \eqref{eq:PSD_x_2} to the statistical properties of the transmitted modulation format, we replace   
\eqref{eq:symbol_fft} into \eqref{eq:PSD_x_2}, obtaining 
\begin{align}
\begin{split}
S_x(f,N_s,L_s)&=\left(\frac{8}{9}\right)^2\gamma^2\Delta_f^3\sum_{i=-\infty}^\infty\delta(f-i\Delta_f)\sum_{\substack{(k,m,n) \in \mathcal{S}_i \\ (k^{\prime},m^{\prime},n^{\prime}) \in \mathcal{S}_i}} \bigl[ \mathcal{P}_{k,m,n,k^{\prime},m^{\prime},n^{\prime}}\eta_{k,m,n}\eta^*_{k^{\prime},m^{\prime},n^{\prime}}\\
& \cdot \sum_{\boldsymbol{i}\in\{0,1,\dots,W-1\}^6}\mathsf{S}_{\boldsymbol{i}}(k,m,n,k^{\prime},m^{\prime},n^{\prime})+2\Re\{\mathcal{P}_{k,m,n,k^{\prime},m^{\prime},n^{\prime}}\eta_{k,m,n}\eta^*_{k^{\prime}, m^{\prime}, n^{\prime}}\\
& \cdot \sum_{\boldsymbol{i}\in\{0,1,\dots,W-1\}^6}\mathsf{T}_{\boldsymbol{i}}(k,m,n,k^{\prime},m^{\prime},n^{\prime})\}\bigr],
\end{split}
\label{eq:6D_sum}
\end{align}
where $\boldsymbol{i}\triangleq (i_1,i_2,\dots,i_6)$,
\begin{align}
\begin{split}
\mathsf{S}_{\boldsymbol{i}}(k,m,n,k^{\prime},m^{\prime},n^{\prime})&\triangleq\Delta_f^3\left[\mathbb{E}\left\{a_{x,i_1}a_{x,i_2}^*a_{x,i_3}a_{x,i_4}^*a_{x,i_5}a_{x,i_6}^*\right\}+\mathbb{E}\left\{a_{y,i_1}a_{y,i_2}^*a_{x,i_3}a_{y,i_4}^*a_{y,i_5}a_{x,i_6}^*\right\}\right]\\ 
&\cdot e^{-j\frac{2\pi}{W}(ki_1-mi_2+ni_3-k^{\prime}i_4+m^{\prime}i_5-n^{\prime}i_6)},
\end{split}
\label{eq:S_sum}
\end{align}
and
\begin{align}
\begin{split}
&\mathsf{T}_{\boldsymbol{i}}(k,m,n,k^{\prime},m^{\prime},n^{\prime})\triangleq\Delta_f^3\mathbb{E}\left\{a_{x,i_1}a_{x,i_2}^*a_{x,i_3}a_{y,i_4}^*a_{y,i_5}a_{x,i_6}^*\right\}e^{-j\frac{2\pi}{W}(ki_1-mi_2+ni_3-k^{\prime}i_4+m^{\prime}i_5-n^{\prime}i_6)}.
\end{split}
\label{eq:T_sum}
\end{align}

The terms $\mathsf{S}_{\boldsymbol{i}}(k,m,n,k^{\prime},m^{\prime},n^{\prime})$ and $\mathsf{T}_{\boldsymbol{i}}(k,m,n,k^{\prime},m^{\prime},n^{\prime})$ give rise to several correlations among the transmitted symbols $a_{x,i}$ and $a_{y,j}$ at different time-slots $i, j$, each weighted by a complex exponential. As discussed in Sec.~\ref{sec:model assumptions}, in this work we operate under the assumption that the sequence of vector RVs $\boldsymbol{a}_{i}$ for $i=0,1,\dots,W-1$ are independent, identically distributed (i.i.d.), and with $\mathbb{E}\{\boldsymbol{a}_{i}\}=\mathbb{E}\{\boldsymbol{a}\}=\boldsymbol{0}$. As shown in the following example, this assumption allows us to discard the  $\mathsf{S}_{\boldsymbol{i}}(k,m,n,k^{\prime},m^{\prime},n^{\prime})$ and $\mathsf{T}_{\boldsymbol{i}}(k,m,n,k^{\prime},m^{\prime},n^{\prime})$ terms which are identically zero for some values of $\boldsymbol{i}$. Moreover, as it will be shown in Example \ref{ex:degenerate_elements}, for all other values of $\boldsymbol{i}$, $\mathsf{S}_{\boldsymbol{i}}(k,m,n,k^{\prime},m^{\prime},n^{\prime})$ and $\mathsf{T}_{\boldsymbol{i}}(k,m,n,k^{\prime},m^{\prime},n^{\prime})$ can be expressed as a product of high-order statistical moments of the RVs $a_x$ and $a_y$ which enables a more compact expression for \eqref{eq:6D_sum}.

\begin{example_a}\label{ex:1st_order_moment_zero_contr}
Under the i.i.d. assumption for the sequence of vector RVs $\boldsymbol{a}_{i}$, $i=0,1,...W-1$ made in this work, in any of the cases where
\begin{equation}
i_{\kappa_1}\neq i_{\kappa_2}=i_{\kappa_3}=\dots=i_{\kappa_6} \qquad \text{for} \quad \kappa_1, \kappa_2, \dots, \kappa_6=1,2,\dots,6; \;\; \kappa_1\neq\kappa_2\neq\dots\neq\kappa_6,
\label{eq:class_def}
\end{equation}
any of the sixth-order correlations in \eqref{eq:S_sum} and \eqref{eq:T_sum} degenerate into a product between a first-order moment and a fifth-order correlation. Such a product is equal to zero under our assumption $\mathbb{E}\{a_{x,i}\}=\mathbb{E}\{a_{x}\}=0$. 
For example, for $i_1\neq i_2=i_3=\dots=i_6$, we have 
\begin{align*}
\begin{split}
&\mathbb{E}\{a_{x,i_1}a_{x,i_2}^*a_{x,i_3}a_{x,i_4}^*a_{x,i_5}a_{x,i_6}^*\}=\mathbb{E}\{a_{x,i_1}\}\mathbb{E}\{|a_{x,i_2}|^4 a_{x,i_2}^*\}=\mathbb{E}\{a_{x}\}\mathbb{E}\{|a_{x}|^4 a_{x}^*\}=0.
\end{split}
\end{align*}
From this follows that for the set of elements defined by \eqref{eq:class_def}, also $\mathsf{S}_i(k,m,n,k^{\prime},m^{\prime},n^{\prime})=0$, and $\mathsf{T}_i(k,m,n,k^{\prime},m^{\prime},n^{\prime})=0$.
\end{example_a}
The terms in the class in Example \ref{ex:1st_order_moment_zero_contr} are identically zero regardless of the values taken by $k,m,n,k^{\prime},m^{\prime},n^{\prime}$. However, as it will be shown in Sec.~\ref{sec:evaluation}, many nonzero sixth-order correlations in the sum \eqref{eq:6D_sum} cancel each other for a specific subset of values $k,m,n,k^{\prime},m^{\prime},n^{\prime}$ because of the complex exponential weighing. 

\begin{example_a}\label{ex:degenerate_elements}
Under the i.i.d. assumption for the sequence of vector RVs $\boldsymbol{a}_{i}$, $i=0,1,...W-1$ made in this work, we have that for all elements in the subset $\{\boldsymbol{i}\in\{0,1,\dots,W-1\}^6, i_{1}=i_{2},i_{3}=i_{4}=i_5=i_{6}, i_1\neq i_3\}$
\begin{align*}
\begin{split}
\mathsf{S}_{\boldsymbol{i}}(k,m,n,k^{\prime},m^{\prime},n^{\prime})&=\Delta_f^3[\mathbb{E}\{|a_{x,i_1}|^2\}\mathbb{E}\{|a_{x,i_3}|^4\}+\mathbb{E}\{|a_{y,i_1}|^2\}\mathbb{E}\{|a_{x,i_3}|^2 |a_{y,i_3}|^2\}]\\
&\cdot e^{-j\frac{2\pi}{W}((k-m)i_1+(n-k^{\prime}+m^{\prime}-n^{\prime})i_3)}\\
&=\Delta_f^3[\mathbb{E}\{|a_x|^2\}\mathbb{E}\{|a_x|^4\}+\mathbb{E}\{|a_y|^2\}\mathbb{E}\{|a_x|^2 |a_y|^2\}]e^{-j\frac{2\pi}{W}((k-m)i_1+(n-k^{\prime}+m^{\prime}-n^{\prime})i_3)}
\end{split}
\end{align*}
It can be noted that: i) the sixth-order correlation degenerates into products of marginal (high-order) moments of $a_x$, $a_y$ and into the cross-polarisation moment $\mathbb{E}\{|a_x|^2 |a_y|^2\}$; ii) all elements within the set in this example contribute to the inner summation in \eqref{eq:6D_sum} with the same set of moments, cross-polarisation correlations and products thereof (i.e.,~$\mathbb{E}\{|a_x|^2 \}, \mathbb{E}\{|a_x|^4 \}, \mathbb{E}\{|a_y|^2 \}, \mathbb{E}\{|a_x|^2 |a_y|^2\}$).     
\end{example_a}

In the remainder of this Section, we first partition the six-dimensional space $\boldsymbol{i}\in \{0,1,\dots,W-1\}^6$ and list all sets corresponding to nonzero elements of $\mathsf{S}_{\boldsymbol{i}}(k,m,n,k^{\prime},m^{\prime},n^{\prime})$ and $\mathsf{T}_{\boldsymbol{i}}(k,m,n,k^{\prime},m^{\prime},n^{\prime})$. As shown in Example \ref{ex:degenerate_elements}, this will help highlighting the contribution of a specific set in terms of high-order moments of the transmitted symbols $\boldsymbol{a}$ in \eqref{eq:6D_sum}. Then we proceed to list all such contributions.

\subsection{Set partitioning}\label{subsec:subset_part}
The six-dimensional space $\boldsymbol{i} \in \{0,1,\dots, W-1\}^6$ can be partitioned in different \emph{subsets} each one uniquely defined by a partition on the set of indices ($i_1$,\,$i_2$\,,$i_3$\,,$i_4$\,,$i_5$\,,$i_6$). Each partition defines its corresponding subset in $\{0,1,\dots, W-1\}^6$ as follows: for each index partition, the indices belonging to the same subset take all the same value, whilst the indices belonging to different subsets have distinct values. This is schematically illustrated in Fig.~\ref{fig:partition}. For example, the subset of $\{0,1,\dots, W-1\}^6$ labelled by the index partition $\{(i_1,i_2),(i_3,i_4),(i_5,i_6)\}$ is defined as $\{\boldsymbol{i}\in\{0,1,\dots,W-1\}^6 : i_1=i_2,i_3=i_4,i_5=i_6, i_1\neq i_3 \neq i_5 \}$. This subset is shown in Fig. ~\ref{fig:partition} as part of $\mathcal{L}_1$.     

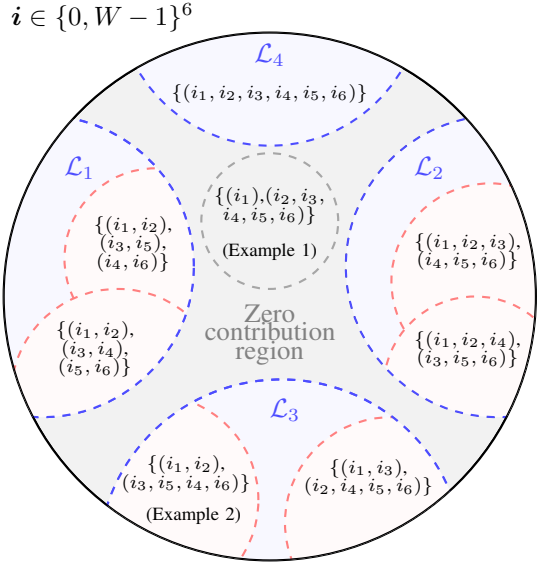
\begin{figure}[tbp]
	\centering
	 	\begin{tikzpicture}
	\tikzset{mynode/.style={align=center,execute at begin node=\setlength{\baselineskip}{0.1em}}} 
	
	\node at (-2.2,3.7) {$\boldsymbol{i}\in \{0, W-1\}^6$};
	\draw[thick,fill = gray!10] (0,0) circle (3.5); 
	\clip(0,0) circle (3.5);

	\filldraw[fill=blue!3, dashed]  (3,0.4) circle (2);
	\begin{scope}
	\clip (3,0.4) circle (2);
	\filldraw[thick,color=red!50, fill=red!2,dashed] (2.9,0.2) circle (1.3);
	\filldraw[thick,color=red!50, fill=red!2, dashed] (2.7,-1.2) circle (1.2);
	\end{scope}
	\node[blue!70] at (2.1,1.7) {$\mathcal{L}_2$};
	\begin{scriptsize}
	\node[mynode,align=center] at (2.6,.6) {$\{ (i_1,i_2,i_3)$,\\ $(i_4,i_5,i_6) \}$};
	\node [mynode,align=center] at (2.6,-.7) {$\{(i_1,i_2,i_4)$,\\ $(i_3,i_5,i_6)\}$};
	\end{scriptsize}
	\draw[thick,color=blue!70,dashed] (3,0.4) circle (2);
	
	\filldraw[thick,color=blue!70, fill=blue!3, dashed] (0,4) circle (2);
	\node[blue!70] at (0,3.2) {$\mathcal{L}_4$};
	\begin{scriptsize}
	\node [align=center] at (0,2.7) {$\{(i_1,i_2,i_3,i_4,i_5,i_6)\}$};
	\end{scriptsize}
	
	\filldraw[thick,color=blue!70, fill=blue!3, dashed] (-3,0.4) circle (2);
	\node[blue!70] at (-2.5,1.7) {$\mathcal{L}_1$};
	\begin{scope}
	\clip (-3,0.4) circle (2);
	\filldraw[thick,color=red!50, fill=red!2,dashed] (-1.5,0.5) circle (1.2);
	\filldraw[thick,color=red!50, fill=red!2, dashed] (-2,-1.3) circle (1.4);
	\end{scope}
	\begin{scriptsize}
	\node [mynode,align=center] at (-1.8,0.7) {$\{(i_1,i_2)$,\\ $(i_3,i_5)$,\\ $(i_4,i_6)\}$};
	\node [mynode,align=center] at (-2.3,-0.7) {$\{(i_1,i_2)$,\\ $(i_3,i_4)$,\\ $(i_5,i_6)\}$};
	\end{scriptsize}
	\draw[thick,color=blue!70,dashed] (-3,0.4) circle (2);
	
	\filldraw[thick,color=blue!70, fill=blue!3, dashed] (0.15,-3.5) circle (2.4);
	\node[blue!70] at (0.2,-1.5) {$\mathcal{L}_3$};
	\begin{scope}
	\clip (0.15,-3.5) circle (2.4);
	\filldraw[thick,color=red!50, fill=red!2,dashed] (-1.55,-2.7) circle (1.4);
	\filldraw[thick,color=red!50, fill=red!2, dashed] (1.6,-3) circle (1.4);
	\end{scope}
	\begin{scriptsize}
	\node [mynode,align=center] at (-1.1,-2.35) {$\{(i_1,i_2)$,\\ $(i_3,i_5,i_4,i_6)\}$};
	\node [mynode,align=center,font=\scriptsize] at (-1,-2.9) {(Example 2)};
	
	\node [mynode,align=center] at (1.3,-2.4) {$\{(i_1,i_3)$,\\ $(i_2,i_4,i_5,i_6)\}$};
	\end{scriptsize}
	\draw[thick,color=blue!70,dashed] (0.15,-3.5) circle (2.4);
	
	\draw[thick,color=gray!70,dashed] (0,1) circle (.9);
	
	\node [mynode,align=center,font=\scriptsize] at (0,1.2) {$\{(i_1)$,$(i_2,i_3,$\\
	$i_4,i_5,i_6)\}$};
	\node [mynode,align=center,font=\scriptsize] at (0,.6) {(Example 1)};
	
	\draw[] (0,0) circle (3.5);
	\node [mynode,align=center,color=gray!99] at (0,-.5) {Zero \\ contribution\\ region};

	\draw[thick] (0,0) circle (3.5); 

	\end{tikzpicture}
	\caption{Venn diagram of the partition on the 6D space $\boldsymbol{i}=(i_1,i_2,i_3,i_4,i_5,i_6)$ discussed in Sec.~\ref{sec:subset_classification}.}\label{fig:partition}
\end{figure}

In Fig.~\ref{fig:partition}, the \emph{families of subsets} of $\{0,1,\dots, W-1\}^6$ labelled $\mathcal{L}_{i}$, $i=1,2,\dots,4$ are also highlighted. These families are characterised by subsets sharing the same cardinality of elements associated to their corresponding index partition. For example, in $\mathcal{L}_{1}$, all index partitions are characterised by 3 subsets each one containing 2 indices. As shown in Example \ref{ex:degenerate_elements}, this way of partitioning the set $\{0,1,\dots, W-1\}^6$ is useful as it separates out the different contributions of \eqref{eq:6D_sum} based on the high-order moments of $\boldsymbol{a}$ as it is highlighted in region $\mathcal{L}_3$ of Fig.~\ref{fig:partition}. 

Since we have 6 different indices the number of subsets in a partition can vary from 1 to 6. Each of these subsets can contain a number of elements also ranging from 1 to 6. However, the subsets of $\{0,1,\dots, W-1\}^6$ where the corresponding index partition has one or more index subsets with only one element bring no contribution to \eqref{eq:6D_sum}, and thus can be discarded. This is illustrated in Example \ref{ex:1st_order_moment_zero_contr}. The above class of index partitions then forms a \emph{zero contribution} region in, as shown in Fig.~\ref{fig:partition}. Such a region also includes all subsets where the corresponding index partitions contain 4 or more index subsets, as at least one of these subsets will have to contain only one element.

As shown in Fig.~\ref{fig:partition}, by removing the \emph{zero contribution} region from $\{0,1,\dots, W-1\}^6$, only 4 different families of subsets are left:   

\begin{enumerate}[label=(\roman*)]

\item $\mathcal{L}_{1}=\{\boldsymbol{i}\in \{0,1,\dots,W-1\}^6 : i_{\kappa_1}=i_{\kappa_2}; \; i_{\kappa_3}=i_{\kappa_4}; \; i_{\kappa_5}=i_{\kappa_6};\; \kappa_1,\kappa_2,\dots\kappa_6=1,2,\dots,6; \;  \kappa_1\neq \kappa_2 \neq \kappa_3\neq\kappa_4\neq\kappa_5\neq\kappa_6\}$. This set contains all sets of elements where the indices $i_1,i_2,\dots,i_6$ can be grouped in 3 pairs. The indices take up the same value within each pair but different values across different pairs. It can be found that this set can be partitioned in 15 different subsets $\mathcal{C}_{1}^{(i)}, i=1,2,\dots,15$ representing all possible distinct ways of pairing the $i_k$ indices for $k=1,2,\dots6$. These sets are listed in Table \ref{tab:L1}, where each column shows a subgroup of indices taking the same value.

\begin{table}[tbp]
\begin{center}
\caption{List of all subsets in $\mathcal{L}_1$. For each subset the index subgroups identify the corresponding pairs of indices assuming the same value.}\label{tab:L1}
\begin{tabular}{c|c|c|c}
\hline \hline
\backslashbox{\footnotesize{\textbf{Subset}}\\\footnotesize{\textbf{label}}}{\footnotesize{\textbf{Index}}\\\footnotesize{\textbf{subgroup}}}
&\makebox[3em]{\textbf{1}}&\makebox[3em]{\textbf{2}}&\makebox[3em]{\textbf{3}}\\
\hline\hline
$\mathcal{C}^{(1)}_{1}$ & $i_1,i_2$ & $i_3,i_4$ & $i_5,i_6$ \\
\hline
$\mathcal{C}^{(2)}_{1}$ & $i_1,i_2$ & $i_3,i_5$ & $i_4,i_6$ \\
\hline
$\mathcal{C}^{(3)}_{1}$ & $i_1,i_2$ & $i_3,i_6$ & $i_4,i_5$ \\
\hline
$\mathcal{C}^{(4)}_{1}$ & $i_1,i_3$ & $i_2,i_4$ & $i_5,i_6$ \\
\hline
$\mathcal{C}^{(5)}_{1}$ & $i_1,i_3$ & $i_2,i_5$ & $i_4,i_6$ \\
\hline
$\mathcal{C}^{(6)}_{1}$ & $i_1,i_3$ & $i_2,i_6$ & $i_4,i_5$ \\
\hline
$\mathcal{C}^{(7)}_{1}$ & $i_1,i_4$ & $i_2,i_3$ & $i_5,i_6$ \\
\hline
$\mathcal{C}^{(8)}_{1}$ & $i_1,i_4$ & $i_2,i_5$ & $i_3,i_6$ \\
\hline
$\mathcal{C}^{(9)}_{1}$ & $i_1,i_4$ & $i_2,i_6$ & $i_3,i_5$ \\
\hline
$\mathcal{C}^{(10)}_{1}$ & $i_1,i_5$ & $i_2,i_3$ & $i_4,i_6$ \\
\hline
$\mathcal{C}^{(11)}_{1}$ & $i_1,i_5$ & $i_2,i_4$ & $i_3,i_6$ \\
\hline
$\mathcal{C}^{(12)}_{1}$ & $i_1,i_5$ & $i_2,i_6$ & $i_3,i_4$ \\
\hline
$\mathcal{C}^{(13)}_{1}$ & $i_1,i_6$ & $i_2,i_3$ & $i_4,i_5$ \\
\hline
$\mathcal{C}^{(14)}_{1}$ & $i_1,i_6$ & $i_2,i_4$ & $i_3,i_5$ \\
\hline
$\mathcal{C}^{(15)}_{1}$ & $i_1,i_6$ & $i_2,i_5$ & $i_3,i_4$\\ 
\hline \hline
\end{tabular}
\end{center}
\end{table}

\item $\mathcal{L}_{2}:\{\boldsymbol{i}\in \{0,1,\dots,W-1\}^6,   i_{\kappa_1}=i_{\kappa_2}=i_{\kappa_3}, i_{\kappa_4}=i_{\kappa_5}=i_{\kappa_6}; \quad \kappa_1,\kappa_2,\dots,\kappa_6=1,2,\dots,6; \; \kappa_1\neq \kappa_2 \neq \kappa_3\neq\kappa_4\neq\kappa_5\neq\kappa_6\}$ which can be broken down in 10 subsets $\mathcal{C}_{2}^{(i)}, \; i=1,2,\dots,10$ listed in Table \ref{tab:L2}. Each index subgroup identifies a triplet of indices assuming the same value.

{\small
\begin{table}[tbp]
\begin{center}
\caption{List of all subsets in $\mathcal{L}_2$. For each subset, the index subgroups identify the corresponding triplets of indices assuming the same value.}\label{tab:L2}
\begin{tabular}{c|c|c}
\hline \hline
\backslashbox{\footnotesize{\textbf{Subset}}\\\footnotesize{\textbf{label}}}{\footnotesize{\textbf{Index}}\\\footnotesize{\textbf{subgroup}}}
&\makebox[3em]{\textbf{1}}&\makebox[3em]{\textbf{2}}\\\hline\hline

$\mathcal{C}^{(1)}_{2}$ & $i_1,i_2, i_3$ & $i_4,i_5,i_6$ \\
\hline
$\mathcal{C}^{(2)}_{2}$ & $i_1,i_2,i_4$, & $i_3,i_5,i_6$ \\
\hline
$\mathcal{C}^{(3)}_{2}$ & $i_1,i_2,i_5$ & $i_3,i_4,i_6$ \\
\hline
$\mathcal{C}^{(4)}_{2}$ & $i_1,i_2,i_6$ & $i_3,i_4,i_5$ \\
\hline
$\mathcal{C}^{(5)}_{2}$ & $i_1,i_3,i_4$ & $i_2,i_5,i_6$ \\
\hline
$\mathcal{C}^{(6)}_{2}$ & $i_1,i_3,i_5$ & $i_2,i_4,i_6$ \\
\hline
$\mathcal{C}^{(7)}_{2}$ & $i_1,i_3,i_6$ & $i_2,i_4,i_5$ \\
\hline
$\mathcal{C}^{(8)}_{2}$ & $i_1,i_4,i_5$ & $i_2,i_3,i_6$ \\
\hline
$\mathcal{C}^{(9)}_{2}$ & $i_1,i_4,i_6$ & $i_2,i_3,i_5$ \\
\hline
$\mathcal{C}^{(10)}_{2}$ & $i_1,i_5,i_6$ & $i_2,i_3,i_4$ \\
\hline \hline
\end{tabular}
\end{center}
\end{table}
}

\item $\mathcal{L}_{3}=\{\boldsymbol{i}\in \{0,1,\dots,W-1\}^6 : i_{\kappa_1}=i_{\kappa_2}, i_{\kappa_3}=i_{\kappa_4}=i_{\kappa_5}=i_{\kappa_6}; \;\; \kappa_1,\kappa_2,\dots,\kappa_6=1,2,\dots,6, \kappa_1\neq \kappa_2 \neq \kappa_3\neq\kappa_4\neq\kappa_5\neq\kappa_6\}$  which can be partitioned in 15 subsets $\mathcal{C}^{(i)}_{3}, \; i=1,2,\dots,15$ listed in Table \ref{tab:L3}. Each of the two index subgroups identifies the pair and the quadruple of indices assuming the same value.

{\small
\begin{table}[tbp]
\begin{center}
\caption{List of all subsets in $\mathcal{L}_3$. For each subset, the index subgroups identify the corresponding pair and quadruple of indices assuming the same value. The highlighted row corresponds to Example \ref{ex:degenerate_elements}.}\label{tab:L3}
\begin{tabular}{c|c|c}
\hline \hline
\backslashbox{\footnotesize{\textbf{Subset}}\\\footnotesize{\textbf{label}}}{\footnotesize{\textbf{Index}}\\\footnotesize{\textbf{subgroup}}}
&\makebox[3em]{\textbf{1}}&\makebox[3em]{\textbf{2}}\\\hline\hline

\rowcolor{gray!15} $\mathcal{C}^{(1)}_{3}$ & $i_1,i_2$ & $i_3,i_4,i_5,i_6$ \\
\hline
$\mathcal{C}^{(2)}_{3}$ & $i_1,i_3$, & $i_2,i_4,i_5,i_6$ \\
\hline
$\mathcal{C}^{(3)}_{3}$ & $i_1,i_4$ & $i_2,i_3,i_5,i_6$ \\
\hline
$\mathcal{C}^{(4)}_{3}$ & $i_1,i_5$ & $i_2,i_3,i_4,i_6$ \\
\hline
$\mathcal{C}^{(5)}_{3}$ & $i_1,i_6$ & $i_2,i_3,i_4,i_5$ \\
\hline
$\mathcal{C}^{(6)}_{3}$ & $i_2,i_3$ & $i_1,i_4,i_5,i_6$ \\
\hline
$\mathcal{C}^{(7)}_{3}$ & $i_2,i_4$ & $i_1,i_3,i_5,i_6$ \\
\hline
$\mathcal{C}^{(8)}_{3}$ & $i_2,i_5$ & $i_1,i_3,i_4,i_6$ \\
\hline
$\mathcal{C}^{(9)}_{3}$ & $i_2,i_6$ & $i_1,i_3,i_4,i_5$ \\
\hline
$\mathcal{C}^{(10)}_{3}$ & $i_3,i_4$ & $i_1,i_2,i_5,i_6$ \\
\hline
$\mathcal{C}^{(11)}_{3}$ & $i_3,i_5$ & $i_1,i_2,i_4,i_6$ \\
\hline
$\mathcal{C}^{(12)}_{3}$ & $i_3,i_6$ & $i_1,i_2,i_4,i_5$ \\
\hline
$\mathcal{C}^{(13)}_{3}$ & $i_4,i_5$ & $i_1,i_2,i_3,i_6$ \\
\hline
$\mathcal{C}^{(14)}_{3}$ & $i_4,i_6$ & $i_1,i_2,i_3,i_5$ \\
\hline
$\mathcal{C}^{(15)}_{3}$ & $i_5,i_6$ & $i_1,i_2,i_3,i_4$ \\
\hline \hline
\end{tabular}
\end{center}
\end{table}
}

\item $\mathcal{L}_{4}: \{\boldsymbol{i}\in \{0,1,\dots,W-1\}^6 : i_1=i_2=i_3=i_4=i_5=i_6 \}$. 

\end{enumerate}

\section{Evaluation of the $\mathcal{L}$-based contributions}\label{sec:evaluation}
In this section, we provide three examples for the computation of the contributions of a generic element in $\mathcal{L}_1$, $\mathcal{L}_2$ and $\mathcal{L}_3$. The full list of contributions in these three sets and the contributions in $\mathcal{L}_4$ are given in Secs.~\ref{subsec:L1contr}--\ref{subsec:L4contr}.
We label each contribution as $\mathsf{M}_{g}^{(h)}(k,m,n,k^{\prime},m^{\prime},n^{\prime})$ and $\mathsf{N}_{g}^{(h)}(k,m,n,k^{\prime},m^{\prime},n^{\prime})$, where
\begin{align}
\begin{split}
\mathsf{M}_{g}^{(h)}(k,m,n,k^{\prime},m^{\prime},n^{\prime})&\triangleq \sum_{\boldsymbol{i}\in\mathcal{C}_{g}^{(h)}}\mathsf{S}_{\boldsymbol{i}}(k,m,n,k^{\prime},m^{\prime},n^{\prime}), \\
\mathsf{N}_{g}^{(h)}(k,m,n,k^{\prime},m^{\prime},n^{\prime})&\triangleq \sum_{\boldsymbol{i}\in\mathcal{C}_{g}^{(h)}}\mathsf{T}_{\boldsymbol{i}}(k,m,n,k^{\prime},m^{\prime},n^{\prime}),
\end{split}
\label{eq:MN_contr}
\end{align}
and the subsets $\mathcal{C}_{g}^{(h)}$ are taken form Tables \ref{tab:L1}, \ref{tab:L2} and \ref{tab:L3}.


\begin{example_a}[Contributions in $\mathcal{L}_1$]\label{ex:L1contr}
$\mathsf{M}_{1}^{(1)}$, i.e., one of the 2 contributions for the set $\mathcal{C}_1^{(1)}=\{\boldsymbol{i}\in\{0,1,\dots,W-1\}^6 : i_1=i_2,i_3=i_4,i_5=i_6, i_1\neq i_3, i_1\neq i_5, i_3\neq i_5\}$) is given by

\begin{align}
\begin{split}
\mathsf{M}_{1}^{(1)}&\triangleq\sum_{\boldsymbol{i}\in\mathcal{C}_1^{(1)}} \mathsf{S}_{\boldsymbol{i}}(k,m,n,k^{\prime},m^{\prime},n^{\prime})\\
&=\Delta_f^3\left[\mathbb{E}^3\left\{|a_{x}|^2\right\}+\mathbb{E}\left\{|a_{y}|^2\right\} |\mathbb{E}\left\{a_{x}a^*_{y}\right\}|^2\right]\sum_{i_1=0}^{W-1} e^{-j\frac{2\pi}{W}(k-m) i_1}\sum_{{i_3\neq i_1}}e^{-j\frac{2\pi}{W}(n-k^{\prime})i_3}\\
&\cdot\sum_{\substack{i_5\neq i_1,\\i_5\neq i_3}}e^{-j\frac{2\pi}{W}(m^{\prime}-n^{\prime})i_5}. 
\end{split}\label{eq:M11_example}
\end{align}

Since  
\begin{equation}\label{eq:compl_exp_sum}
\sum_{k=0}^{W-1} e^{jnk\frac{2\pi}{W}}=
\begin{cases}
W, & \text{for} \;\; n=pW,\; p\in\mathbb{Z} \\
0, & \text{elsewhere}
\end{cases},
\end{equation}
we can compute \eqref{eq:M11_example} using the following approach: 
\begin{enumerate}
\item We add up the terms for all $i_1, i_3, i_5$ values including all cases when $i_1$, $i_3$ and $i_5$ are equal among each other. Because of \eqref{eq:compl_exp_sum}, these terms sum up to $W^3$ only when $k=m+pW,\, n=k^{\prime}+pW,\, m^{\prime}=n^{\prime}+pW$, $p\in\mathbb{Z}$ otherwise they sum to 0. 
\item We subtract the terms corresponding to the cases: $i_1=i_3, i_1\neq i_5$; $i_1=i_5, i_1\neq i_3$; and $i_3=i_5, i_1\neq i_3$. As an example, the number of terms defined by $i_1=i_3, i_1\neq i_5$ are given by the difference between the number of all pairs $i_1, i_5 \in \{0,1,2,\dots,W-1\}$ and the number of terms for $i_1=i_5$. According to \eqref{eq:compl_exp_sum}, the former terms sum to $W^2$ only for $k-m+n-k^{\prime}=pW,\, m^{\prime}-n^{\prime}=pW$, whereas the latter sum to $W$ only for $k-m+n-k^{\prime}+m^{\prime}-n^{\prime}=pW$, with $p\in\mathbb{Z}$. In all other cases they all bring zero contribution. Similar results are obtained for $i_1=i_5, i_1\neq i_3$ and $i_3=i_5, i_1\neq i_3$.  
\item We finally subtract the terms $i_1=i_3=i_5$ which sum to $W$ only for $k-m+n-k^{\prime}+m^{\prime}-n^{\prime}=pW$, $p\in\mathbb{Z}$ otherwise they sum to 0 (see \eqref{eq:compl_exp_sum}).  
\end{enumerate}
Hence, we obtain
\begin{align*}
\mathsf{M}_{1}^{(1)}&=\Delta_f^3[\mathbb{E}^{3}\{|a_x|^2\}+2\mathbb{E}^2\{|a_x|^2\}|\mathbb{E}\{a_xa_y^*\}|^2+\mathbb{E}\{|a_y|^2\}|\mathbb{E}\{a_xa_y^*\}|^{2}][W^3\delta_{k-m-pW}\delta_{n-k^{\prime}-pW}\delta_{m^{\prime}-n^{\prime}-pW}
\\&-[W^2(\delta_{k-m+n-k^{\prime}-pW}\delta_{m^{\prime}-n^{\prime}-pW}+\delta_{k-m+m^{\prime}-n^{\prime}-pW}\delta_{n-k^{\prime}-pW}+\delta_{m^{\prime}-n^{\prime}+n-k^{\prime}-pW}\delta_{k-m-pW})\\
&-3W\delta_{k-m+m^{\prime}-n^{\prime}+n-k^{\prime}-pW}]-W\delta_{k-m+m^{\prime}-n^{\prime}+n-k^{\prime}-pW}]\\
&=[\mathbb{E}^{3}\{|a_x|^2\}+2\mathbb{E}^2\{|a_x|^2\}|\mathbb{E}\{a_xa_y^*\}|^2+\mathbb{E}\{|a_y|^2\}|\mathbb{E}\{a_xa_y^*\}|^{2}][R_s^3\delta_{k-m+pW}\delta_{n-k^{\prime}-pW}\delta_{m^{\prime}-n^{\prime}-pW}\\
&-R_s^2\Delta_f(\delta_{k-m+n-k^{\prime}-pW}\delta_{m^{\prime}-n^{\prime}-pW}+\delta_{k-m+m^{\prime}-n^{\prime}-pW}\delta_{n-k^{\prime}-pW}+\delta_{m^{\prime}-n^{\prime}+n-k^{\prime}-pW}\delta_{k-m-pW})\\
&+2R_s\Delta_f^2\delta_{k-m+m^{\prime}-n^{\prime}+n-k^{\prime}-pW}],
\end{align*}
where we have used $R_s=W\Delta_f$. 

The same approach can be followed to compute $\mathsf{N}_{1}^{(1)}$ which is, thus, given by
\begin{align*}
\mathsf{N}_{1}^{(1)}&\triangleq\sum_{\boldsymbol{i}\in\mathcal{C}_1^{(1)}} \mathsf{T}_{\boldsymbol{i}}(k,m,n,k^{\prime},m^{\prime},n^{\prime})=\mathbb{E}^2\{|a_x|^2\}|\mathbb{E}\{a_xa_y^*\}|^2[R_s^3\delta_{k-m-pW}\delta_{n-k^{\prime}-pW}\delta_{m^{\prime}-n^{\prime}-pW}\\&-R_s^2\Delta_f(\delta_{k-m+n-k^{\prime}-pW}\delta_{m^{\prime}-n^{\prime}-pW}+\delta_{k-m+m^{\prime}-n^{\prime}-pW}\delta_{n-k^{\prime}-pW}+\delta_{m^{\prime}-n^{\prime}+n-k^{\prime}-pW}\delta_{k-m-pW})\\
&+2R_s\Delta_f^2\delta_{k-m+m^{\prime}-n^{\prime}+n-k^{\prime}-pW}].
\end{align*}

\end{example_a}

\begin{example_a}[Contributions in $\mathcal{L}_2$]\label{ex:L2contr}
$\mathsf{M}_{2}^{(1)}$, i.e., the contribution for the set $\mathcal{C}_2^{(1)}=\{\boldsymbol{i}\in \{0,1,\dots,W-1\}^6 : i_1=i_2=i_3, i_4=i_5=i_6, i_1\neq i_4\}$ is given by

\begin{align}\label{eq:M21}
\begin{split}
\mathsf{M}_{2}^{(1)}&=\sum_{\boldsymbol{i}\in\mathcal{C}_2^{(1)}} \mathsf{S}_{\boldsymbol{i}}(k,m,n,k^{\prime},m^{\prime},n^{\prime})\\
& =\Delta_f^3[\mathbb{E}\{a_{x,i_1}|a_{x,i_1}|^2\}\mathbb{E}^*\{a_{x,i_4}|a_{x,i_4}|^2\}+\mathbb{E}\{a_{x,i_1} |a_{y,i_1}|^2\}\mathbb{E}^*\{a_{x,i_4} |a_{y,i_4}|^2\}]\\
&\cdot\sum_{i_1=0}^{W-1} e^{-j\frac{2\pi}{W}(k-m+n) i_1}\sum_{{i_4\neq i_1}}e^{-j\frac{2\pi}{W}(-k^{\prime}+m^{\prime}-n^{\prime})i_4}\\
& =\Delta_f^3[|\mathbb{E}\{a_x|a_x|^2\}|^2+|\mathbb{E}\{a_x|a_y|^2\}|^2]\sum_{i_1=0}^{W-1} e^{-j\frac{2\pi}{W}(k-m+n) i_1}\sum_{{i_4\neq i_1}}e^{-j\frac{2\pi}{W}(-k^{\prime}+m^{\prime}-n^{\prime})i_4}. 
\end{split}
\end{align}

Following a similar approach as in Example \ref{ex:L1contr}, we compute \eqref{eq:M21} by:
\begin{enumerate}
\item Adding up the terms for all $i_1$ and $i_4$ values including all cases when $i_1$, $i_4$ are equal to each other. These terms sum up to $W^2$ only when $k-m+n=pW$, and $-k^{\prime}+m^{\prime}-n^{\prime}=pW$, with $p\in\mathbb{Z}$ otherwise they sum to 0. 
\item Subtracting the terms corresponding to the cases $i_1=i_4$. These terms sum to $W$ only for $k-m+n-k^{\prime}+ m^{\prime}-n^{\prime}=pW$, $p\in\mathbb{Z}$ otherwise they sum to zero. 
\end{enumerate}
We, thus, obtain
\begin{equation*}
\mathsf{M}_{2}^{(1)}=[|\mathbb{E}\{a_x|a_x|^2\}|^2+|\mathbb{E}\{a_x|a_y|^2\}|^2][R_s^2\Delta_f\delta_{k-m+n-pW}\delta_{k^{\prime}-m^{\prime}+n^{\prime}-pW}-R_s\Delta_f^2\delta_{k-m+n-k^{\prime}+m^{\prime}-n^{\prime}-pW}],
\end{equation*}

Following the same approach for $\mathsf{N}_{2}^{(1)}$ we have 
\begin{align*}
\mathsf{N}_{2}^{(1)}&\triangleq\sum_{\boldsymbol{i}\in\mathcal{C}_1^{(1)}} \mathsf{T}_{\boldsymbol{i}}(k,m,n,k^{\prime},m^{\prime},n^{\prime})\\
&=\mathbb{E}\{a_x |a_x|^2\}\mathbb{E}\{a_x |a_y|^2\}[R_s^2\Delta_f\delta_{k-m+n}\delta_{k^{\prime}-m^{\prime}+n^{\prime}-pW}-R_s\Delta_f^2\delta_{k-m+n-k^{\prime}+m^{\prime}-n^{\prime}-pW}].
\end{align*}

\end{example_a}

\begin{example_a}[Contributions in $\mathcal{L}_3$]\label{ex:L3contr}
$\mathsf{M}_{3}^{(3)}$, i.e., the contribution for the values in the set $\mathcal{C}_3^{(3)}=\{\boldsymbol{i}\in \{0,1,\dots,W-1\}^6 :  i_1=i_4,i_2=i_3=i_5=i_6, i_1\neq i_4, i_2\neq i_3, i_5\neq i_6\}$ is given by

\begin{align}
\begin{split}
\mathsf{M}_{3}^{(3)}&=\sum_{\boldsymbol{i}\in\mathcal{C}_3^{(3)}} \mathsf{S}_{\boldsymbol{i}}(k,m,n,k^{\prime},m^{\prime},n^{\prime})\\
&=\Delta_f^3[\mathbb{E}\{|a_{x,i_1}|^2\}\mathbb{E}\{|a_{x,i_2}|^4\}+\mathbb{E}\{|a_{y,i_1}|^2\}\mathbb{E}\{|a_{x,i_2}|^2|a_{y,i_2}|^2\}]\\
&\cdot\sum_{i_1=1}^{W-1}e^{-j\frac{2\pi}{W}(k-k^{\prime})i_1}\sum_{i_2\neq i_1}e^{-j\frac{2\pi}{W}(-m+n+m^{\prime}-n^{\prime})i_2}\\
&=\Delta_f^3[\mathbb{E}\{|a_x|^2\}\mathbb{E}\{|a_x|^4\}+\mathbb{E}\{|a_y|^2\}\mathbb{E}\{|a_x|^2|a_y|^2\}]\sum_{i_1=1}^{W-1}e^{-j\frac{2\pi}{W}(k-k^{\prime})i_1}\sum_{i_2\neq i_1}e^{-j\frac{2\pi}{W}(-m+n+m^{\prime}-n^{\prime})i_2}. \label{eq:M33_ex}
\end{split}
\end{align}
As in the $\mathcal{L}_2$ case described in Example \ref{ex:L2contr}, in $\mathcal{L}_3$ each subset is characterized by 2 subgroups of indices. Hence, the approach followed to compute \eqref{eq:M33_ex} is identical to \eqref{eq:M21}, and gives 
\begin{align*}
\mathsf{M}_{3}^{(3)}&=[\mathbb{E}\{|a_x|^4\}\mathbb{E}\{|a_x|^2\}+\mathbb{E}\{|a_x|^2|a_y|^2\}\mathbb{E}\{|a_y|^2\}][R_s^2\Delta_f\delta_{k-k^{\prime}-pW}\delta_{m-n-m^{\prime}+n^{\prime}-pW}-R_s\Delta_f^2\delta_{k-m+n-k^{\prime}+m^{\prime}-n^{\prime}-pW}].
\end{align*}
Similarly,
\begin{equation*}
\mathsf{N}_{3}^{(3)}=\mathbb{E}\{a_x a_y^*\}\mathbb{E}\{a_x^*a_y|a_x|^2\}[R_s^2\Delta_f\delta_{k-k^{\prime}-pW}\delta_{m-n-m^{\prime}+n^{\prime}-pW}-R_s\Delta_f^2\delta_{k-m+n-k^{\prime}+m^{\prime}-n^{\prime}-pW}].    
\end{equation*}
\end{example_a}

As shown in the above examples, each contribution $\mathsf{M}_{g}^{(h)}$, $\mathsf{N}_{g}^{(h)}$ is nonzero only for a specific set of $(k,m,n,k^{\prime},m^{\prime},n^{\prime})$ values which is spanned by $p\in\mathbb{Z}$. However, the terms $(k,m,n,k^{\prime},m^{\prime},n^{\prime})$ arising for all $p\neq0$ bring a total contribution to \eqref{eq:6D_sum} that can be considered negligible. This is due to our assumption on $P(f)$ being strictly band-limited (see Sec.~\ref{sec:system_model}), and to the magnitude of the functions product $\eta_{k,m,n}\eta^*_{k^{\prime}, m^{\prime}, n^{\prime}}$ (see definitions \eqref{eq:eta_def} and \eqref{eq:P_def}). Thus, in the computations performed in the following subsections, we will restrict ourselves to the case $p=0$.

\subsection{Contributions  in $\mathcal{L}_{1}$}\label{subsec:L1contr}
In this section, the contributions  $\mathsf{M}_{1}^{(i)}$, $\mathsf{N}_{1}^{(i)}$ for $i=1,2,\dots,15$ are computed following Example~\ref{ex:L1contr}. These contributions are listed in Table \ref{tab:MN1contr}.

{\small
\begin{longtable}{ p{.03\textwidth} | p{.23\textwidth} | p{.22\textwidth} |
p{.4\textwidth}}
\caption{\scriptsize\uppercase{List of contributions} $\mathsf{M}_{1}^{(h)}$ \uppercase{and} $\mathsf{N}_{1}^{(h)}$ for $i=1,2,\dots,15$.}
\label{tab:MN1contr}\\
\hline \hline 
$h$ & \textbf{Corr. terms in}  $\mathsf{M}_{1}^{(h)}$ & \textbf{Corr. terms in} $\mathsf{N}_{1}^{(h)}$ &
\textbf{Delta products} \endhead \\
\hline \hline
1 & \makecell[l]{$\mathbb{E}^{3}\{|a_x|^2\}$\\$+|\mathbb{E}\{a_xa_y^*\}|^{2}\mathbb{E}\{|a_y|^2\}$} & $\mathbb{E}\{|a_x|^2\}|\mathbb{E}\{a_xa_y^*\}|^2$ & \makecell[l]{$R_s^3\delta_{k-m}\delta_{n-k^{\prime}}\delta_{m^{\prime}-n^{\prime}}-R_s^2\Delta_f(\delta_{m^{\prime}-n^{\prime}}\delta_{k-m+n-k^{\prime}}$\\$+\delta_{n-k^{\prime}}\delta_{k-m+m^{\prime}-n^{\prime}}+\delta_{k-m}\delta_{m^{\prime}-n^{\prime}+n-k^{\prime}})$\\
$+2R_s\Delta_f^2\delta_{k-m+n-k^{\prime}+m^{\prime}-n^{\prime}}$} \\
\hline

2 & \makecell[l]{$\mathbb{E}\{|a_x|^2\}|\mathbb{E}\{a_x^2\}|^2$\\$+|\mathbb{E}\{a_xa_y\}|^{2}\mathbb{E}\{|a_y|^2\}$} & $\mathbb{E}\{|a_x|^2\}|\mathbb{E}\{a_xa_y\}|^2$ & \makecell[l]{$R_s^3\delta_{k-m}\delta_{n+m^{\prime}}\delta_{k^{\prime}+n^{\prime}}-R_s^2\Delta_f(\delta_{k^{\prime}+n^{\prime}}\delta_{k-m+n+m^{\prime}}$\\
$+\delta_{n+m^{\prime}}\delta_{k-m-k^{\prime}-n^{\prime}}+\delta_{k-m}\delta_{n-k^{\prime}+m^{\prime}-n^{\prime}})$\\
$+2R_s\Delta_f^2\delta_{k-m+n-k^{\prime}+m^{\prime}-n^{\prime}}$} \\
\hline

3&\makecell[l]{$\mathbb{E}^{3}\{|a_x|^2\}$\\$+\mathbb{E}\{|a_x|^2\}\mathbb{E}^2\{|a_y|^2\}$} & $\mathbb{E}^2\{|a_x|^2\}\mathbb{E}\{|a_y|^2\}$&
\makecell[l]{$R_s^3\delta_{k-m}\delta_{n-n^{\prime}}\delta_{k^{\prime}-m^{\prime}}-R_s^2\Delta_f(\delta_{k^{\prime}-m^{\prime}}\delta_{k-m+n-n^{\prime}}$\\
$+\delta_{n-n^{\prime}}\delta_{k-m-k^{\prime}+m^{\prime}}+\delta_{k-m}\delta_{n-n^{\prime}-k^{\prime}+m^{\prime}})$\\$+2R_s\Delta_f^2\delta_{k-m+n-k^{\prime}+m^{\prime}-n^{\prime}}$}\\
\hline
4&\makecell[l]{$\mathbb{E}\{|a_x|^2\}|\mathbb{E}\{a_x^2\}|^2$\\$+\mathbb{E}\{a_xa_y\}\mathbb{E}\{a_x^*a_y\}\mathbb{E}^*\{a_y^2\}$} & \makecell[l]{$\mathbb{E}\{a_x^2\}\mathbb{E}^*\{a_xa_y\}\mathbb{E}\{a_x^*a_y\}$} & \makecell[l]{$-R_s^2\Delta_f(\delta_{m^{\prime}-n^{\prime}}\delta_{k+n-m-k^{\prime}}+\delta_{m+k^{\prime}}\delta_{k+n+m^{\prime}-n^{\prime}}$\\
$+\delta_{k+n}\delta_{m+k^{\prime}-m^{\prime}+n^{\prime}})+2R_s\Delta_f^2\delta_{k-m+n-k^{\prime}+m^{\prime}-n^{\prime}}$}\\
\hline
5& \makecell[l]{$\mathbb{E}\{|a_x|^2\}|\mathbb{E}\{a_x^2\}|^2$\\$\cdot|\mathbb{E}\{a_xa_y\}|^2\mathbb{E}\{|a_y|^2\}$} & \makecell[l]{$\mathbb{E}\{a_x^2\}\mathbb{E}^*\{a_xa_y\}\mathbb{E}\{a_x^*a_y\}$} & \makecell[l]{$R_s^3\delta_{k+n}\delta_{m-m^{\prime}}\delta_{k^{\prime}+n^{\prime}}-R_s^2\Delta_f(\delta_{k^{\prime}+n^{\prime}}\delta_{k+n-m+m^{\prime}}$\\$+\delta_{m-m^{\prime}}\delta_{k+n-k^{\prime}-n^{\prime}}+\delta_{k+n}\delta_{m-m^{\prime}+k^{\prime}+n^{\prime}})$\\
$+2R_s\Delta_f^2\delta_{k-m+n-k^{\prime}+m^{\prime}-n^{\prime}}$}\\
\hline 
6 & \makecell[l]{$\mathbb{E}\{|a_x|^2\}|\mathbb{E}\{a_x^2\}|^2$\\$+|\mathbb{E}\{a_xa_y\}|^2\mathbb{E}\{|a_y|^2\}$}  & $|\mathbb{E}\{a_x^2\}|^2\mathbb{E}\{|a_y|^2\}$ & \makecell[l]{$R_s^3\delta_{k+n}\delta_{m+n^{\prime}}\delta_{k^{\prime}-m^{\prime}}-R_s^2\Delta_f(\delta_{k^{\prime}-m^{\prime}}\delta_{k+n-m-n^{\prime}}$\\
$+\delta_{m+n^{\prime}}\delta_{k+n-k^{\prime}+m^{\prime}}+\delta_{k+n}\delta_{m+k^{\prime}-m^{\prime}+n^{\prime}})$\\
$+2R_s\Delta_f^2\delta_{k-m+n-k^{\prime}+m^{\prime}-n^{\prime}}$}\\
\hline
7 & \makecell[l]{$\mathbb{E}^3\{|a_x|^2\}$\\$+|\mathbb{E}\{a_xa_y^*\}|^2\mathbb{E}\{|a_y|^2\}$} & $\mathbb{E}\{|a_x|^2\}|\mathbb{E}\{a_x a_y^*\}|^2$ & \makecell[l]{$R_s^3\delta_{k-k^{\prime}}\delta_{m-n}\delta_{m^{\prime}-n^{\prime}}-R_s^2\Delta_f(\delta_{m^{\prime}-n^{\prime}}\delta_{k-k^{\prime}-m+n}$\\
$+\delta_{m-n}\delta_{k-k^{\prime}+m^{\prime}-n^{\prime}}+\delta_{k-k^{\prime}}\delta_{m-n-m^{\prime}+n^{\prime}})$\\
$+2R_s\Delta_f^2\delta_{k-m+n-k^{\prime}+m^{\prime}-n^{\prime}}$}\\
\hline 
8 & \makecell[l]{$\mathbb{E}^3\{|a_x|^2\}$\\$+\mathbb{E}\{|a_x|^2\}\mathbb{E}^2\{|a_y|^2\}$ }& $\mathbb{E}\{|a_x|^2\}|\mathbb{E}\{a_x a_y^*\}|^2$ & \makecell[l]{$R_s^3\delta_{k-k^{\prime}}\delta_{m-m^{\prime}}\delta_{n-n^{\prime}}-R_s^2\Delta_f(\delta_{n-n^{\prime}}\delta_{k-m-k^{\prime}+m^{\prime}}$\\
$+\delta_{m-m^{\prime}}\delta_{k+n-k^{\prime}-n^{\prime}}+\delta_{k-k^{\prime}}\delta_{m-n-m^{\prime}+n^{\prime}})$\\
$+2R_s\Delta_f^2\delta_{k-m+n-k^{\prime}+m^{\prime}-n^{\prime}}$}\\
\hline
9 & \makecell[l]{$\mathbb{E}\{|a_x|^2\}|\mathbb{E}\{a_x^2\}|^2$\\$+|\mathbb{E}\{a_xa_y\}|^2\mathbb{E}\{|a_y|^2\}$} & \makecell[l]{$\mathbb{E}^*\{a_x^2\}\mathbb{E}\{a_xa_y\}\mathbb{E}\{a_xa_y^*\}$} & \makecell[l]{$R_s^3\delta_{k-k^{\prime}}\delta_{m+n^{\prime}}\delta_{n+m^{\prime}}-R_s^2\Delta_f(\delta_{n+m^{\prime}}\delta_{k-m-k^{\prime}-n^{\prime}}$\\
$+\delta_{m+n^{\prime}}\delta_{k+n-k^{\prime}+m^{\prime}}+\delta_{k-k^{\prime}}\delta_{m-n-m^{\prime}+n^{\prime}})$\\
$+2R_s\Delta_f^2\delta_{k-m+n-k^{\prime}+m^{\prime}-n^{\prime}}$}\\
\hline
10 & \makecell[l]{$\mathbb{E}\{|a_x|^2\}|\mathbb{E}\{a_x^2\}|^2$\\$+\mathbb{E}^*\{a_xa_y\}\mathbb{E}\{a_x a_y^*\}\mathbb{E}\{a_y^2\}$} & $\mathbb{E}\{|a_x|^2\}|\mathbb{E}\{a_xa_y\}|^2$
& \makecell[l]{$R_s^3\delta_{k+m^{\prime}}\delta_{m-n}\delta_{k^{\prime}+n^{\prime}}-R_s^2\Delta_f(\delta_{k^{\prime}+n^{\prime}}\delta_{k-m+n+m^{\prime}}$\\
$+\delta_{m-n}\delta_{k-k^{\prime}+m^{\prime}-n^{\prime}}+\delta_{k+m^{\prime}}\delta_{m-n+k^{\prime}+n^{\prime}})$\\
$+2R_s\Delta_f^2\delta_{k-m+n-k^{\prime}+m^{\prime}-n^{\prime}}$}\\
\hline
11 & \makecell[l]{$\mathbb{E}\{|a_x|^2\}|\mathbb{E}\{a_x^2\}|^2$\\$+\mathbb{E}\{|a_x|^2\}|\mathbb{E}\{a_y^2\}|^2$} & $\mathbb{E}\{|a_x|^2\}|\mathbb{E}\{a_xa_y\}|^2$ & \makecell[l]{ $R_s^3\delta_{k+m^{\prime}}\delta_{m+k^{\prime}}\delta_{n-n^{\prime}}-R_s^2\Delta_f(\delta_{n-n^{\prime}}\delta_{k-m-k^{\prime}+m^{\prime}}$\\$+\delta_{m+k^{\prime}}\delta_{k+n+m^{\prime}-n^{\prime}}+\delta_{k+m^{\prime}}\delta_{m-n+k^{\prime}+n^{\prime}})$\\$+2R_s\Delta_f^2\delta_{k-m+n-k^{\prime}+m^{\prime}-n^{\prime}}$}\\
\hline
12 & \makecell[l]{$\mathbb{E}\{|a_x|^2\}|\mathbb{E}\{a_x^2\}|^2$\\$+\mathbb{E}^*\{a_xa_y\}\mathbb{E}\{a_x a_y^*\}\mathbb{E}\{a_y^2\}$} & \makecell[l]{$\mathbb{E}^*\{a_x^2\}\mathbb{E}\{a_xa_y\}\mathbb{E}\{a_x a_y^*\}$} & \makecell[l]{$R_s^3\delta_{k+m^{\prime}}\delta_{m+n^{\prime}}\delta_{n-k^{\prime}}-R_s^2\Delta_f(\delta_{n-k^{\prime}}\delta_{k-m+m^{\prime}-n^{\prime}}$\\$+\delta_{m+n^{\prime}}\delta_{k+n-k^{\prime}+m^{\prime}}+\delta_{k+m^{\prime}}\delta_{m-n+k^{\prime}+n^{\prime}})$\\
$+2R_s\Delta_f^2\delta_{k-m+n-k^{\prime}+m^{\prime}-n^{\prime}}$}\\
\hline
13 & \makecell[l]{$\mathbb{E}^3\{|a_x|^2\}$\\$+|\mathbb{E}\{a_x a_y^*\}|^2\mathbb{E}\{|a_y|^2\}$} & $\mathbb{E}^2\{|a_x|^2\}\mathbb{E}\{|a_y|^2\}$ &
\makecell[l]{$R_s^3\delta_{k-n^{\prime}}\delta_{m-n}\delta_{k^{\prime}-m^{\prime}}-R_s^2\Delta_f(\delta_{k^{\prime}-m^{\prime}}\delta_{k-m+n-n^{\prime}}$\\$+\delta_{m-n}\delta_{k-k^{\prime}+m^{\prime}-n^{\prime}}+\delta_{k-n^{\prime}}\delta_{m-n+k^{\prime}-m^{\prime}})$\\
$+2R_s\Delta_f^2\delta_{k-m+n-k^{\prime}+m^{\prime}-n^{\prime}}$}\\
\hline
14 & \makecell[l]{$\mathbb{E}\{|a_x|^2\}|\mathbb{E}\{a_x^2\}|^2$\\$+\mathbb{E}\{a_xa_y\}\mathbb{E}\{_xa^*a_y\}\mathbb{E}^*\{a_y^2\}$} & $\mathbb{E}\{|a_x|^2\}|\mathbb{E}\{a_xa_y\}|^2$ & \makecell[l]{$R_s^3\delta_{k-n^{\prime}}\delta_{m+k^{\prime}}\delta_{n+m^{\prime}}-R_s^2\Delta_f(\delta_{n+m^{\prime}}\delta_{k-m-k^{\prime}-n^{\prime}}$\\
$+\delta_{m+k^{\prime}}\delta_{k+n+m^{\prime}-n^{\prime}}+\delta_{k-n^{\prime}}\delta_{m-n+k^{\prime}-m^{\prime}})$\\
$+2R_s\Delta_f^2\delta_{k-m+n-k^{\prime}+m^{\prime}-n^{\prime}}$}\\
\hline
15 & \makecell[l]{$\mathbb{E}^3\{|a_x|^2\}$\\$+|\mathbb{E}\{a_x a_y^*\}|^2\mathbb{E}\{|a_y|^2\}$} & $\mathbb{E}\{|a_x|^2\}|\mathbb{E}\{a_x a_y^*\}|^2$ & \makecell[l]{$R_s^3\delta_{k-n^{\prime}}\delta_{m-m^{\prime}}\delta_{n-k^{\prime}}-R_s^2\Delta_f(\delta_{n-k^{\prime}}\delta_{k-m+m^{\prime}-n^{\prime}}$\\
$+\delta_{m-m^{\prime}}\delta_{k+n-k^{\prime}-n^{\prime}}+\delta_{k-n^{\prime}}\delta_{m-n+k^{\prime}-m^{\prime}})$\\
$+2R_s\Delta_f^2\delta_{k-m+n-k^{\prime}+m^{\prime}-n^{\prime}}$}\\
\hline\hline
\end{longtable}
}

\subsection{Contributions in $\mathcal{L}_{2}$}\label{subsec:L2contr}
Following Example \ref{ex:L2contr}, the contributions $\mathsf{M}_{2}^{(h)}, \mathsf{N}_{2}^{(h)}, h=1,2,\dots,10$ are computed and listed in Table \ref{tab:MN2contr}.

{\small
\begin{longtable}[tbp]{ p{.03\textwidth} | p{.36\textwidth} | p{.21\textwidth} |
p{.25\textwidth}}
\caption{\scriptsize\uppercase{List of contributions} $\mathsf{M}_{2}^{(h)}$ and $\mathsf{N}_{2}^{(h)}$ \uppercase{for} $i=1,2,\ldots,10$.} 
\label{tab:MN2contr}\\
\hline \hline
$h$ & \textbf{Corr. terms in  $\mathsf{M}_{2}^{(h)}$} & \textbf{Corr. terms in} $\mathsf{N}_{2}^{(h)}$ &
\textbf{Delta products} \\
\hline \hline 
1 & $|\mathbb{E}\{a_x|a_x|^2\}|^2+|\mathbb{E}\{a_x|a_y|^2\}|^2$ & $\mathbb{E}\{a_x |a_x|^2\}\mathbb{E}\{a_x |a_y|^2\}$ & \makecell[l]{$R_s^2\Delta_f\delta_{k-m+n}\delta_{k^{\prime}-m^{\prime}+n^{\prime}}$\\
$-R_s\Delta_f^2\delta_{k-m+n-k^{\prime}+m^{\prime}-n^{\prime}}$}\\
\hline
2 & $|\mathbb{E}\{a_x|a_x|^2\}|^2+\mathbb{E}\{a_y^* |a_y|^2\}\mathbb{E}\{|a_x|^2a_y\}$ & $|\mathbb{E}\{ |a_x|^2a_y\}|^2$ & \makecell[l]{$2R_s^2\Delta_f\delta_{k-m-k^{\prime}}\delta_{n+m^{\prime}-n^{\prime}}$\\
$-R_s\Delta_f^2\delta_{k-m+n-k^{\prime}+m^{\prime}-n^{\prime}}$}\\
\hline
3 & $|\mathbb{E}\{a_x|a_x|^2\}|^2+\mathbb{E}\{|a_x|^2a_y^*\}\mathbb{E}\{a_y |a_y|^2\}$ & $|\mathbb{E}\{ |a_x|^2a_y\}|^2$ & \makecell[l]{$R_s^2\Delta_f\delta_{k-m+m^{\prime}}\delta_{n-k^{\prime}-n^{\prime}}$\\
$-R_s\Delta_f^2\delta_{k-m+n-k^{\prime}+m^{\prime}-n^{\prime}}$} \\
\hline
4 & $|\mathbb{E}\{a_x|a_x|^2\}|^2+|\mathbb{E}\{a_x|a_y|^2\}|^2$ & $\mathbb{E}\{a_x^*|a_x|^2\}\mathbb{E}\{a_x |a_y|^2\}$ & \makecell[l]{$R_s^2\Delta_f\delta_{k-m-n^{\prime}}\delta_{n-k^{\prime}+m^{\prime}}$\\
$-R_s\Delta_f^2\delta_{k-m+n-k^{\prime}+m^{\prime}-n^{\prime}}$}\\
\hline
5& $|\mathbb{E}\{a_x|a_x|^2\}|^2+|\mathbb{E}\{a_x|a_y|^2\}|^2$ & $|\mathbb{E}\{a_x^2a_y^*\}|^2$ & \makecell[l]{$R_s^2\Delta_f\delta_{k+n-k^{\prime}}\delta_{m-m^{\prime}+n^{\prime}}$\\
$-R_s\Delta_f^2\delta_{k-m+n-k^{\prime}+m^{\prime}-n^{\prime}}$}\\
\hline
6& $|\mathbb{E}\{a_x^3\}|^2+|\mathbb{E}\{a_xa_y^2\}|^2$ & $|\mathbb{E}\{a_x^2a_y\}|^2$ & \makecell[l]{$R_s^2\Delta_f\delta_{k+n+m^{\prime}}\delta_{m+k^{\prime}+n^{\prime}}$\\
$-R_s\Delta_f^2\delta_{k-m+n-k^{\prime}+m^{\prime}-n^{\prime}}$}\\
\hline
7 & $|\mathbb{E}\{a_x|a_x|^2\}|^2+\mathbb{E}\{|a_x|^2a_y\}\mathbb{E}\{a_y^*|a_y|^2\}$ & $\mathbb{E}\{a_x |a_x|^2\}\mathbb{E}\{a_x^*|a_y|^2\}$ & \makecell[l]{$R_s^2\Delta_f\delta_{k+n-n^{\prime}}\delta_{m+k^{\prime}-m^{\prime}}$\\
$-R_s\Delta_f^2\delta_{k-m+n-k^{\prime}+m^{\prime}-n^{\prime}}$}\\
\hline 
8 & $|\mathbb{E}\{a_x|a_x|^2\}|^2+\mathbb{E}\{ |a_x|^2a_y^*\}\mathbb{E}\{a_y |a_y|^2\}$ & $\mathbb{E}\{a_x |a_y|^2\}\mathbb{E}\{a_x^*|a_x|^2\}$ & \makecell[l]{$R_s^2\Delta_f\delta_{k-k^{\prime}+m^{\prime}}\delta_{m-n+n^{\prime}}$\\
$-R_s\Delta_f^2\delta_{k-m+n-k^{\prime}+m^{\prime}-n^{\prime}}$}\\
\hline
9 & $|\mathbb{E}\{a_x|a_x|^2\}|^2+|\mathbb{E}\{a_x|a_y|^2\}|^2$ & $|\mathbb{E}\{|a_x|^2a_y\}|^2$ & \makecell[l]{$R_s^2\Delta_f\delta_{k-k^{\prime}-n^{\prime}}\delta_{m-n-m^{\prime}}$\\
$-R_s\Delta_f^2\delta_{k-m+n-k^{\prime}+m^{\prime}-n^{\prime}}$}\\

\hline 
10 & $|\mathbb{E}\{a_x|a_x|^2\}|^2+|\mathbb{E}\{a_x^*a_y^2\}|^2$ & $|\mathbb{E}\{|a_x|^2a_y\}|^2$ & \makecell[l]{$R_s^2\Delta_f\delta_{k+m^{\prime}-n^{\prime}}\delta_{m-n+k^{\prime}}$\\
$-R_s\Delta_f^2\delta_{k-m+n-k^{\prime}+m^{\prime}-n^{\prime}}$}\\
\hline\hline
\end{longtable}
}

\subsection{Contributions in $\mathcal{L}_{3}$}\label{subsec:L3contr}

Following Example \ref{ex:L3contr}, the contributions $\mathsf{M}_{3}^{(h)}, \mathsf{N}_{3}^{(h)}, h=1,2,\dots,15$, are computed and listed in Table \ref{tab:MN3contr}.
{\small
\begin{longtable}{ p{.02\textwidth} | p{.38\textwidth} | p{.21\textwidth} |
p{.23\textwidth}}
\caption{\scriptsize\uppercase{List of contributions} $\mathsf{M}_{3}^{(h)}$ \uppercase{and} $\mathsf{N}_{3}^{(h)}$ for $i=1,2,\dots,15$.} 
\label{tab:MN3contr}\\
\hline \hline
$h$ & \textbf{Corr. terms in}  $\mathsf{M}_{3}^{(h)}$ & \textbf{Corr. terms in} $\mathsf{N}_{3}^{(h)}$ &
\textbf{Delta products} \endhead \\
\hline \hline 
1 & $\mathbb{E}\{|a_x|^4\}\mathbb{E}\{|a_x|^2\}+\mathbb{E}\{|a_x|^2 |a_y|^2\}\mathbb{E}\{|a_y|^2\}$ & $\mathbb{E}\{|a_x|^2\}\mathbb{E}\{|a_x|^2 |a_y|^2\}$ & \makecell[l]{$R_s^2\Delta_f\delta_{k-m}\delta_{n-k^{\prime}+m^{\prime}-n^{\prime}}$\\
$-R_s\Delta_f^2\delta_{k-m+n-k^{\prime}+m^{\prime}-n^{\prime}}$}\\
\hline
2 & \makecell[l]{$\mathbb{E}^*\{a_x^2|a_x|^2\}\mathbb{E}\{a_x^2\}+\mathbb{E}\{a_xa_y\}\mathbb{E}^*\{a_xa_y|a_y|^2\}$} & $\mathbb{E}\{a_x^2\}\mathbb{E}^*\{a_x^2 |a_y|^2\}$ & \makecell[l]{$R_s^2\Delta_f\delta_{k+n}\delta_{m+k^{\prime}-m^{\prime}+n^{\prime}}$\\
$-R_s\Delta_f^2\delta_{k-m+n-k^{\prime}+m^{\prime}-n^{\prime}}$}\\
\hline
3 & $\mathbb{E}\{|a_x|^4\}\mathbb{E}\{|a_x|^2\}+\mathbb{E}\{|a_x|^2|a_y|^2\}\mathbb{E}\{|a_y|^2\}$ & $\mathbb{E}\{a_x a_y^*\}\mathbb{E}\{a_x^*a_y|a_x|^2\}$ & \makecell[l]{$R_s^2\Delta_f\delta_{k-k^{\prime}}\delta_{m-n-m^{\prime}+n^{\prime}}$\\
$-R_s\Delta_f^2\delta_{k-m+n-k^{\prime}+m^{\prime}-n^{\prime}}$}\\
\hline 
4 & $\mathbb{E}^*\{a_x^2|a_x|^2\}\mathbb{E}\{a_x^2\}+\mathbb{E}^*\{|a_x|^2a_y^2\}\mathbb{E}\{a_y^2\}$ & $\mathbb{E}\{a_xa_y\}\mathbb{E}^*\{a_xa_y|a_x|^2\}$ & \makecell[l]{$R_s^2\Delta_f\delta_{k+m^{\prime}}\delta_{m-n+k^{\prime}+n^{\prime}}$\\
$-R_s\Delta_f^2\delta_{k-m+n-k^{\prime}+m^{\prime}-n^{\prime}}$}\\
\hline 
5 & $\mathbb{E}\{|a_x|^4\}\mathbb{E}\{|a_x|^2\}+\mathbb{E}\{a_x^*a_y\}\mathbb{E}\{a_x a_y^*|a_y|^2\}$ & $\mathbb{E}\{|a_x|^2\}\mathbb{E}\{|a_x|^2|a_y|^2\}$ & \makecell[l]{$R_s^2\Delta_f\delta_{k-n^{\prime}}\delta_{m-n+k^{\prime}-m^{\prime}}$\\$-R_s\Delta_f^2\delta_{k-m+n-k^{\prime}+m^{\prime}-n^{\prime}}$}\\
\hline
6 & $\mathbb{E}\{|a_x|^4\}\mathbb{E}\{|a_x|^2\}+\mathbb{E}\{a_x a_y^*\}\mathbb{E}\{a_x^*a_y|a_y|^2\}$ & $\mathbb{E}\{|a_x|^2\}\mathbb{E}\{|a_x|^2|a_y|^2\}$ & \makecell[l]{$R_s^2\Delta_f\delta_{m-n}\delta_{k-k^{\prime}+m^{\prime}-n^{\prime}}$\\
$-R_s\Delta_f^2\delta_{k-m+n-k^{\prime}+m^{\prime}-n^{\prime}}$} \\
\hline 
7 & $\mathbb{E}\{a_x^2|a_x|^2\}\mathbb{E}^*\{a^2_x\}+\mathbb{E}\{|a_x|^2 a_y^2\}\mathbb{E}^*\{a_y^2\}$ & $\mathbb{E}^*\{a_xa_y\}\mathbb{E}\{a_xa_y|a_x|^2\}$ & \makecell[l]{$R_s^2\Delta_f\delta_{m+k^{\prime}}\delta_{k-n+m^{\prime}-n^{\prime}}$\\
$-R_s\Delta_f^2\delta_{k-m+n-k^{\prime}+m^{\prime}-n^{\prime}}$} \\
\hline 
8 & $\mathbb{E}\{|a_x|^4\}\mathbb{E}\{|a_x|^2\}+\mathbb{E}\{|a_x|^2 |a_y|^2\}\mathbb{E}\{ |a_y|^2\}$ & $\mathbb{E}\{a_x^*a_y\}\mathbb{E}\{a_x a_y^*|a_x|^2\}$ & \makecell[l]{$R_s^2\Delta_f\delta_{m-m^{\prime}}\delta_{k+n-k^{\prime}-n^{\prime}}$\\
$-R_s\Delta_f^2\delta_{k-m+n-k^{\prime}+m^{\prime}-n^{\prime}}$}\\
\hline
9 & \makecell[l]{$\mathbb{E}\{a_x^2|a_x|^2\}\mathbb{E}^*\{a_x^2\}+\mathbb{E}^*\{a_xa_y\}\mathbb{E}\{a_xa_y|a_y|^2\}$} & $\mathbb{E}^*\{a^2_x\}\mathbb{E}\{a^2_x |a_y|^2\}$ & \makecell[l]{$R_s^2\Delta_f\delta_{m+n^{\prime}}\delta_{k+n-k^{\prime}+m^{\prime}}$\\
$-R_s\Delta_f^2\delta_{k-m+n-k^{\prime}+m^{\prime}-n^{\prime}}$}\\
\hline
10 & $\mathbb{E}\{|a_x|^4\}\mathbb{E}\{|a_x|^2\}+\mathbb{E}\{a_x a_y^*\}\mathbb{E}\{a_x^*a_y|a_y|^2\}$ & $\mathbb{E}\{a_x a_y^*\}\mathbb{E}\{a_x^*a_y |a_x|^2\}$ & \makecell[l]{$R_s^2\Delta_f\delta_{n-k^{\prime}}\delta_{k-m+m^{\prime}-n^{\prime}}$\\
$-R_s\Delta_f^2\delta_{k-m+n-k^{\prime}+m^{\prime}-n^{\prime}}$}\\
\hline
11 & \makecell[l]{$\mathbb{E}^*\{a_x^2|a_x|^2\}\mathbb{E}\{a_x^2\}+\mathbb{E}\{a_xa_y\}\mathbb{E}^*\{a_xa_y|a_y|^2\}$} & $\mathbb{E}\{a_xa_y\}\mathbb{E}^*\{a_xa_y|a_x|^2\}$ & \makecell[l]{$R_s^2\Delta_f\delta_{n+m^{\prime}}\delta_{k-m-k^{\prime}-n^{\prime}}$\\
$-R_s\Delta_f^2\delta_{k-m+n-k^{\prime}+m^{\prime}-n^{\prime}}$}\\
\hline
12 & $\mathbb{E}\{|a_x|^4\}\mathbb{E}\{|a_x|^2\}+\mathbb{E}\{|a_x|^2\}\mathbb{E}\{|a_y|^4\}$ & $\mathbb{E}\{|a_x|^2\}\mathbb{E}\{|a_x|^2|a_y|^2\}$ & \makecell[l]{$R_s^2\Delta_f\delta_{n-n^{\prime}}\delta_{k-m-k^{\prime}+m^{\prime}}$\\
$-R_s\Delta_f^2\delta_{k-m+n-k^{\prime}+m^{\prime}-n^{\prime}}$}\\
\hline 
13 & $\mathbb{E}\{|a_x|^4\}\mathbb{E}\{|a_x|^2\}+\mathbb{E}\{|a_x|^2|a_y|^2\}\mathbb{E}\{|a_y|^2\}$ & $\mathbb{E}\{|a_x|^4\}\mathbb{E}\{|a_y|^2\}$ & \makecell[l]{$R_s^2\Delta_f\delta_{k^{\prime}-m^{\prime}}\delta_{k-m+n-n^{\prime}}$\\$-R_s\Delta_f^2\delta_{k-m+n-k^{\prime}+m^{\prime}-n^{\prime}}$}\\
\hline 
14 & \makecell[l]{$\mathbb{E}\{a_x^2|a_x|^2\}\mathbb{E}^*\{a_x^2\}+\mathbb{E}^*\{a_xa_y\}\mathbb{E}\{a_xa_y|a_y|^2\}$} & $\mathbb{E}^*\{a_xa_y\}\mathbb{E}\{a_xa_y|a_x|^2\}$ & \makecell[l]{$R_s^2\Delta_f\delta_{k^{\prime}+n^{\prime}}\delta_{k-m+n+m^{\prime}}$\\
$-R_s\Delta_f^2\delta_{k-m+n-k^{\prime}+m^{\prime}-n^{\prime}}$}\\
\hline
15 & $\mathbb{E}\{|a_x|^4\}\mathbb{E}\{|a_x|^2\}+\mathbb{E}\{a_x^*a_y\}\mathbb{E}\{a_xa_y^*|a_y|^2\}$ & $\mathbb{E}\{a_x^*a_y\}\mathbb{E}\{a_xa_y^*|a_x|^2$ & \makecell[l]{$R_s^2\Delta_f\delta_{m^{\prime}-n^{\prime}}\delta_{k-m+n-k^{\prime}}$\\
$-R_s\Delta_f^2\delta_{k-m+n-k^{\prime}+m^{\prime}-n^{\prime}}$}\\
\hline\hline
\end{longtable}
}

\subsection{Contributions in $\mathcal{L}_{4}$}\label{subsec:L4contr}
Since $\mathcal{L}_{4}$ comprises a single subset characterised by the single subgroup of all 6 indices (see Sec.~\ref{subsec:subset_part}), only one pair of contributions $\mathsf{M}_{4}^{(1)}$, $\mathsf{N}_{4}^{(1)}$ exists and it is given by
\begin{align*}
\mathsf{M}_{4}^{(1)}&\triangleq\sum_{\boldsymbol{i}\in\mathcal{C}_4^{(1)}} \mathsf{S}_{\boldsymbol{i}}(k,m,n,k^{\prime},m^{\prime},n^{\prime})\\
&=\sum_{i_1=0}^{W-1} [\mathbb{E}\{|a_x|^6\}+\mathbb{E}\{|a_x|^2|a_y|^4\}]e^{-j\frac{2\pi}{W}(k-m) i_1}\\
&=[\mathbb{E}\{|a_x|^6\}+\mathbb{E}\{|a_x|^2|a_y|^4\}]R_s\Delta_f^2\delta_{k-m+n-k^{\prime}+m^{\prime}-n^{\prime}},\\
\mathsf{N}_{4}^{(1)}&\triangleq\sum_{\boldsymbol{i}\in\mathcal{C}_4^{(1)}} \mathsf{S}_{\boldsymbol{i}}(k,m,n,k^{\prime},m^{\prime},n^{\prime})=\mathbb{E}\{|a_x|^4 |a_y|^2\}R_s\Delta_f^2\delta_{k-m+n-k^{\prime}+m^{\prime}-n^{\prime}}.
\end{align*}

\section{Sum of all contributions}\label{sec:sumallcontr}
In Sec.~\ref{sec:evaluation}, we evaluated all the contributions $\mathsf{M}_{g}^{(h)}$ and $\mathsf{N}_{g}^{(h)}$ to the PSD in \eqref{eq:6D_sum}. In particular, from \eqref{eq:6D_sum}, \eqref{eq:S_sum}, \eqref{eq:T_sum}, and \eqref{eq:MN_contr} we have
\begin{align}
\begin{split}
&S_x(f,N_s,L_s)=\left(\frac{8}{9}\right)^2\gamma^2\Delta_f^3\sum_{i=-\infty}^{\infty}\delta(f-i\Delta_f)\sum_{\substack{(k,m,n) \in \mathcal{S}_i \\ (k^{\prime},m^{\prime},n^{\prime}) \in \mathcal{S}_i}}\left[\mathsf{P}\sum_{g=1}^{4}\sum_{h=1}^{H(g)}\mathsf{M}_{g}^{(h)}+2\Re\biggl\{\mathsf{P}\sum_{g=1}^{4}\sum_{h=1}^{H(g)}\mathsf{N}_{g}^{(h)}\biggr\}\right],
\end{split}
\label{eq:psd_sumcontr}
\end{align}
where $H(g)$ is the number of subsets in the partitions of $\mathcal{L}_{g}$, $g=1,2,3,4$, described in Sec.~\ref{subsec:subset_part} ($H(g)=15,10,15,1$ for $g=1,2,3,4$, resp.) and 
\begin{equation}
\mathsf{P}\triangleq\mathcal{P}_{k,m,n,k^{\prime},m^{\prime},n^{\prime}}\eta_{k,m,n}\eta^*_{k^{\prime},m^{\prime},n^{\prime}}.
\label{eq:Psf_def}
\end{equation}
In this section, we evaluate $\sum_{g=1}^{4}\sum_{h=1}^{H(g)}\mathsf{M}_{g}^{(h)}$ and $\sum_{g=1}^{4}\sum_{h=1}^{H(g)}\mathsf{N}_{g}^{(h)}$, as well as compacting the resulting expression as much as possible. 

Before we proceed with computing the above mentioned summation, we remove the Kronecker deltas in $\mathsf{M}_{g}^{(h)}$ and $\mathsf{N}_{g}^{(h)}$ corresponding to contributions in the following subspaces: i) $k=m$; ii) $n=m$; iii) $k^{\prime}=m^{\prime}$; iv) $n^{\prime}=m^{\prime}$. These contributions correspond to so-called \emph{bias terms}, i.e., they arise from a component of the field $\boldsymbol{E}_{1}(f,z)$ which is fully correlated with the transmitted field $\boldsymbol{E}_{0}(t,0)$. This component, after CDC and MF, only results in a deterministic complex scaling of the received constellation which can be easily compensated for. Thus, it does not contribute to the power of the additive zero-mean interference component we observe at the output of the MF+sampling stage once the received constellation is synchronised (in phase and amplitude) with the transmitted one. A more detailed discussion on these bias terms can be found in \cite[Appendix C]{Johannisson2013}, \cite[Appendix A]{Poggiolini2012}. Moreover, also the component $\delta_{k-m+n}\delta_{k^{\prime}-m^{\prime}+n^{\prime}}$ in $\mathsf{M}_{g}^{(h)}$ and $\mathsf{N}_{g}^{(h)}$ is removed as it only gives nonzero contribution to the PSD for the frequency $f=i\Delta_f=0$, hence its effect on the total NLI variance vanishes as we let $\Delta_f\rightarrow 0$ (see Sec.~\ref{sec:final result}). A total of 23 terms from the last columns of Tables~\ref{tab:MN1contr}, \ref{tab:MN2contr}, and \ref{tab:MN3contr} are thus removed. The remaining contributions are given in Table \ref{tab:contr_no_bias}.

{\small
\begin{longtable}{ p{.02\textwidth} | p{.02\textwidth} | p{.38\textwidth} | p{.20\textwidth} |	p{.24\textwidth}}
\caption{\scriptsize\uppercase{List of the} $\mathsf{M}_{g}^{(h)}$ \uppercase{contributions computed in} Sec.~\ref{sec:subset_classification} \uppercase{without the bias terms.}}\label{tab:contr_no_bias}\\
	\hline \hline 
	$g$ &$h$ & \textbf{Corr. terms in}  $\mathsf{M}_{g}^{(h)}$ & \textbf{Corr. terms in} $\mathsf{N}_{g}^{(h)}$ &
	\textbf{Delta products}  \\
	\hline \hline
\multirow{15}{*}{1} 
& 1 & $\mathbb{E}^{3}\{|a_x|^2\}+|\mathbb{E}\{a_xa_y^*\}|^{2}\mathbb{E}\{|a_y|^2\}$ & $\mathbb{E}\{|a_x|^2\}|\mathbb{E}\{a_xa_y^*\}|^2$& \makecell[l]{$-R_s^2\Delta_f\delta_{n-k^{\prime}}\delta_{k-m+m^{\prime}-n^{\prime}}$\\
$+2R_s\Delta_f^2\delta_{k-m+n-k^{\prime}+m^{\prime}-n^{\prime}}$} \\ \cline{2-5}

& 2 & \makecell[l]{$\mathbb{E}\{|a_x|^2\}|\mathbb{E}\{a_x^2\}|^2+|\mathbb{E}\{a_xa_y\}|^{2}\mathbb{E}\{|a_y|^2\}$} & $\mathbb{E}\{|a_x|^2\}|\mathbb{E}\{a_xa_y\}|^2$ & 
\makecell[l]{$-R_s^2\Delta_f(\delta_{k^{\prime}+n^{\prime}}\delta_{k-m+n+m^{\prime}}$ \\ $+\delta_{n+m^{\prime}}\delta_{k-m-k^{\prime}-n^{\prime}})$\\
$+2R_s\Delta_f^2\delta_{k-m+n-k^{\prime}+m^{\prime}-n^{\prime}}$}\\ \cline{2-5}

& 3 & $\mathbb{E}^{3}\{|a_x|^2\}+\mathbb{E}\{|a_x|^2\}\mathbb{E}^2\{|a_y|^2\}$ & $\mathbb{E}^2\{|a_x|^2\}\mathbb{E}\{|a_y|^2\}$ & \makecell[l]{$-R_s^2\Delta_f\delta_{n-n^{\prime}}\delta_{k-m-k^{\prime}+m^{\prime}}$\\
$+2R_s\Delta_f^2\delta_{k-m+n-k^{\prime}+m^{\prime}-n^{\prime}}$} \\ \cline{2-5}

& 4 & \makecell[l]{$\mathbb{E}\{|a_x|^2\}|\mathbb{E}\{a_x^2\}|^2+\mathbb{E}\{a_xa_y\}\mathbb{E}\{a_x^*a_y\}\mathbb{E}^*\{a_y^2\}$} & \makecell[l]{$\mathbb{E}\{a_x^2\}\mathbb{E}^*\{a_xa_y\}\mathbb{E}\{a_x^*a_y\}$} & \makecell[l]{$-R_s^2\Delta_f(\delta_{m+k^{\prime}}\delta_{k+n+m^{\prime}-n^{\prime}} $\\ $+\delta_{k+n}\delta_{m+k^{\prime}-m^{\prime}+n^{\prime}})$\\ $+2R_s\Delta_f^2\delta_{k-m+n-k^{\prime}+m^{\prime}-n^{\prime}}$} \\ \cline{2-5}

& 5 & \makecell[l]{$\mathbb{E}\{|a_x|^2\}|\mathbb{E}\{a_x^2\}|^2+|\mathbb{E}\{a_xa_y\}|^2\mathbb{E}\{|a_y|^2\}$} & \makecell[l]{$\mathbb{E}\{a_x^2\}\mathbb{E}^*\{a_xa_y\}\mathbb{E}\{a_x^*a_y\}$}
& \makecell[l]{$R_s^3\delta_{k+n}\delta_{m-m^{\prime}}\delta_{k^{\prime}+n^{\prime}}$\\
$-R_s^2\Delta_f(\delta_{k^{\prime}+n^{\prime}}\delta_{k-m+n+m^{\prime}}$\\
$+\delta_{m-m^{\prime}}\delta_{k+n-k^{\prime}-n^{\prime}}$\\
$+\delta_{k+n}\delta_{m-m^{\prime}+k^{\prime}+n^{\prime}})$\\ $+2R_s\Delta_f^2\delta_{k-m+n-k^{\prime}+m^{\prime}-n^{\prime}}$} \\ \cline{2-5}

& 6 & \makecell[l]{$\mathbb{E}\{|a_x|^2\}|\mathbb{E}\{a_x^2\}|^2+|\mathbb{E}\{a_xa_y\}|^2\mathbb{E}\{|a_y|^2\}$} & $|\mathbb{E}\{a_x^2\}|^2\mathbb{E}\{|a_y|^2\}$ 
& \makecell[l]{$-R_s^2\Delta_f(\delta_{m+n^{\prime}}\delta_{k+n-k^{\prime}+m^{\prime}}$\\ $+\delta_{k+n}\delta_{m+k^{\prime}-m^{\prime}+n^{\prime}})$ \\
$+2R_s\Delta_f^2\delta_{k-m+n-k^{\prime}+m^{\prime}-n^{\prime}}$ }\\ \cline{2-5}

& 7 & $\mathbb{E}^3\{|a_x|^2\}+|\mathbb{E}\{a_xa_y^*\}|^2\mathbb{E}\{|a_y|^2\}$ & $\mathbb{E}\{|a_x|^2\}|\mathbb{E}\{a_x a_y^*\}|^2$ 
& \makecell[l]{$-R_s^2\Delta_f\delta_{k-k^{\prime}}\delta_{m-n-m^{\prime}+n^{\prime}}$\\
$+2R_s\Delta_f^2\delta_{k-m+n-k^{\prime}+m^{\prime}-n^{\prime}}$}\\ \cline{2-5}

& 8 & $\mathbb{E}^3\{|a_x|^2\}+\mathbb{E}\{|a_x|^2\}\mathbb{E}^2\{|a_y|^2\}$ & $\mathbb{E}\{|a_x|^2\}|\mathbb{E}\{a_x a_y^*\}|^2$ 
& \makecell[l]{$R_s^3\delta_{k-k^{\prime}}\delta_{m-m^{\prime}}\delta_{n-n^{\prime}}$\\
$-R_s^2\Delta_f(\delta_{n-n^{\prime}}\delta_{k-m-k^{\prime}+m^{\prime}}$\\
$+\delta_{m-m^{\prime}}\delta_{k+n-k^{\prime}-n^{\prime}}$\\ $+\delta_{k-k^{\prime}}\delta_{m-n-m^{\prime}+n^{\prime}})$\\
$+2R_s\Delta_f^2\delta_{k-m+n-k^{\prime}+m^{\prime}-n^{\prime}}$} \\ \cline{2-5}

& 9 & \makecell[l]{$\mathbb{E}\{|a_x|^2\}|\mathbb{E}\{a_x^2\}|^2+|\mathbb{E}\{a_xa_y\}|^2\mathbb{E}\{|a_y|^2\}$} & \makecell[l]{$\mathbb{E}^*\{a_x^2\}\mathbb{E}\{a_xa_y\}\mathbb{E}\{a_xa_y^*\}$}
& \makecell[l]{$R_s^3\delta_{k-k^{\prime}}\delta_{m+n^{\prime}}\delta_{n+m^{\prime}}$\\
$-R_s^2\Delta_f(\delta_{n+m^{\prime}}\delta_{k-m-k^{\prime}-n^{\prime}}$\\
$+\delta_{m+n^{\prime}}\delta_{k+n-k^{\prime}+m^{\prime}}$\\
$+\delta_{k-k^{\prime}}\delta_{m-n-m^{\prime}+n^{\prime}})$\\
$+2R_s\Delta_f^2\delta_{k-m+n-k^{\prime}+m^{\prime}-n^{\prime}}$}\\ \cline{2-5}

& 10 & \makecell[l]{$\mathbb{E}\{|a_x|^2\}|\mathbb{E}\{a_x^2\}|^2+\mathbb{E}^*\{a_xa_y\}\mathbb{E}\{a_x a_y^*\}\mathbb{E}\{a_y^2\}$} & $\mathbb{E}\{|a_x|^2\}|\mathbb{E}\{a_xa_y\}|^2$
& \makecell[l]{$-R_s^2\Delta_f(\delta_{k^{\prime}+n^{\prime}}\delta_{k-m+n+m^{\prime}}$ \\
$+\delta_{k+m^{\prime}}\delta_{m-n+k^{\prime}+n^{\prime}})$\\
$+2R_s\Delta_f^2\delta_{k-m+n-k^{\prime}+m^{\prime}-n^{\prime}}$}\\ \cline{2-5}

& 11 & \makecell[l]{$\mathbb{E}\{|a_x|^2\}|\mathbb{E}\{a_x^2\}|^2+\mathbb{E}\{|a_x|^2\}|\mathbb{E}\{a_y^2\}|^2$} & $\mathbb{E}\{|a_x|^2\}|\mathbb{E}\{a_xa_y\}|^2$
& \makecell[l]{$R_s^3\delta_{k+m^{\prime}}\delta_{m+k^{\prime}}\delta_{n-n^{\prime}}$\\
$-R_s^2\Delta_f(\delta_{n-n^{\prime}}\delta_{k-m-k^{\prime}+m^{\prime}}$\\
$+\delta_{m+k^{\prime}}\delta_{k+n+m^{\prime}-n^{\prime}}$\\
$+\delta_{k+m^{\prime}}\delta_{m-n+k^{\prime}+n^{\prime}})$\\
$+2R_s\Delta_f^2\delta_{k-m+n-k^{\prime}+m^{\prime}-n^{\prime}}$}\\ \cline{2-5}

& 12 & \makecell[l]{$\mathbb{E}\{|a_x|^2\}|\mathbb{E}\{a_x^2\}|^2+\mathbb{E}^*\{a_xa_y\}\mathbb{E}\{a_x a_y^*\}\mathbb{E}\{a_y^2\}$} & \makecell[l]{$\mathbb{E}^*\{a_x^2\}\mathbb{E}\{a_xa_y\}\mathbb{E}\{a_x a_y^*\}$} 
& \makecell[l]{$R_s^3\delta_{k+m^{\prime}}\delta_{m+n^{\prime}}\delta_{n-k^{\prime}}$\\
$-R_s^2\Delta_f(\delta_{n-k^{\prime}}\delta_{k-m+m^{\prime}-n^{\prime}}$\\
$+\delta_{m+n^{\prime}}\delta_{k+n-k^{\prime}+m^{\prime}}$\\
$+\delta_{k+m^{\prime}}\delta_{m-n+k^{\prime}+n^{\prime}})$\\ 
$+2R_s\Delta_f^2\delta_{k-m+n-k^{\prime}+m^{\prime}-n^{\prime}}$} \\ \cline{2-5}

& 13 & $\mathbb{E}^3\{|a_x|^2\}+|\mathbb{E}\{a_x a_y^*\}|^2\mathbb{E}\{|a_y|^2\}$ & $\mathbb{E}^2\{|a_x|^2\}\mathbb{E}\{|a_y|^2\}$
& \makecell[l]{$-R_s^2\Delta_f\delta_{k-n^{\prime}}\delta_{m-n+k^{\prime}-m^{\prime}}$\\
$+2R_s\Delta_f^2\delta_{k-m+n-k^{\prime}+m^{\prime}-n^{\prime}}$} \\ \cline{2-5}

& 14 & \makecell[l]{$\mathbb{E}\{|a_x|^2\}|\mathbb{E}\{a_x^2\}|^2+\mathbb{E}\{a_xa_y\}\mathbb{E}\{a_x^*a_y\}\mathbb{E}^*\{a_y^2\}$} & $\mathbb{E}\{|a_x|^2\}|\mathbb{E}\{a_xa_y\}|^2$ 
& \makecell[l]{$R_s^3\delta_{k-n^{\prime}}\delta_{m+k^{\prime}}\delta_{n+m^{\prime}}$\\ 
$-R_s^2\Delta_f(\delta_{n+m^{\prime}}\delta_{k-m-k^{\prime}-n^{\prime}}$\\
$+\delta_{m+k^{\prime}}\delta_{k+n+m^{\prime}-n^{\prime}}$\\
$+\delta_{k-n^{\prime}}\delta_{m-n+k^{\prime}-m^{\prime}})$\\
$+2R_s\Delta_f^2\delta_{k-m+n-k^{\prime}+m^{\prime}-n^{\prime}}$}\\ \cline{2-5}

&15 & $\mathbb{E}^3\{|a_x|^2\}+|\mathbb{E}\{a_x a_y^*\}|^2\mathbb{E}\{|a_y|^2\}$ & $\mathbb{E}\{|a_x|^2\}|\mathbb{E}\{a_x a_y^*\}|^2$ 
& \makecell[l]{$R_s^3\delta_{k-n^{\prime}}\delta_{m-m^{\prime}}\delta_{n-k^{\prime}}$\\
$-R_s^2\Delta_f(\delta_{n-k^{\prime}}\delta_{k-m+m^{\prime}-n^{\prime}}$\\
$+\delta_{m-m^{\prime}}\delta_{k+n-k^{\prime}-n^{\prime}}$\\
$+\delta_{k-n^{\prime}}\delta_{m-n+k^{\prime}-m^{\prime}})$\\
$+2R_s\Delta_f^2\delta_{k-m+n-k^{\prime}+m^{\prime}-n^{\prime}}$}\\
\hline\hline
\multirow{10}{*}{2}
& 1& $|\mathbb{E}\{a_x|a_x|^2\}|^2+|\mathbb{E}\{a_x|a_y|^2\}|^2$ & $\mathbb{E}\{a_x |a_x|^2\}\mathbb{E}\{a_x |a_y|^2\}$ & 
$-R_s\Delta_f^2\delta_{k-m+n-k^{\prime}+m^{\prime}-n^{\prime}}$\\ \cline{2-5}

& 2 & \makecell[l]{$|\mathbb{E}\{a_x|a_x|^2\}|^2+\mathbb{E}\{a_y^* |a_y|^2\}\mathbb{E}\{|a_x|^2a_y\}$} & $|\mathbb{E}\{ |a_x|^2a_y\}|^2$ & 
\makecell[l]{$R_s^2\Delta_f\delta_{k-m-k^{\prime}}\delta_{n+m^{\prime}-n^{\prime}}$\\$-R_s\Delta_f^2\delta_{k-m+n-k^{\prime}+m^{\prime}-n^{\prime}}$ } \\ \cline{2-5}

& 3 & \makecell[l]{$|\mathbb{E}\{a_x|a_x|^2\}|^2+\mathbb{E}\{|a_x|^2a_y^*\}\mathbb{E}\{a_y |a_y|^2\}$}& $|\mathbb{E}\{ |a_x|^2a_y\}|^2$ & 
\makecell[l]{$R_s^2\Delta_f\delta_{k-m+m^{\prime}}\delta_{n-k^{\prime}-n^{\prime}}$\\$-R_s\Delta_f^2\delta_{k-m+n-k^{\prime}+m^{\prime}-n^{\prime}}$} \\ \cline{2-5}

& 4 &$|\mathbb{E}\{a_x|a_x|^2\}|^2+|\mathbb{E}\{a_x|a_y|^2\}|^2$ & $\mathbb{E}\{a_x^*|a_x|^2\}\mathbb{E}\{a_x |a_y|^2\}$ & 
\makecell[l]{$R_s^2\Delta_f\delta_{k-m-n^{\prime}}\delta_{n-k^{\prime}+m^{\prime}}$\\ $-R_s\Delta_f^2\delta_{k-m+n-k^{\prime}+m^{\prime}-n^{\prime}}$} \\ \cline{2-5}

& 5 & $|\mathbb{E}\{a_x|a_x|^2\}|^2+|\mathbb{E}\{a_x|a_y|^2\}|^2$ & $|\mathbb{E}\{a_x^2a_y^*\}|^2$ & 
\makecell[l]{$R_s^2\Delta_f\delta_{k+n-k^{\prime}}\delta_{m-m^{\prime}+n^{\prime}}$\\$-R_s\Delta_f^2\delta_{k-m+n-k^{\prime}+m^{\prime}-n^{\prime}}$} \\ \cline{2-5}

& 6 & $|\mathbb{E}\{a_x^3\}|^2+|\mathbb{E}\{a_xa_y^2\}|^2$ & $|\mathbb{E}\{a_x^2a_y\}|^2$ & 
\makecell[l]{$R_s^2\Delta_f\delta_{k+n+m^{\prime}}\delta_{m+k^{\prime}+n^{\prime}}$\\$-R_s\Delta_f^2\delta_{k-m+n-k^{\prime}+m^{\prime}-n^{\prime}}$} \\ \cline{2-5}

& 7 & \makecell[l]{$|\mathbb{E}\{a_x|a_x|^2\}|^2+\mathbb{E}\{|a_x|^2a_y\}\mathbb{E}\{a_y^*|a_y|^2\}$ }& $\mathbb{E}\{a_x |a_x|^2\}\mathbb{E}\{a_x^*|a_y|^2\}$ & 
\makecell[l]{$R_s^2\Delta_f\delta_{k+n-n^{\prime}}\delta_{m+k^{\prime}-m^{\prime}}$\\$-R_s\Delta_f^2\delta_{k-m+n-k^{\prime}+m^{\prime}-n^{\prime}}$} \\ \cline{2-5}

& 8 & \makecell[l]{$|\mathbb{E}\{a_x|a_x|^2\}|^2+\mathbb{E}\{ |a_x|^2a_y^*\}\mathbb{E}\{a_y |a_y|^2\}$} & $\mathbb{E}\{a_x |a_y|^2\}\mathbb{E}\{a_x^*|a_x|^2\}$ & 
\makecell[l]{$R_s^2\Delta_f\delta_{k-k^{\prime}+m^{\prime}}\delta_{m-n+n^{\prime}}$\\$-R_s\Delta_f^2\delta_{k-m+n-k^{\prime}+m^{\prime}-n^{\prime}}$} \\ \cline{2-5}

& 9 &$|\mathbb{E}\{a_x|a_x|^2\}|^2+|\mathbb{E}\{a_x|a_y|^2\}|^2$ & $|\mathbb{E}\{|a_x|^2a_y\}|^2$ & 
\makecell[l]{$R_s^2\Delta_f\delta_{k-k^{\prime}-n^{\prime}}\delta_{m-n-m^{\prime}}$\\$-R_s\Delta_f^2\delta_{k-m+n-k^{\prime}+m^{\prime}-n^{\prime}}$} \\ \cline{2-5}

& 10 & $|\mathbb{E}\{a_x|a_x|^2\}|^2+|\mathbb{E}\{a_x^*a_y^2\}|^2$ & $|\mathbb{E}\{|a_x|^2a_y\}|^2$ & 
\makecell[l]{$R_s^2\Delta_f\delta_{k+m^{\prime}-n^{\prime}}\delta_{m-n+k^{\prime}}$\\$-R_s\Delta_f^2\delta_{k-m+n-k^{\prime}+m^{\prime}-n^{\prime}}$} \\ 
\hline\hline
\multirow{15}{*}{3}
& 1 & \makecell[l]{$\mathbb{E}\{|a_x|^4\}\mathbb{E}\{|a_x|^2\}+\mathbb{E}\{|a_x|^2 |a_y|^2\}\mathbb{E}\{|a_y|^2\}$} 
& $\mathbb{E}\{|a_x|^2\}\mathbb{E}\{|a_x|^2 |a_y|^2\}$ & 
$-R_s\Delta_f^2\delta_{k-m+n-k^{\prime}+m^{\prime}-n^{\prime}}$ \\ \cline{2-5}

& 2 & \makecell[l]{$\mathbb{E}^*\{a_x^2|a_x|^2\}\mathbb{E}\{a_x^2\}+\mathbb{E}\{a_xa_y\}\mathbb{E}^*\{a_xa_y|a_y|^2\}$} & $\mathbb{E}\{a_x^2\}\mathbb{E}^*\{a_x^2 |a_y|^2\}$ & 
\makecell[l]{$R_s^2\Delta_f\delta_{k+n}\delta_{m+k^{\prime}-m^{\prime}+n^{\prime}}$\\$-R_s\Delta_f^2\delta_{k-m+n-k^{\prime}+m^{\prime}-n^{\prime}}$} \\ \cline{2-5}

& 3 & \makecell[l]{$\mathbb{E}\{|a_x|^4\}\mathbb{E}\{|a_x|^2\}+\mathbb{E}\{|a_x|^2|a_y|^2\}\mathbb{E}\{|a_y|^2\}$} & $\mathbb{E}\{a_x a_y^*\}\mathbb{E}\{a_x^*a_y|a_x|^2\}$  & 
\makecell[l]{$R_s^2\Delta_f\delta_{k-k^{\prime}}\delta_{m-n-m^{\prime}+n^{\prime}}$\\$-R_s\Delta_f^2\delta_{k-m+n-k^{\prime}+m^{\prime}-n^{\prime}}$} \\ \cline{2-5}

& 4 & \makecell[l]{$\mathbb{E}^*\{a_x^2|a_x|^2\}\mathbb{E}\{a_x^2\}+\mathbb{E}^*\{|a_x|^2a_y^2\}\mathbb{E}\{a_y^2\}$} & $\mathbb{E}\{a_xa_y\}\mathbb{E}^*\{a_xa_y|a_x|^2\}$ & 
\makecell[l]{$R_s^2\Delta_f\delta_{k+m^{\prime}}\delta_{m-n+k^{\prime}+n^{\prime}}$\\$-R_s\Delta_f^2\delta_{k-m+n-k^{\prime}+m^{\prime}-n^{\prime}}$} \\ \cline{2-5}

& 5 & \makecell[l]{$\mathbb{E}\{|a_x|^4\}\mathbb{E}\{|a_x|^2\}+\mathbb{E}\{a_x^*a_y\}\mathbb{E}\{a_x a_y^*|a_y|^2\}$ } & $\mathbb{E}\{|a_x|^2\}\mathbb{E}\{|a_x|^2|a_y|^2\}$ & 
\makecell[l]{$R_s^2\Delta_f\delta_{k-n^{\prime}}\delta_{m-n+k^{\prime}-m^{\prime}}$\\$-R_s\Delta_f^2\delta_{k-m+n-k^{\prime}+m^{\prime}-n^{\prime}}$} \\ \cline{2-5}

& 6 & \makecell[l]{$\mathbb{E}\{|a_x|^4\}\mathbb{E}\{|a_x|^2\}+\mathbb{E}\{a_x a_y^*\}\mathbb{E}\{a_x^*a_y|a_y|^2\}$} & $\mathbb{E}\{|a_x|^2\}\mathbb{E}\{|a_x|^2|a_y|^2\}$ & 
$-R_s\Delta_f^2\delta_{k-m+n-k^{\prime}+m^{\prime}-n^{\prime}}$\\ \cline{2-5}

& 7 & \makecell[l]{$\mathbb{E}\{a_x^2|a_x|^2\}\mathbb{E}^*\{a^2_x\}+\mathbb{E}\{|a_x|^2 a_y^2\}\mathbb{E}^*\{a_y^2\}$ } & $\mathbb{E}^*\{a_xa_y\}\mathbb{E}\{a_xa_y|a_x|^2\}$ & 
\makecell[l]{$R_s^2\Delta_f\delta_{m+k^{\prime}}\delta_{k-n+m^{\prime}-n^{\prime}}$\\$-R_s\Delta_f^2\delta_{k-m+n-k^{\prime}+m^{\prime}-n^{\prime}}$} \\ \cline{2-5}

& 8 & \makecell[l]{$\mathbb{E}\{|a_x|^4\}\mathbb{E}\{|a_x|^2\}+\mathbb{E}\{|a_x|^2 |a_y|^2\}\mathbb{E}\{ |a_y|^2\}$} & $\mathbb{E}\{a_x^*a_y\}\mathbb{E}\{a_x a_y^*|a_x|^2\}$ & 
\makecell[l]{$R_s^2\Delta_f\delta_{m-m^{\prime}}\delta_{k+n-k^{\prime}-n^{\prime}}$\\$-R_s\Delta_f^2\delta_{k-m+n-k^{\prime}+m^{\prime}-n^{\prime}}$} \\ \cline{2-5}

& 9 & \makecell[l]{$\mathbb{E}\{a_x^2|a_x|^2\}\mathbb{E}^*\{a_x^2\}+\mathbb{E}^*\{a_xa_y\}\mathbb{E}\{a_xa_y|a_y|^2\}$} & $\mathbb{E}^*\{a^2_x\}\mathbb{E}\{a^2_x |a_y|^2\}$ & 
\makecell[l]{$R_s^2\Delta_f\delta_{m+n^{\prime}}\delta_{k+n-k^{\prime}+m^{\prime}}$\\$-R_s\Delta_f^2\delta_{k-m+n-k^{\prime}+m^{\prime}-n^{\prime}}$}  \\ \cline{2-5}

& 10 & \makecell[l]{$\mathbb{E}\{|a_x|^4\}\mathbb{E}\{|a_x|^2\}+\mathbb{E}\{a_x a_y^*\}\mathbb{E}\{a_x^*a_y|a_y|^2\}$} &
$\mathbb{E}\{a_x a_y^*\}\mathbb{E}\{a_x^*a_y |a_x|^2\}$ & 
\makecell[l]{$R_s^2\Delta_f\delta_{n-k^{\prime}}\delta_{k-m+m^{\prime}-n^{\prime}}$\\$-R_s\Delta_f^2\delta_{k-m+n-k^{\prime}+m^{\prime}-n^{\prime}}$} \\ \cline{2-5}

& 11 & \makecell[l]{$\mathbb{E}^*\{a_x^2|a_x|^2\}\mathbb{E}\{a_x^2\}+\mathbb{E}\{a_xa_y\}\mathbb{E}^*\{a_xa_y|a_y|^2\}$} & $\mathbb{E}\{a_xa_y\}\mathbb{E}^*\{a_xa_y|a_x|^2\}$  & 
\makecell[l]{$R_s^2\Delta_f\delta_{n+m^{\prime}}\delta_{k-m-k^{\prime}-n^{\prime}}$\\$-R_s\Delta_f^2\delta_{k-m+n-k^{\prime}+m^{\prime}-n^{\prime}}$} \\ \cline{2-5}

& 12 & \makecell[l]{$\mathbb{E}\{|a_x|^4\}\mathbb{E}\{|a_x|^2\}+\mathbb{E}\{|a_x|^2\}\mathbb{E}\{|a_y|^4\}$} & $\mathbb{E}\{|a_x|^2\}\mathbb{E}\{|a_x|^2|a_y|^2\}$& 
\makecell[l]{$R_s^2\Delta_f\delta_{n-n^{\prime}}\delta_{k-m-k^{\prime}+m^{\prime}}$\\$-R_s\Delta_f^2\delta_{k-m+n-k^{\prime}+m^{\prime}-n^{\prime}}$} \\ \cline{2-5}

& 13 & \makecell[l]{$\mathbb{E}\{|a_x|^4\}\mathbb{E}\{|a_x|^2\}+\mathbb{E}\{|a_x|^2|a_y|^2\}\mathbb{E}\{|a_y|^2\}$ } & $\mathbb{E}\{|a_x|^4\}\mathbb{E}\{|a_y|^2\}$ & 
$-R_s\Delta_f^2\delta_{k-m+n-k^{\prime}+m^{\prime}-n^{\prime}}$ \\ \cline{2-5}

& 14 & \makecell[l]{$\mathbb{E}\{a_x^2|a_x|^2\}\mathbb{E}^*\{a_x^2\}+\mathbb{E}^*\{a_xa_y\}\mathbb{E}\{a_xa_y|a_y|^2\}$} & $\mathbb{E}^*\{a_xa_y\}\mathbb{E}\{a_xa_y|a_x|^2\}$ & 
\makecell[l]{$R_s^2\Delta_f\delta_{k^{\prime}+n^{\prime}}\delta_{k-m+n+m^{\prime}}$\\$-R_s\Delta_f^2\delta_{k-m+n-k^{\prime}+m^{\prime}-n^{\prime}}$} \\ \cline{2-5}

& 15 & \makecell[l]{$\mathbb{E}\{|a_x|^4\}\mathbb{E}\{|a_x|^2\}+\mathbb{E}\{a_x^*a_y\}\mathbb{E}\{a_xa_y^*|a_y|^2\}$} & $\mathbb{E}\{a_x^*a_y\}\mathbb{E}\{a_xa_y^*|a_x|^2$ & 
$-R_s\Delta_f^2\delta_{k-m+n-k^{\prime}+m^{\prime}-n^{\prime}}$  \\ 
\hline\hline
4 & 1 & $\mathbb{E}\{|a_x|^6\}+\mathbb{E}\{|a_x|^2|a_y|^4\}$ & $\mathbb{E}\{|a_x|^4 |a_y|^2\}$ & $R_s\Delta_f^2\delta_{k-m+n-k^{\prime}+m^{\prime}-n^{\prime}}$\\
\hline \hline
\end{longtable}
}

We now compact the contributions in Table \ref{tab:contr_no_bias} by grouping the Kronecker delta products based on each correlation term they multiply. We use three pairs of curly brackets $\{\cdot\}$ to denote the terms multiplying $R_s^3$, $R_s^2$, and $R_s$. The list of all Kronecker delta products multiplying each correlation term is shown in Table~\ref{tab:kron_delta}. The correlation terms are divided into intra-polarisation (expectations containing only $a_x$) and cross-polarisation terms (expectations containing both $a_x$ and $a_y$). Moreover, the correlations are categorised based on the specific contribution (either $\mathsf{M}$ or $\mathsf{N}$) in \eqref{eq:psd_sumcontr} they belong to. 

{\small
\begin{longtable}{p{.2\textwidth} | p{.75\textwidth} }
\caption{\scriptsize\uppercase{List of Kronecker delta contributions ordered by the corresponding correlation term.}}\label{tab:kron_delta}\\
\hline \hline
\textbf{Correlation terms} & \textbf{Kronecker delta products} \\ 
\hline \hline
\multicolumn{2}{c}{\textbf{Intra-polarisation terms}} \\
\hline 
\multicolumn{1}{c|}{\textbf{In $\mathsf{M}_{g}^{(h)}$}} \\
\hline 
$\mathbb{E}^3\{|a_x|^2\}$  &      \{$\delta_{k-k^{\prime}}\delta_{m-m^{\prime}}\delta_{n-n^{\prime}},\delta_{k-n^{\prime}}\delta_{m-m^{\prime}}\delta_{n-k^{\prime}}$\},\{$-2\delta_{n-k^{\prime}}\delta_{k-m+m^{\prime}-n^{\prime}},-2\delta_{n-n^{\prime}}\delta_{k-m-k^{\prime}+m^{\prime}},$\\
&$-2\delta_{k-k^{\prime}}\delta_{m-n-m^{\prime}+n^{\prime}},-2\delta_{m-m^{\prime}}\delta_{k+n-k^{\prime}-n^{\prime}},-2\delta_{k-n^{\prime}}\delta_{m-n+k^{\prime}-m^{\prime}}$\},\{$12\delta_{k-m+n-k^{\prime}+m^{\prime}-n^{\prime}}$\} \\
\hline

$\mathbb{E}\{|a_x|^2\}|\mathbb{E}\{a_x^2\}|^2$ & \{$\delta_{k+n}\delta_{m-m^{\prime}}\delta_{k^{\prime}+n^{\prime}},\delta_{k-k^{\prime}}\delta_{m+n^{\prime}}\delta_{n+m^{\prime}},\delta_{k+m^{\prime}}\delta_{m+k^{\prime}}\delta_{n-n^{\prime}},\delta_{k+m^{\prime}}\delta_{m+n^{\prime}}\delta_{n-k^{\prime}},$\\
& $\delta_{k-n^{\prime}}\delta_{m+k^{\prime}}\delta_{n+m^{\prime}}$\},\{$-\delta_{n-k^{\prime}}\delta_{k-m+m^{\prime}-n^{\prime}},-3\delta_{m+k^{\prime}}\delta_{k+n+m^{\prime}-n^{\prime}},-3\delta_{k+n}\delta_{m+k^{\prime}-m^{\prime}+n^{\prime}}$\\
& $-3\delta_{k^{\prime}+n^{\prime}}\delta_{k+n-m+m^{\prime}},-\delta_{m-m^{\prime}}\delta_{k+n-k^{\prime}-n^{\prime}},-3\delta_{m+n^{\prime}}\delta_{k+n-k^{\prime}+m^{\prime}},-3\delta_{n+m^{\prime}}\delta_{k-m-k^{\prime}-n^{\prime}},$\\
&$-\delta_{k-k^{\prime}}\delta_{m-n-m^{\prime}+n^{\prime}},-3\delta_{k+m^{\prime}}\delta_{m-n+k^{\prime}+n^{\prime}},-\delta_{n-n^{\prime}}\delta_{k-m-k^{\prime}+m^{\prime}},-\delta_{k-n^{\prime}}\delta_{m-n+k^{\prime}-m^{\prime}},$\}\\
&\{$18\delta_{k-m+n-k^{\prime}+m^{\prime}-n^{\prime}}$\} \\
\hline

$|\mathbb{E}\{a_x|a_x|^2\}|^2$                 & \{\}, \{$\delta_{k-m-k^{\prime}}\delta_{n+m^{\prime}-n^{\prime}},\delta_{k-m+m^{\prime}}\delta_{n-k^{\prime}-n^{\prime}},\delta_{k-m-n^{\prime}}\delta_{n-k^{\prime}+m^{\prime}},\delta_{k+n-k^{\prime}}\delta_{m-m^{\prime}+n^{\prime}},$\\
&$\delta_{k+n-n^{\prime}}\delta_{m+k^{\prime}-m^{\prime}},\delta_{k-k^{\prime}+m^{\prime}}\delta_{m-n+n^{\prime}},\delta_{k-k^{\prime}-n^{\prime}}\delta_{m-n+m^{\prime}},\delta_{k+m^{\prime}-n^{\prime}}\delta_{m-n+k^{\prime}}$\},\\
&\{$-9\delta_{k-m+n-k^{\prime}+m^{\prime}-n^{\prime}}$\} \\
\hline

$|\mathbb{E}\{a_x^3\}|^2$ & \{\}, $\{\delta_{k+n+m^{\prime}}\delta_{m+k^{\prime}+n^{\prime}}\},\{-\delta_{k-m+n-k^{\prime}+m^{\prime}-n^{\prime}}\}$ \\
\hline

$\mathbb{E}\{|a_x|^4\}\mathbb{E}\{|a_x|^2\}$ & \{\}, $\{\delta_{k-k^{\prime}}\delta_{m-n-m^{\prime}+n^{\prime}},\delta_{k-n^{\prime}}\delta_{m-n+k^{\prime}-m^{\prime}},\delta_{m-m^{\prime}}\delta_{k+n-k^{\prime}-n^{\prime}},\delta_{n-k^{\prime}}\delta_{k-m+m^{\prime}-n^{\prime}},$ \\& $\delta_{n-n^{\prime}}\delta_{k-m-k^{\prime}+m^{\prime}}\},\{-9\delta_{k-m+n-k^{\prime}+m^{\prime}-n^{\prime}}\}$ \\
\hline

$\mathbb{E}^*\{a_x^2|a_x|^2\}\mathbb{E}\{a_x^2\}$  & \{\}, $\{\delta_{k+n}\delta_{m+k^{\prime}-m^{\prime}+n^{\prime}},\delta_{k+m^{\prime}}\delta_{m-n+k^{\prime}+n^{\prime}},\delta_{n+m^{\prime}}\delta_{k-m-k^{\prime}-n^{\prime}}\}$,\\
&$\{-3\delta_{k-m+n-k^{\prime}+m^{\prime}-n^{\prime}}\}$\\
\hline

$\mathbb{E}\{a_x^2|a_x|^2\}\mathbb{E}^*\{a_x^2\}$  & \{\}, $\{\delta_{m+k^{\prime}}\delta_{k-n+m^{\prime}-n^{\prime}},\delta_{m+n^{\prime}}\delta_{k+n-k^{\prime}+m^{\prime}},\delta_{k^{\prime}+n^{\prime}}\delta_{k-m+n+m^{\prime}}\},\{-3\delta_{k-m+n-k^{\prime}+m^{\prime}-n^{\prime}}\}$ \\
\hline

$\mathbb{E}\{|a_x|^6\}$ & $\{\}$, $\{\}$, $\{\delta_{k-m+n-k^{\prime}+m^{\prime}-n^{\prime}}\}$ \\
\hline \hline

\multicolumn{2}{c}{\textbf{Cross-polarisation terms}}     \\
\hline
\multicolumn{1}{c|}{\textbf{In $\mathsf{M}_{g}^{(h)}$}} \\
\hline 

$\mathbb{E}\{|a_x|^2\}\mathbb{E}^2\{|a_y|^2\}$ & $\{\delta_{k-k^{\prime}}\delta_{m-m^{\prime}}\delta_{n-n^{\prime}}\},\{-2\delta_{n-n^{\prime}}\delta_{k-m-k^{\prime}+m^{\prime}},-\delta_{m-m^{\prime}}\delta_{k+n-k^{\prime}-n^{\prime}},-\delta_{k-k^{\prime}}\delta_{m-n-m^{\prime}+n^{\prime}}\},$\\
&$\{4\delta_{k-m+n-k^{\prime}+m^{\prime}-n^{\prime}}\}$ \\
\hline

$\mathbb{E}\{|a_x|^2\}|\mathbb{E}\{a_y^2\}|^2$ & $\{\delta_{k+m^{\prime}}\delta_{m+k^{\prime}}\delta_{n-n^{\prime}}\},\{-\delta_{n-n^{\prime}}\delta_{k-m-k^{\prime}+m^{\prime}},-\delta_{m+k^{\prime}}\delta_{k+n+m^{\prime}-n^{\prime}},-\delta_{k+m^{\prime}}\delta_{m-n+k^{\prime}+n^{\prime}}\},$ \\
& $\{2\delta_{k-m+n-k^{\prime}+m^{\prime}-n^{\prime}}\} $\\
\hline

$\mathbb{E}\{|a_x|^2\}\mathbb{E}\{|a_y|^4\}$    & \{\}, \{$\delta_{n-n^{\prime}}\delta_{k-m-k^{\prime}+m^{\prime}}\},\{-\delta_{k-m+n-k^{\prime}+m^{\prime}-n^{\prime}}\}$ \\
\hline

$\mathbb{E}\{|a_x|^2a_y^*\}\mathbb{E}\{a_y |a_y|^2\}$ &  \{\}, $\{\delta_{k-m+m^{\prime}}\delta_{n-k^{\prime}-n^{\prime}},\delta_{k-k^{\prime}+m^{\prime}}\delta_{m-n+n^{\prime}}\},\{-2\delta_{k-m+n-k^{\prime}+m^{\prime}-n^{\prime}}\}$\\
\hline

$\mathbb{E}\{|a_x|^2|a_y|^2\}\mathbb{E}\{|a_y|^2\}$ &  \{\}, $\{\delta_{k-k^{\prime}}\delta_{m-n-m^{\prime}+n^{\prime}},\delta_{m-m^{\prime}}\delta_{k+n-k^{\prime}-n^{\prime}}\},\{-4\delta_{k-m+n-k^{\prime}+m^{\prime}-n^{\prime}}\}$\\
\hline

$\mathbb{E}\{a_xa_y\}\mathbb{E}^*\{a_xa_y|a_x|^2\}$ &  \{\}, $\{\delta_{k+m^{\prime}}\delta_{m-n+k^{\prime}+n^{\prime}},\delta_{m+k^{\prime}}\delta_{k-n+m^{\prime}-n^{\prime}},\delta_{n+m^{\prime}}\delta_{k-m-k^{\prime}-n^{\prime}},\delta_{k^{\prime}+n^{\prime}}\delta_{k-m+n+m^{\prime}}\},$\\
& $\{-4\delta_{k-m+n-k^{\prime}+m^{\prime}-n^{\prime}}\}$\\
\hline

$\mathbb{E}^*\{|a_x|^2a_y^2\}\mathbb{E}\{a_y^2\}$ &  \{\}, $\{\delta_{k+m^{\prime}}\delta_{m-n+k^{\prime}+n^{\prime}}\},\{-\delta_{k-m+n-k^{\prime}+m^{\prime}-n^{\prime}}\}$ \\
\hline

$\mathbb{E}\{|a_x|^2 a_y^2\}\mathbb{E}^*\{a_y^2\}$  &  \{\},  $\{\delta_{m+k^{\prime}}\delta_{k-n+m^{\prime}-n^{\prime}}\},\{-\delta_{k-m+n-k^{\prime}+m^{\prime}-n^{\prime}}\}$ \\
\hline

$\mathbb{E}\{|a_x|^2|a_y|^4\}$ & \{\}, \{\}, $\{\delta_{k-m+n-k^{\prime}+m^{\prime}-n^{\prime}}\}$ \\
\hline

$|\mathbb{E}\{a_xa_y^*\}|^{2}\mathbb{E}\{|a_y|^2\}$ & $\{\delta_{k-n^{\prime}}\delta_{m-m^{\prime}}\delta_{n-k^{\prime}}\},\{-2\delta_{n-k^{\prime}}\delta_{k-m+m^{\prime}-n^{\prime}},-2\delta_{k-n^{\prime}}\delta_{m-n+k^{\prime}-m^{\prime}},-\delta_{k-k^{\prime}}\delta_{m-n-m^{\prime}+n^{\prime}},$\\
& $-\delta_{m-m^{\prime}}\delta_{k+n-k^{\prime}-n^{\prime}}\},\{8\delta_{k-m+n-k^{\prime}+m^{\prime}-n^{\prime}}\}$\\
\hline

$\mathbb{E}\{a_xa_y\}\mathbb{E}\{a_x^*a_y\}\mathbb{E}^*\{a_y^2\}$ & $\{\delta_{k-n^{\prime}}\delta_{m+k^{\prime}}\delta_{n+m^{\prime}}\},\{-2\delta_{m+k^{\prime}}\delta_{k+n+m^{\prime}-n^{\prime}},-\delta_{k+n}\delta_{m+k^{\prime}-m^{\prime}+n^{\prime}},-\delta_{n+m^{\prime}}\delta_{k-m-k^{\prime}-n^{\prime}},$ \\
& $-\delta_{k-n^{\prime}}\delta_{m-n+k^{\prime}-m^{\prime}}\},\{4\delta_{k-m+n-k^{\prime}+m^{\prime}-n^{\prime}}\}$\\
\hline

$\mathbb{E}^*\{a_xa_y\}\mathbb{E}\{a_x a_y^*\}\mathbb{E}\{a_y^2\}$ & $\{\delta_{k+m^{\prime}}\delta_{m+n^{\prime}}\delta_{n-k^{\prime}}\},\{-2\delta_{k+m^{\prime}}\delta_{m-n+k^{\prime}+n^{\prime}},-\delta_{k^{\prime}+n^{\prime}}\delta_{k-m+n+m^{\prime}},-\delta_{n-k^{\prime}}\delta_{k-m+m^{\prime}-n^{\prime}},$\\
& $-\delta_{m+n^{\prime}}\delta_{k+n-k^{\prime}+m^{\prime}}\},\{4\delta_{k-m+n-k^{\prime}+m^{\prime}-n^{\prime}}\}$\\
\hline

$|\mathbb{E}\{a_xa_y\}|^2\mathbb{E}\{|a_y|^2\}$ & \{$\delta_{k+n}\delta_{m-m^{\prime}}\delta_{k^{\prime}+n^{\prime}},\delta_{k-k^{\prime}}\delta_{m+n^{\prime}}\delta_{n+m^{\prime}}\},\{-2\delta_{k^{\prime}+n^{\prime}}\delta_{k-m+n+m^{\prime}},-2\delta_{k+n}\delta_{m-m^{\prime}+k^{\prime}+n^{\prime}},$ \\ 
& $-2\delta_{n+m^{\prime}}\delta_{k-m-k^{\prime}-n^{\prime}},-2\delta_{m+n^{\prime}}\delta_{k+n-k^{\prime}+m^{\prime}},-\delta_{m-m^{\prime}}\delta_{k+n-k^{\prime}-n^{\prime}},-\delta_{k-k^{\prime}}\delta_{m-n-m^{\prime}+n^{\prime}}\}$,\\
& $\{8\delta_{k-m+n-k^{\prime}+m^{\prime}-n^{\prime}}\}$
\\
\hline

$|\mathbb{E}\{a_x|a_y|^2\}|^2$ &  \{\}, $\{\delta_{k-m-n^{\prime}}\delta_{n-k^{\prime}+m^{\prime}},\delta_{k+n-k^{\prime}}\delta_{m-m^{\prime}+n^{\prime}},\delta_{k-k^{\prime}-n^{\prime}}\delta_{m-n+m^{\prime}}\},\{-4\delta_{k-m+n-k^{\prime}+m^{\prime}-n^{\prime}}\}$
\\
\hline

$|\mathbb{E}\{a_xa_y^2\}|^2$ &  \{\}, $\{\delta_{k+n+m^{\prime}}\delta_{m+k^{\prime}+n^{\prime}}\},\{-\delta_{k-m+n-k^{\prime}+m^{\prime}-n^{\prime}}\}$\\
\hline

$|\mathbb{E}\{a_x^*a_y^2\}|^2$ &  \{\}, $\{\delta_{k+m^{\prime}-n^{\prime}}\delta_{m-n+k^{\prime}}\},\{-\delta_{k-m+n-k^{\prime}+m^{\prime}-n^{\prime}}\}$\\
\hline

\multicolumn{1}{c|}{\textbf{In $\mathsf{N}_{g}^{(h)}$} } \\
\hline 

$\mathbb{E}\{|a_x|^2\}|\mathbb{E}\{a_xa_y^*\}|^2$                 & \{$\delta_{k-k^{\prime}}\delta_{m-m^{\prime}}\delta_{n-n^{\prime}},\delta_{k-n^{\prime}}\delta_{m-m^{\prime}}\delta_{n-k^{\prime}}\},\{-2\delta_{n-k^{\prime}}\delta_{k-m+m^{\prime}-n^{\prime}},-2\delta_{k-k^{\prime}}\delta_{m-n-m^{\prime}+n^{\prime}},$\\
&$-2\delta_{m-m^{\prime}}\delta_{k+n-k^{\prime}-n^{\prime}},-\delta_{n-n^{\prime}}\delta_{k-m-k^{\prime}+m^{\prime}},-\delta_{k-n^{\prime}}\delta_{m-n+k^{\prime}-m^{\prime}}\},\{8\delta_{k-m+n-k^{\prime}+m^{\prime}-n^{\prime}}\}$ \\
\hline

$\mathbb{E}\{|a_x|^2\}|\mathbb{E}\{a_xa_y\}|^2$ & $\{\delta_{k+m^{\prime}}\delta_{m+k^{\prime}}\delta_{n-n^{\prime}},\delta_{k-n^{\prime}}\delta_{m+k^{\prime}}\delta_{n+m^{\prime}}\},\{-2\delta_{k^{\prime}+n^{\prime}}\delta_{k-m+n+m^{\prime}},-2\delta_{n+m^{\prime}}\delta_{k-m-k^{\prime}-n^{\prime}},$\\
&$-2\delta_{k+m^{\prime}}\delta_{m-n+k^{\prime}+n^{\prime}},-2\delta_{m+k^{\prime}}\delta_{k+n+m^{\prime}-n^{\prime}},-\delta_{n-n^{\prime}}\delta_{k-m-k^{\prime}+m^{\prime}},-\delta_{k-n^{\prime}}\delta_{m-n+k^{\prime}-m^{\prime}}\}$, \\
&$\{8\delta_{k-m+n-k^{\prime}+m^{\prime}-n^{\prime}}\}.$ \\
\hline

$\mathbb{E}^2\{|a_x|^2\}\mathbb{E}\{|a_y|^2\}$ & \{\}, $\{-\delta_{n-n^{\prime}}\delta_{k-m-k^{\prime}+m^{\prime}},-\delta_{k-n^{\prime}}\delta_{m-n+k^{\prime}-m^{\prime}}\},\{4\delta_{k-m+n-k^{\prime}+m^{\prime}-n^{\prime}}\}$\\
\hline

$\mathbb{E}\{|a_x|^2\}\mathbb{E}\{|a_x|^2 |a_y|^2\}$ & \{\}, $\{\delta_{k-n^{\prime}}\delta_{m-n+k^{\prime}-m^{\prime}},\delta_{n-n^{\prime}}\delta_{k-m-k^{\prime}+m^{\prime}}\},\{-4\delta_{k-m+n-k^{\prime}+m^{\prime}-n^{\prime}}$\} \\
\hline

$\mathbb{E}\{|a_x|^4\}\mathbb{E}\{|a_y|^2\}$ & \{\},  \{\}, $\{-\delta_{k-m+n-k^{\prime}+m^{\prime}-n^{\prime}}\}$\\
\hline

$\mathbb{E}\{a_x |a_x|^2\}\mathbb{E}\{a_x |a_y|^2\}$ & \{\},  \{\}, $\{-\delta_{k-m+n-k^{\prime}+m^{\prime}-n^{\prime}}\}$\\
\hline

$\mathbb{E}\{a_y^* |a_y|^2\}\mathbb{E}\{|a_x|^2a_y\}$ & \{\},  $\{\delta_{k-m-k^{\prime}}\delta_{n+m^{\prime}-n^{\prime}},\delta_{k+n-n^{\prime}}\delta_{m+k^{\prime}-m^{\prime}}\},\{-2\delta_{k-m+n-k^{\prime}+m^{\prime}-n^{\prime}}\}$\\
\hline

$\mathbb{E}\{a_x^* |a_x|^2\}\mathbb{E}\{a_x |a_y|^2\}$ & \{\}, $\{\delta_{k-m-n^{\prime}}\delta_{n-k^{\prime}+m^{\prime}},\delta_{k+n-n^{\prime}}\delta_{m+k^{\prime}-m^{\prime}}\}, \{-2\delta_{k-m+n-k^{\prime}+m^{\prime}-n^{\prime}}\}$\\
\hline

$|\mathbb{E}\{|a_x|^2a_y\}|^2$ &  \{\}, $\{\delta_{k-m-k^{\prime}}\delta_{n+m^{\prime}-n^{\prime}},\delta_{k-m+m^{\prime}}\delta_{n-k^{\prime}-n^{\prime}},\delta_{k-k^{\prime}-n^{\prime}}\delta_{m-n+m^{\prime}},\delta_{k+m^{\prime}-n^{\prime}}\delta_{m-n+k^{\prime}}\},$\\
&$\{-4\delta_{k-m+n-k^{\prime}+m^{\prime}-n^{\prime}}\}$\\
\hline

$\mathbb{E}\{a_x a_y^*\}\mathbb{E}\{a_x^*a_y|a_x|^2\}$ &  \{\}, $\{\delta_{k-k^{\prime}}\delta_{m-n-m^{\prime}+n^{\prime}},\delta_{m-m^{\prime}}\delta_{k+n-k^{\prime}-n^{\prime}},\delta_{n-k^{\prime}}\delta_{k-m+m^{\prime}-n^{\prime}}\},\{-4\delta_{k-m+n-k^{\prime}+m^{\prime}-n^{\prime}}\}$\\
\hline

$\mathbb{E}\{|a_x|^4 |a_y|^2\}$ &  \{\}, \{\}, $\{\delta_{k-m+n-k^{\prime}+m^{\prime}-n^{\prime}}\}$ \\
\hline

$\mathbb{E}\{a_x^2\}\mathbb{E}^*\{a_xa_y\}\mathbb{E}\{a_x^*a_y\}$ & $\{\delta_{k+n}\delta_{m-m^{\prime}}\delta_{k^{\prime}+n^{\prime}}\},\{-2\delta_{k+n}\delta_{m+k^{\prime}-m^{\prime}+n^{\prime}},-\delta_{m+k^{\prime}}\delta_{k+n+m^{\prime}-n^{\prime}},-\delta_{k^{\prime}+n^{\prime}}\delta_{k+n-m+m^{\prime}},$\\
&  $-\delta_{m-m^{\prime}}\delta_{k+n-k^{\prime}-n^{\prime}}\},\{4\delta_{k-m+n-k^{\prime}+m^{\prime}-n^{\prime}}\}$\\
\hline

$\mathbb{E}^*\{a_x^2\}\mathbb{E}\{a_xa_y\}\mathbb{E}\{a_xa^*_y\}$ & $\{\delta_{k-k^{\prime}}\delta_{m+n^{\prime}}\delta_{n+m^{\prime}},\delta_{k+m^{\prime}}\delta_{m+n^{\prime}}\delta_{n-k^{\prime}}\},\{-2\delta_{m+n^{\prime}}\delta_{k+n-k^{\prime}+m^{\prime}},-\delta_{n+m^{\prime}}\delta_{k-m-k^{\prime}-n^{\prime}},$\\
&  $-\delta_{k-k^{\prime}}\delta_{m-n-m^{\prime}+n^{\prime}},-\delta_{n-k^{\prime}}\delta_{k-m+m^{\prime}-n^{\prime}},-\delta_{k+m^{\prime}}\delta_{m-n+k^{\prime}+n^{\prime}}\},\{4\delta_{k-m+n-k^{\prime}+m^{\prime}-n^{\prime}}\}$\\
\hline

$|\mathbb{E}\{a_x^2\}|^2\mathbb{E}\{|a_y|^2\}$ &  \{\}, $\{-\delta_{m+n^{\prime}}\delta_{k+n-k^{\prime}+m^{\prime}},-\delta_{k+n}\delta_{m+k^{\prime}-m^{\prime}+n^{\prime}}\},\{2\delta_{k-m+n-k^{\prime}+m^{\prime}-n^{\prime}}\}$\\
\hline

$|\mathbb{E}\{a_x^2a_y^*\}|^2$ &  \{\}, $\{\delta_{k+n-k^{\prime}}\delta_{m-m^{\prime}+n^{\prime}}\},\{-\delta_{k-m+n-k^{\prime}+m^{\prime}-n^{\prime}}\}$ \\
\hline

$|\mathbb{E}\{a_x^2a_y\}|^2$ &  \{\}, $\{\delta_{k+n+m^{\prime}}\delta_{m+k^{\prime}+n^{\prime}}\}, \{-\delta_{k-m+n-k^{\prime}+m^{\prime}-n^{\prime}}\}$\\
\hline

$\mathbb{E}\{a_x^2\}\mathbb{E}^*\{a_x^2 |a_y|^2\}$ &  \{\}, $\{\delta_{k+n}\delta_{m+k^{\prime}-m^{\prime}+n^{\prime}},\delta_{m+n^{\prime}}\delta_{k+n-k^{\prime}+m^{\prime}}\},\{-2\delta_{k-m+n-k^{\prime}+m^{\prime}-n^{\prime}}\}$ \\
\hline

$\mathbb{E}\{a_xa_y\}\mathbb{E}^*\{a_xa_y|a_y|^2\}$ &  \{\}, $\{\delta_{k+n}\delta_{m+k^{\prime}-m^{\prime}+n^{\prime}},\delta_{n+m^{\prime}}\delta_{k-m-k^{\prime}-n^{\prime}}\},\{-2\delta_{k-m+n-k^{\prime}+m^{\prime}-n^{\prime}}\}$\\
\hline

$\mathbb{E}^*\{a_xa_y\}\mathbb{E}\{a_xa_y|a_y|^2\}$ &  \{\}, $\{\delta_{m+n^{\prime}}\delta_{k+n-k^{\prime}+m^{\prime}},\delta_{k^{\prime}+n^{\prime}}\delta_{k-m+n+m^{\prime}}\},\{-2\delta_{k-m+n-k^{\prime}+m^{\prime}-n^{\prime}}\}$\\
\hline

$\mathbb{E}\{a_x^*a_y\}\mathbb{E}\{a_x a_y^*|a_y|^2\}$ &  \{\}, $\{\delta_{k-n^{\prime}}\delta_{m-n+k^{\prime}-m^{\prime}}\},\{-\delta_{k-m+n-k^{\prime}+m^{\prime}-n^{\prime}}\}$ \\
\hline

$\mathbb{E}\{a_x a_y^*\}\mathbb{E}\{a_x^*a_y|a_y|^2\}$ &  \{\}, $\{\delta_{n-k^{\prime}}\delta_{k-m+m^{\prime}-n^{\prime}}\},\{-\delta_{k-m+n-k^{\prime}+m^{\prime}-n^{\prime}}\}$\\

\hline \hline 
\end{longtable}
}

As it can be observed in Table~\ref{tab:kron_delta}, each correlation term is associated with different delta functions. To compact these terms we exploit a property introduced in the following proposition.
\begin{proposition}\label{prop:1}
Let $D_1(k,m,n,k^{\prime},m^{\prime},n^{\prime})$ and $D_2(k,m,n,k^{\prime},m^{\prime},n^{\prime})$ be two Kronecker delta products of the kind shown in Table \ref{tab:kron_delta}. If 
\begin{equation}
D_1(k,m,n,k^{\prime},m^{\prime},n^{\prime})=D_2(n,m,k,n^{\prime},m^{\prime},k^{\prime}), 
\label{eq:transf_D1D2}
\end{equation}
then
\begin{align}
\begin{split}
&\sum_{\substack{(k,m,n) \in \mathcal{S}_i \\ (k^{\prime},m^{\prime},n^{\prime}) \in \mathcal{S}_i}}\mathcal{P}_{k,m,n,k^{\prime},m^{\prime},n^{\prime}}\eta_{k,m,n}\eta^*_{k^{\prime},m^{\prime},n^{\prime}}D_1(k,m,n,k^{\prime},m^{\prime},n^{\prime})\\
&=\sum_{\substack{(k,m,n) \in \mathcal{S}_i \\ (k^{\prime},m^{\prime},n^{\prime}) \in \mathcal{S}_i}}\mathcal{P}_{k,m,n,k^{\prime},m^{\prime},n^{\prime}}\eta_{k,m,n}\eta^*_{k^{\prime},m^{\prime},n^{\prime}}D_2(k,m,n,k^{\prime},m^{\prime},n^{\prime}).
\end{split}
\label{eq:prop_1}
\end{align}
This property also holds when applying the transformations $k=n$, $n=k$ and $k^{\prime}= n^{\prime}$, $n^{\prime}=k^{\prime}$ individually.
\begin{proof}
See Appendix \ref{app:C}.
\end{proof}
\end{proposition}

The property in \eqref{eq:prop_1} allows us to group many of the Kronecker function products in Table \ref{tab:kron_delta} under a single term. Namely, the Kronecker delta products in Table \ref{tab:kron_delta} can be grouped in subsets that are closed to property \eqref{eq:transf_D1D2}, since they all result in the same value of the summations in \eqref{eq:prop_1}. In particular, 14 distinct subsets can be identified for the list of Kronecker delta products in Table \ref{tab:kron_delta}. We label these subsets as $\mathcal{D}_{l}$ for $l=1,2,\dots,14$, which are shown in Table~\ref{tab:kron_delta_sets}, and where we have dropped the separation in cross-polarisation and intra-polarisation terms.

{\small
\begin{table}[h]
\centering
\caption{Subsets of Kronecker delta products which are closed to property \eqref{eq:transf_D1D2}. The terms in boldface are the ones used to group all the other elements within each set.}\label{tab:kron_delta_sets}
\begin{tabular}{c|c}
\hline \hline
\textbf{Set name} & \multicolumn{1}{c}{\textbf{Set elements}}  \\
\hline \hline
$\mathcal{D}_1$ & $\boldsymbol{\delta_{k-k^{\prime}}\delta_{m-m^{\prime}}\delta_{n-n^{\prime}}},\delta_{k-n^{\prime}}\delta_{m-m^{\prime}}\delta_{n-k^{\prime}}$\\
\hline
$\mathcal{D}_2$ & $\boldsymbol{\delta_{k-k^{\prime}}\delta_{m+n^{\prime}}\delta_{n+m^{\prime}}},\delta_{k+m^{\prime}}\delta_{m+k^{\prime}}\delta_{n-n^{\prime}},\delta_{k+m^{\prime}}\delta_{m+n^{\prime}}\delta_{n-k^{\prime}},\delta_{k-n^{\prime}}\delta_{m+k^{\prime}}\delta_{n+m^{\prime}}$  \\
\hline
$\mathcal{D}_3$ & $\boldsymbol{\delta_{k+n}\delta_{m-m^{\prime}}\delta_{k^{\prime}+n^{\prime}}}$  \\
\hline
$\mathcal{D}_4$ & $\boldsymbol{\delta_{k-m-k^{\prime}}\delta_{n+m^{\prime}-n^{\prime}}},\delta_{k-m-n^{\prime}}\delta_{n-k^{\prime}+m^{\prime}},\delta_{k+m^{\prime}-n^{\prime}}\delta_{m-n+k^{\prime}},\delta_{k-k^{\prime}+m^{\prime}}\delta_{m-n+n^{\prime}}$ \\
\hline
$\mathcal{D}_5$ & $\boldsymbol{\delta_{k-m+m^{\prime}}\delta_{n-k^{\prime}-n^{\prime}}},\delta_{k-k^{\prime}-n^{\prime}}\delta_{m-n-m^{\prime}}$\\
\hline
$\mathcal{D}_6$ & $\boldsymbol{\delta_{k+n-k^{\prime}}\delta_{m-m^{\prime}+n^{\prime}}},\delta_{k+n-n^{\prime}}\delta_{m+k^{\prime}-m^{\prime}}$ \\
\hline
$\mathcal{D}_7$ & $\boldsymbol{\delta_{k+n+m^{\prime}}\delta_{m+k^{\prime}+n^{\prime}}}$ \\

\hline
$\mathcal{D}_8$ & $\boldsymbol{\delta_{k+n}\delta_{m+k^{\prime}-m^{\prime}+n^{\prime}}}$ \\
\hline
$\mathcal{D}_{9}$ & $\boldsymbol{\delta_{k-k^{\prime}}\delta_{m-n-m^{\prime}+n^{\prime}}},\delta_{k-n^{\prime}}\delta_{m-n+k^{\prime}-m^{\prime}},\delta_{n-k^{\prime}}\delta_{k-m+m^{\prime}-n^{\prime}},\delta_{n-n^{\prime}}\delta_{k-m-k^{\prime}+m^{\prime}}$ \\
\hline
$\mathcal{D}_{10}$ & $\boldsymbol{\delta_{k+m^{\prime}}\delta_{m-n+k^{\prime}+n^{\prime}}},\delta_{n+m^{\prime}}\delta_{k-m-k^{\prime}-n^{\prime}}$ \\
\hline
$\mathcal{D}_{11}$ & $\boldsymbol{\delta_{m-m^{\prime}}\delta_{k+n-k^{\prime}-n^{\prime}}}$ \\
\hline
$\mathcal{D}_{12}$ & $\boldsymbol{\delta_{m+k^{\prime}}\delta_{k+n+m^{\prime}-n^{\prime}}},\delta_{m+n^{\prime}}\delta_{k+n-k^{\prime}+m^{\prime}}$ \\
\hline
$\mathcal{D}_{13}$ & $\boldsymbol{\delta_{k^{\prime}+n^{\prime}}\delta_{k-m+n+m^{\prime}}}$ \\
\hline
$\mathcal{D}_{14}$ & $\boldsymbol{\delta_{k-m+n-k^{\prime}+m^{\prime}-n^{\prime}}}$ \\
\hline \hline 
\end{tabular}
\end{table}
}

Summing all the contributions in Table \ref{tab:kron_delta}, using Proposition \ref{prop:1} for the elements in the subsets listed in Table \ref{tab:kron_delta}, and finally ordering by Kronecker delta product, we obtain from \eqref{eq:psd_sumcontr}

\begin{align}
\begin{split}
S_x(f,N_s&,L_s)=\biggl(\frac{8}{9}\biggr)^2\gamma^2\Delta_f\sum_{i=-\infty}^{\infty}\delta(f-i\Delta_f)\\
&\cdot\sum_{\substack{(k,m,n) \in \mathcal{S}_i \\ (k^{\prime},m^{\prime},n^{\prime}) \in \mathcal{S}_i}}\bigl[R_s^3\Delta_f^2\bigl[(\mathsf{a}_1\mathsf{P}+2\Re\{\mathsf{a}_1^{\prime}\mathsf{P}\})\delta_{k-k^{\prime}}\delta_{m-m^{\prime}}\delta_{n-n^{\prime}}+(\mathsf{a}_2\mathsf{P}+2\Re\{\mathsf{a}_2^{\prime}\mathsf{P}\})\delta_{k-k^{\prime}}\delta_{m+n^{\prime}}\delta_{n+m^{\prime}}\\
&+(\mathsf{a}_3\mathsf{P}+2\Re\{\mathsf{a}_3^{\prime}\mathsf{P}\})\delta_{k+n}\delta_{m-m^{\prime}}\delta_{k^{\prime}+n^{\prime}}\bigr]+R_s^2\Delta_f^3\bigl[(\mathsf{b}_1\mathsf{P}+2\Re\{\mathsf{b}_1^{\prime}\mathsf{P}\})\delta_{k-m-k^{\prime}}\delta_{n+m^{\prime}-n^{\prime}}\\
&+(\mathsf{b}_2\mathsf{P}+2\Re\{\mathsf{b}_2^{\prime}\mathsf{P}\})\delta_{k-m+m^{\prime}}\delta_{n-k^{\prime}-n^{\prime}}+(\mathsf{b}_2^*\mathsf{P}+2\Re\{\mathsf{b}_3^{\prime}\mathsf{P}\})\delta_{k+n-k^{\prime}}\delta_{m-m^{\prime}+n^{\prime}}\\
&+(\mathsf{b}_4\mathsf{P}+2\Re\{\mathsf{b}_4^{\prime}\mathsf{P}\})\delta_{k+n+m^{\prime}}\delta_{m+k^{\prime}+n^{\prime}}+(\mathsf{c}_1\mathsf{P}+2\Re\{\mathsf{c}_1^{\prime}\mathsf{P}\})\delta_{k+n}\delta_{m+k^{\prime}-m^{\prime}+n^{\prime}}\\
&+(\mathsf{c}_2\mathsf{P}+2\Re\{\mathsf{c}_2^{\prime}\mathsf{P}\})\delta_{k-k^{\prime}}\delta_{m-n-m^{\prime}+n^{\prime}}+(\mathsf{c}_3\mathsf{P}+2\Re\{\mathsf{c}_3^{\prime}\mathsf{P}\})\delta_{k+m^{\prime}}\delta_{m-n+k^{\prime}+n^{\prime}}\\
&+(\mathsf{c}_4\mathsf{P}+2\Re\{\mathsf{c}_4^{\prime}\mathsf{P}\})\delta_{m-m^{\prime}}\delta_{k+n-k^{\prime}-n^{\prime}}+(\mathsf{c}_3^*\mathsf{P}+2\Re\{\mathsf{c}_5^{\prime}\mathsf{P}\})\delta_{m+k^{\prime}}\delta_{k+n+m^{\prime}-n^{\prime}}\\
&+(\mathsf{c}_1^*\mathsf{P}+2\Re\{\mathsf{c}_6^{\prime}\mathsf{P}\})\delta_{k^{\prime}+n^{\prime}}\delta_{k-m+n+m^{\prime}}\bigr]+R_s\Delta_f^4(\mathsf{d}_1\mathsf{P}+2\Re\{\mathsf{d}_1^{\prime}\mathsf{P}\})\delta_{k-m+n-k^{\prime}+m^{\prime}-n^{\prime}}\bigr],\
\end{split}\label{eq:final_exp1}
\end{align}
where the coefficients multiplying $\mathsf{P}$ are listed in Table \ref{tab:final_coeff0} and where we have occupied the coset leaders in Table \ref{tab:kron_delta_sets}.

{\small
\begin{table}[tbp]
\caption{Expressions for the coefficients in \eqref{eq:final_exp1}.}\label{tab:final_coeff0}
\begin{tabular}{c|l||c|l}
\hline \hline
\textbf{Name} &\textbf{Value} & \textbf{Name} &\textbf{Value}\\
\hline
$\mathsf{a}_1$ &\makecell[l]{$2\mathbb{E}^3\{|a_x|^2\}+\mathbb{E}\{|a_x|^2\}\mathbb{E}^2\{|a_y|^2\}$\\$+|\mathbb{E}\{a_xa_y^*\}|^2\mathbb{E}\{|a_y|^2\}$} & $\mathsf{a}^{\prime}_1$ &
$2\mathbb{E}\{|a_x|^2\}|\mathbb{E}\{a_xa_y^*\}|^2$ \\
\hline

$\mathsf{a}_2$ &\makecell[l]{$4\mathbb{E}\{|a_x|^2\}|\mathbb{E}\{a_x^2\}|^2+\mathbb{E}\{|a_x|^2\}|\mathbb{E}\{a_y^2\}|^2$ \\ $+2\Re\{\mathbb{E}\{a_xa_y\}\mathbb{E}\{a_x^*a_y\}\mathbb{E}^*\{a_y^2\}\}+|\mathbb{E}\{a_xa_y\}|^2\mathbb{E}\{|a_y|^2\}$} & $\mathsf{a}^{\prime}_2$& \makecell[l]{$2\mathbb{E}\{|a_x|^2\}|\mathbb{E}\{a_xa_y\}|^2+2\mathbb{E}^*\{a_x^2\}\mathbb{E}\{a_xa_y\}\mathbb{E}\{a_xa_y^*\}$} \\
\hline

$\mathsf{a}_3$ &$\mathbb{E}\{|a_x|^2\}|\mathbb{E}\{a_x^2\}|^2+|\mathbb{E}\{a_xa_y\}|^2\mathbb{E}\{|a_y|^2\}$ &$\mathsf{a}^{\prime}_3$ &
$\mathbb{E}\{a_x^2\}\mathbb{E}^*\{a_xa_y\}\mathbb{E}\{a_x^*a_y\}$ \\
\hline

$\mathsf{b}_1$ &\makecell[l]{$4|\mathbb{E}\{a_x|a_x|^2\}|^2+\mathbb{E}\{|a_x|^2a_y\}\mathbb{E}\{a_y^*|a_y|^2\}$\\$+\mathbb{E}\{|a_x|^2a_y^*\}\mathbb{E}\{a_y|a_y|^2\}+|\mathbb{E}\{a_x|a_y|^2\}|^2+|\mathbb{E}\{a_x^*a_y^2\}|^2$} &$\mathsf{b}^{\prime}_1$& 
\makecell[l]{$\mathbb{E}\{a_x^*|a_x|^2\}\mathbb{E}\{a_x|a_y|^2\}+2|\mathbb{E}\{|a_x|^2a_y\}|^2$} \\
\hline

$\mathsf{b}_2$ &\makecell[l]{$2|\mathbb{E}\{a_x|a_x|^2\}|^2+\mathbb{E}\{|a_x|^2a_y^*\}\mathbb{E}\{a_y|a_y|^2\}+|\mathbb{E}\{a_x|a_y|^2\}|^2$} & $\mathsf{b}_2^{\prime}$& $2|\mathbb{E}\{|a_x|^2a_y\}|^2$ \\
\hline

\multicolumn{2}{c||}{} &$\mathsf{b}^{\prime}_3$& $\mathbb{E}\{a_x^*|a_x|^2\}\mathbb{E}\{a_x|a_y|^2\}+|\mathbb{E}\{a_x^2a_y^*\}|^2$
\\
\hline

$\mathsf{b}_4$ &$|\mathbb{E}\{a_x^3\}|^2+|\mathbb{E}\{a_xa_y^2\}|^2$&$\mathsf{b}^{\prime}_4$&
$|\mathbb{E}\{a_x^2a_y\}|^2$ \\
\hline

$\mathsf{c}_1$ &\makecell[l]{$-3\mathbb{E}\{|a_x|^2\}|\mathbb{E}\{a_x^2\}|^2+\mathbb{E}^*\{a_x^2|a_x|^2\}\mathbb{E}\{a_x^2\}$\\$-\mathbb{E}\{a_xa_y\}\mathbb{E}\{a_x^*a_y\}\mathbb{E}^*\{a_y^2\}-2|\mathbb{E}\{a_xa_y\}|^2\mathbb{E}\{|a_y|^2\}$\\$+\mathbb{E}\{a_xa_y\}\mathbb{E}^*\{a_xa_y|a_y|^2\}$} &$\mathsf{c}^{\prime}_1$& \makecell[l]{$-2\mathbb{E}\{a_x^2\}\mathbb{E}^*\{a_xa_y\}\mathbb{E}\{a_x^*a_y\}-|\mathbb{E}\{a_x^2\}|^2\mathbb{E}\{|a_y|^2\}$\\$+\mathbb{E}\{a_x^2\}\mathbb{E}^*\{a_x^2|a_y|^2\}$} \\
\hline

$\mathsf{c}_2$ &\makecell[l]{$-8\mathbb{E}^3\{|a_x|^2\}-4\mathbb{E}\{|a_x|^2\}|\mathbb{E}\{a_x^2\}|^2+4\mathbb{E}\{|a_x|^4\}\mathbb{E}\{|a_x|^2\}$\\$-3\mathbb{E}\{|a_x|^2\}\mathbb{E}^2\{|a_y|^2\}-\mathbb{E}\{|a_x|^2\}|\mathbb{E}\{a_y^2\}|^2$\\$+\mathbb{E}\{|a_x|^2\}\mathbb{E}\{|a_y|^4\}+\mathbb{E}\{|a_x|^2|a_y|^2\}\mathbb{E}\{|a_y|^2\}$\\$-5|\mathbb{E}\{a_xa_y^*\}|^2\mathbb{E}\{|a_y|^2\}-|\mathbb{E}\{a_xa_y\}|^2\mathbb{E}\{|a_y|^2\}$\\$-2\Re\{\mathbb{E}\{a_xa_y\}\mathbb{E}\{a_x^*a_y\}\mathbb{E}^*\{a_y^2\}\}$\\$+2\Re\{\mathbb{E}\{a_x^*a_y\}\mathbb{E}\{a_xa_y^*|a_y|^2\}\}$} & $\mathsf{c}^{\prime}_2$& \makecell[l]{$-6\mathbb{E}\{|a_x|^2\}|\mathbb{E}\{a_xa_y^*\}|^2-2\mathbb{E}\{|a_x|^2\}|\mathbb{E}\{a_xa_y\}|^2$\\$-2\mathbb{E}^2\{|a_x|^2\}\mathbb{E}\{|a_y|^2\}+2\mathbb{E}\{|a_x|^2\}\mathbb{E}\{|a_x|^2|a_y|^2\}$\\$+2\mathbb{E}\{a_xa_y^*\}\mathbb{E}\{a_x^*a_y|a_x|^2\}$\\$-2\mathbb{E}^*\{a_x^2\}\mathbb{E}\{a_xa_y\}\mathbb{E}\{a_xa_y^*\}$} \\
\hline

$\mathsf{c}_3$ &\makecell[l]{$-6\mathbb{E}\{|a_x|^2\}|\mathbb{E}\{a_x^2\}|^2+2\mathbb{E}^*\{a_x^2|a_x|^2\}\mathbb{E}\{a_x^2\}$\\$-\mathbb{E}\{|a_x|^2\}|\mathbb{E}\{a_y^2\}|^2+\mathbb{E}^*\{|a_x|^2a_y^2\}\mathbb{E}\{a_y^2\}
$\\$-2\Re\{\mathbb{E}^*\{a_xa_y\}\mathbb{E}\{a_xa^*_y\}\mathbb{E}\{a_y^2\}\}$\\$-\mathbb{E}^*\{a_xa_y\}\mathbb{E}\{a_xa^*_y\}\mathbb{E}\{a_y^2\}-2|\mathbb{E}\{a_xa_y\}|^2\mathbb{E}\{|a_y|^2\}$\\$+\mathbb{E}\{a_xa_y\}\mathbb{E}^*\{a_xa_y|a_y|^2\}$} & $\mathsf{c}^{\prime}_3$&
\makecell[l]{$4\mathbb{E}\{|a_x|^2\}|\mathbb{E}\{a_xa_y\}|^2+2\mathbb{E}\{a_xa_y\}\mathbb{E}^*\{a_x|a_x|^2a_y\}$\\$-2\mathbb{E}^*\{a_x^2\}\mathbb{E}\{a_xa_y\}\mathbb{E}\{a_xa_y^*\}$} \\
\hline

$\mathsf{c}_4$ &\makecell[l]{$2\mathbb{E}^3\{|a_x|^2\}-\mathbb{E}\{|a_x|^2\}|\mathbb{E}\{a_x^2\}|^2-\mathbb{E}\{|a_x|^2\}\mathbb{E}^2\{|a_y|^2\}$\\$+\mathbb{E}\{|a_x|^2|a_y|^2\}\mathbb{E}\{|a_y|^2\}-|\mathbb{E}\{a_xa_y^*\}|^2\mathbb{E}\{|a_y|^2\}$\\$-|\mathbb{E}\{a_xa_y\}|^2\mathbb{E}\{|a_y|^2\}+\mathbb{E}\{|a_x|^4\}\mathbb{E}\{|a_x|^2\}$} &$\mathsf{c}^{\prime}_4$& \makecell[l]{$-2\mathbb{E}\{|a_x|^2\}|\mathbb{E}\{a_xa_y^*\}|^2+\mathbb{E}\{a_xa_y^*\}\mathbb{E}\{a_x^*a_y|a_x|^2\}$\\$-\mathbb{E}\{a_x^2\}\mathbb{E}^*\{a_xa_y\}\mathbb{E}\{a_x^*a_y\}$} \\
\hline

\multicolumn{2}{c||}{}&$\mathsf{c}^{\prime}_5$& \makecell[l]{$2\mathbb{E}\{|a_x|^2\}|\mathbb{E}\{a_xa_y\}|^2+\mathbb{E}\{a_xa_y\}\mathbb{E}^*\{a_xa_y|a_x|^2\}$\\$-2\Re\{\mathbb{E}\{a_x^2\}\mathbb{E}^*\{a_xa_y\}\mathbb{E}\{a_x^*a_y\}\}$\\$-\mathbb{E}^*\{a_x^2\}\mathbb{E}\{a_xa_y\}\mathbb{E}\{a_xa_y^*\}-|\mathbb{E}\{a_x^2\}|^2\mathbb{E}\{|a_y|^2\}$} \\
\hline

\multicolumn{2}{c||}{}&$\mathsf{c}^{\prime}_6$ & \makecell[l]{$-2\mathbb{E}\{|a_x|^2\}|\mathbb{E}\{a_xa_y\}|^2+\mathbb{E}\{a_xa_y\}\mathbb{E}^*\{a_xa_y|a_x|^2\}$\\$ -\mathbb{E}\{a_x^2\}\mathbb{E}^*\{a_xa_y\}\mathbb{E}\{a_x^*a_y\}$} \\
\hline

$\mathsf{d}_1$ &\makecell[l]{$12\mathbb{E}^3\{|a_x|^2\}+18\mathbb{E}\{|a_x|^2\}|\mathbb{E}\{a_x^2\}|^2-|\mathbb{E}\{a_x^3\}|^2$\\$-9|\mathbb{E}\{a_x|a_x|^2\}|^2-9\mathbb{E}\{|a_x|^4\}\mathbb{E}\{|a_x|^2$\}\\$-6\Re\{\mathbb{E}\{a_x^2|a_x|^2\}\mathbb{E}^*\{a_x^2\}\}+\mathbb{E}\{|a_x|^6\}$\\$+4\mathbb{E}\{|a_x|^2\}\mathbb{E}^2\{|a_y|^2\}+2\mathbb{E}\{|a_x|^2\}|\mathbb{E}\{a_y^2\}|^2$\\$-\mathbb{E}\{|a_x|^2\}\mathbb{E}\{|a_y|^4\}-4\mathbb{E}\{|a_x|^2|a_y|^2\}\mathbb{E}\{|a_y|^2\}$\\$-4\Re\{\mathbb{E}\{|a_x|^2a_y\}\mathbb{E}\{a_y^* |a_y|^2\}\}$\\$-2\Re\{\mathbb{E}\{|a_x|^2 a_y^2\}\mathbb{E}^*\{ a_y^2\}\}-4|\mathbb{E}\{a_x|a_y|^2\}|^2$\\$+8|\mathbb{E}\{a_xa_y^*\}|^{2}\mathbb{E}\{|a_y|^2\}+8|\mathbb{E}\{a_xa_y\}|^2\mathbb{E}\{|a_y|^2\}$\\$+8\Re\{\mathbb{E}\{a_xa_y\}\mathbb{E}\{a_x^*a_y\}\mathbb{E}^*\{a_y^2\}\}$\\$-|\mathbb{E}\{a_xa_y^2\}|^2-|\mathbb{E}\{a_x^*a_y^2\}|^2+\mathbb{E}\{|a_x|^2|a_y|^4\}$\\$-4\Re\{\mathbb{E}\{a_xa_y\}\mathbb{E}^*\{a_xa_y|a_y|^2\}\}$\\$-2\Re\{\mathbb{E}\{a_x a_y^*\}\mathbb{E}\{a_x^*a_y|a_y|^2\}\}$} & $\mathsf{d}^{\prime}_1$& \makecell[l]{$8\mathbb{E}\{|a_x|^2\}|\mathbb{E}\{a_xa_y^*\}|^2-|\mathbb{E}\{a_x^2a_y^*\}|^2$\\$+8\mathbb{E}\{|a_x|^2\}|\mathbb{E}\{a_xa_y\}|^2+4\mathbb{E}^2\{|a_x|^2\}\mathbb{E}\{|a_y|^2\}$\\$-4\mathbb{E}\{|a_x|^2\}\mathbb{E}\{|a_x|^2 |a_y|^2\}-\mathbb{E}\{|a_x|^4\}\mathbb{E}\{|a_y|^2\}$\\$-2\mathbb{E}\{a_x^* |a_x|^2\}\mathbb{E}\{a_x |a_y|^2\}-2\mathbb{E}\{a_x^2\}\mathbb{E}^*\{a_x^2 |a_y|^2$\\$-\mathbb{E}\{a_x |a_x|^2\}\mathbb{E}\{a_x |a_y|^2\}$\\$-4|\mathbb{E}\{|a_x|^2a_y\}|^2+\mathbb{E}\{|a_x|^4 |a_y|^2\}$\\$-4\mathbb{E}\{a_x a_y^*\}\mathbb{E}\{a_x^*a_y|a_x|^2\}$\\$-4\mathbb{E}\{a_xa_y\}\mathbb{E}^*\{a_xa_y|a_x|^2\}$\\$+8\Re\{\mathbb{E}\{a_x^2\}\mathbb{E}^*\{a_xa_y\}\mathbb{E}\{a_x^*a_y\}\}$\\$+2|\mathbb{E}\{a_x^2\}|^2\mathbb{E}\{|a_y|^2\}-|\mathbb{E}\{a_x^2a_y\}|^2$} \\
\hline \hline
\end{tabular}
\end{table}
}

Eq.~\eqref{eq:final_exp1} can be further manipulated using the following proposition.   
\begin{proposition}\label{prop:4}
Let  $D_1(k,m,n,k^{\prime},m^{\prime},n^{\prime})$ and $D_2(k,m,n,k^{\prime},m^{\prime},n^{\prime})$ be two Kronecker delta products of the kind shown in the second column of Table \ref{tab:kron_delta}. If
$D_1(k,m,n,k^{\prime},m^{\prime},n^{\prime})=D_2(k^{\prime},m^{\prime},n^{\prime},k,m,n)$ then
\begin{align}
\begin{split}
&\sum_{\substack{(k,m,n) \in \mathcal{S}_i \\ (k^{\prime},m^{\prime},n^{\prime}) \in \mathcal{S}_i}}\mathcal{P}_{k,m,n,k^{\prime},m^{\prime},n^{\prime}}\eta_{k,m,n}\eta^*_{k^{\prime},m^{\prime},n^{\prime}}D_1(k,m,n,k^{\prime},m^{\prime},n^{\prime})\\
&=\left(\sum_{\substack{(k,m,n) \in \mathcal{S}_i \\ (k^{\prime},m^{\prime},n^{\prime}) \in \mathcal{S}_i}}\mathcal{P}_{k,m,n,k^{\prime},m^{\prime},n^{\prime}}\eta_{k,m,n}\eta^*_{k^{\prime},m^{\prime},n^{\prime}}D_2(k,m,n,k^{\prime},m^{\prime},n^{\prime})\right)^*.
\end{split}\label{eq:prop2}
\end{align}
\end{proposition}
\begin{proof}
See Appendix \ref{app:D}.
\end{proof}

\begin{corollary}\label{cor:1}
Let $\mathcal{D}$ be a set of delta products $\mathcal{D}$ closed to the transformation in Proposition \ref{prop:4}, i.e. $\forall D_1 \in \mathcal{D},\; \exists \; D_2 \in \mathcal{D}:D_2(k,m,n,k^{\prime},m^{\prime},n^{\prime})=D_1(k^{\prime},m^{\prime},n^{\prime},k,m,n)$, then 
\begin{equation}
\sum_{\substack{(k,m,n) \in \mathcal{S}_i \\ (k^{\prime},m^{\prime},n^{\prime}) \in \mathcal{S}_i}}\mathcal{P}_{k,m,n,k^{\prime},m^{\prime},n^{\prime}}\eta_{k,m,n}\eta^*_{k^{\prime},m^{\prime},n^{\prime}}D_1(k,m,n,k^{\prime},m^{\prime},n^{\prime}) \in \mathbb{R}, \;\; \forall D_1 \in
\mathcal{D}.
\label{eq:realval_sum}
\end{equation}
This clearly includes the case $D_1=D_2$.
\end{corollary}
\begin{proof}
This corollary directly follows from the fact that for any $D_1$ and $D_2$ in the set $\mathcal{D}$ in the hypothesis, both \eqref{eq:prop_1} and \eqref{eq:prop2} must hold.
\end{proof}

Since the sets $\mathcal{D}_i$ for $i=1,2,3,4,7,9,11,14$, are closed to the transformation in Corollary \ref{cor:1}, following \eqref{eq:realval_sum} we can rewrite  \eqref{eq:final_exp1} as
\begin{align*}
\begin{split}
S_x(f,N_s,&L_s)=\biggl(\frac{8}{9}\biggr)^2\gamma^2\Delta_f\sum_{i=-\infty}^{\infty}\delta(f-i\Delta_f)\\
&\cdot\sum_{\substack{(k,m,n) \in \mathcal{S}_i \\ (k^{\prime},m^{\prime},n^{\prime}) \in \mathcal{S}_i}}\bigl[R_s^3\Delta_f^2\bigl[(\mathsf{a}_1+2\Re\{\mathsf{a}_1^{\prime}\})\mathsf{P}\delta_{k-k^{\prime}}\delta_{m-m^{\prime}}\delta_{n-n^{\prime}}+(\mathsf{a}_2+2\Re\{\mathsf{a}_2^{\prime}\})\mathsf{P}\delta_{k-k^{\prime}}\delta_{m+n^{\prime}}\delta_{n+m^{\prime}}\\
&+(\mathsf{a}_3+2\Re\{\mathsf{a}_3^{\prime}\})\mathsf{P}\delta_{k+n}\delta_{m-m^{\prime}}\delta_{k^{\prime}+n^{\prime}}\bigr]+R_s^2\Delta_f^3\bigl[(\mathsf{b}_1+2\Re\{\mathsf{b}_1^{\prime}\})\mathsf{P}\delta_{k-m-k^{\prime}}\delta_{n+m^{\prime}-n^{\prime}}\\
&+(\mathsf{b}_2\mathsf{P}+2\Re\{\mathsf{b}_2^{\prime}\mathsf{P}\})\delta_{k-m+m^{\prime}}\delta_{n-k^{\prime}-n^{\prime}}+(\mathsf{b}_2^*\mathsf{P}+2\Re\{\mathsf{b}_3^{\prime}\mathsf{P}\})\delta_{k+n-k^{\prime}}\delta_{m-m^{\prime}+n^{\prime}}\\
&+(\mathsf{b}_4+2\Re\{\mathsf{b}_4^{\prime}\})\mathsf{P}\delta_{k+n+m^{\prime}}\delta_{m+k^{\prime}+n^{\prime}}+(\mathsf{c}_1\mathsf{P}+2\Re\{\mathsf{c}_1^{\prime}\mathsf{P}\})\delta_{k+n}\delta_{m+k^{\prime}-m^{\prime}+n^{\prime}}\\
&+(\mathsf{c}_2+2\Re\{\mathsf{c}_2^{\prime}\})\mathsf{P}\delta_{k-k^{\prime}}\delta_{m-n-m^{\prime}+n^{\prime}}+(\mathsf{c}_3\mathsf{P}+2\Re\{\mathsf{c}_3^{\prime}\mathsf{P}\})\delta_{k+m^{\prime}}\delta_{m-n+k^{\prime}+n^{\prime}}\\
&+(\mathsf{c}_4+2\Re\{\mathsf{c}_4^{\prime}\})\mathsf{P}\delta_{m-m^{\prime}}\delta_{k+n-k^{\prime}-n^{\prime}}+(\mathsf{c}_5\mathsf{P}+2\Re\{\mathsf{c}_5^{\prime}\mathsf{P}\})\delta_{m+k^{\prime}}\delta_{k+n+m^{\prime}-n^{\prime}}\\
&+(\mathsf{c}_6\mathsf{P}+2\Re\{\mathsf{c}_6^{\prime}\mathsf{P}\})\delta_{k^{\prime}+n^{\prime}}\delta_{k-m+n+m^{\prime}}\bigr]+R_s\Delta_f^4(\mathsf{d}_1+2\Re\{\mathsf{d}_1^{\prime}\})\mathsf{P}\delta_{k-m+n-k^{\prime}+m^{\prime}-n^{\prime}}\bigr].
\end{split}
\end{align*}
Furthermore, we note that $(\mathcal{D}_5,\mathcal{D}_6)$, $(\mathcal{D}_8,\mathcal{D}_{13})$ and $(\mathcal{D}_{10},\mathcal{D}_{12})$ represent pairs of complementary sets under the transformation in Proposition \ref{prop:4}, hence their delta product elements can be grouped to form pairs of complex conjugate summations. This finally leads to
\begin{align}
\begin{split}
S_x(f,N_s,L_s)&=\left(\frac{8}{9}\right)^2\gamma^2\Delta_f\sum_{i=-\infty}^{\infty}\delta(f-i\Delta_f)\\
&\cdot\left[R_s^3\Delta_f^2\left[\Phi_1\mathsf{Q}_1+\Phi_2\mathsf{Q}_2+\Phi_3\mathsf{Q}_3\right]+R_s^2\Delta_f^3\left[\Psi_1\mathsf{Q}_4+2\Re\{\Psi_2\mathsf{Q}_5+\Psi_3\mathsf{Q}_5^*\}+\Psi_4\mathsf{Q}_6\right.\right.\\
&\left.\left.+2\Re\{\Lambda_1\mathsf{Q}_7+\Lambda_2\mathsf{Q}_7^*\}+\Lambda_3\mathsf{Q}_8+2\Re\{\Lambda_4\mathsf{Q}_9+\Lambda_5\mathsf{Q}_9^*\}+\Lambda_6\mathsf{Q}_{10}\right]+R_s\Delta_f^4\Xi_1\mathsf{Q}_{11}\right],
\end{split}
\label{eq:psd_periodic} \end{align}
where 
\begin{equation}
\mathsf{Q}_l\triangleq\sum_{\substack{(k,m,n) \in \mathcal{S}_i \\ (k^{\prime},m^{\prime},n^{\prime}) \in \mathcal{S}_i}}\mathsf{P}D^{(l)}=\sum_{\mathcal{T}_{l,i}}\mathsf{P} \qquad l=1,2,\dots,11,
\label{eq:Qdef}
\end{equation}
the coefficients $\Phi_i$, $i=1,2,3$, $\Psi_i$, $i=1,...,4$, $\Lambda_i$, $i=1,...,6$, and $\Xi_1$ in \eqref{eq:psd_periodic} are given in Table \ref{tab:final_coeff1}, the sets $\mathcal{S}_i$ are defined in \eqref{eq:setSi}, $D^{(l)}$ are the coset leaders highlighted in boldface in Table \ref{tab:kron_delta_sets} and listed in Table \ref{tab:final_coeff2} with their corresponding set $\mathcal{D}$. Finally, the sets $\mathcal{T}_{l,i}$ are defined as
\begin{equation*}
\mathcal{T}_{l,i}\triangleq\{(k,m,n,k^{\prime},m^{\prime},n^{\prime}) \in \{0,1,\ldots,W-1\}^6:(k,m,n)\in \mathcal{S}_i,\,  (k^{\prime},m^{\prime},n^{\prime}) \in \mathcal{S}_i,\, D^{(l)}=1\}.
\end{equation*}
Note how in the second equality of \eqref{eq:Qdef} we have accounted for the multiplication by $D^{(l)}$ by restricting the summation set to $\mathcal{T}_{l,i}$.

\begin{table}[tbp]
\parbox{.1\linewidth}{
\hspace{5ex}
}
\parbox{.4\linewidth}{
\centering
\caption{Delta products $D^{(l)}$ in the $\mathsf{Q}_l$ terms $l=1,2,\dots,11,$ in \eqref{eq:Qdef}  with their corresponding $\mathcal{D}$ set in Table~\ref{tab:kron_delta_sets}.}\label{tab:final_coeff2}
\begin{tabular}{c|l|l}
\hline\hline
$l$ & $D^{(l)}$ & \textbf{Set} $\mathcal{D}$\\
\hline\hline
1 & $\delta_{k-k^{\prime}}\delta_{m-m^{\prime}}\delta_{n-n^{\prime}}$ & $\mathcal{D}_1$ \\
\hline
2 & $\delta_{k-k^{\prime}}\delta_{m+n^{\prime}}\delta_{n+m^{\prime}}$ & $\mathcal{D}_2$ \\
\hline
3 & $\delta_{k+n}\delta_{m-m^{\prime}}\delta_{k^{\prime}+n^{\prime}}$ & $\mathcal{D}_3$\\
\hline
4 & $\delta_{k-m-k^{\prime}}\delta_{n+m^{\prime}-n^{\prime}}$ & $\mathcal{D}_4$\\
\hline
5 & $\delta_{k-m+m^{\prime}}\delta_{n-k^{\prime}-n^{\prime}}$ & $\mathcal{D}_5$ \\
\hline
6 & $\delta_{k+n+m^{\prime}}\delta_{m+k^{\prime}+n^{\prime}}$ & $\mathcal{D}_7$\\
\hline
7 & $\delta_{k+n}\delta_{m+k^{\prime}-m^{\prime}+n^{\prime}}$ & $\mathcal{D}_8$\\
\hline
8 & $\delta_{k-k^{\prime}}\delta_{m-n-m^{\prime}+n^{\prime}}$ & $\mathcal{D}_9$ \\
\hline
9 & $\delta_{k+m^{\prime}}\delta_{m-n+k^{\prime}+n^{\prime}}$ & $\mathcal{D}_{10}$ \\
\hline
10 & $\delta_{m-m^{\prime}}\delta_{k+n-k^{\prime}-n^{\prime}}$ & $\mathcal{D}_{11}$\\
\hline
11 & $\delta_{k-m+n-k^{\prime}+m^{\prime}-n^{\prime}}$ & $\mathcal{D}_{14}$\\
\hline\hline
\end{tabular}
\vspace{8ex} 
}
\hfill
\parbox{.35\linewidth}{
\centering
\caption{Correlation coefficients in \eqref{eq:psd_periodic}. The values of $\mathsf{a}_1,\mathsf{a}_1',\mathsf{b}_1,\mathsf{b}_1',\ldots$ are given in Table \ref{tab:final_coeff0}.}\label{tab:final_coeff1}
\begin{tabular}{c|c}
\hline\hline
\textbf{Name} & \textbf{Value} \\
\hline 
$\Phi_1$ & $\mathsf{a}_1+2\Re\{\mathsf{a}_1^{\prime}\}$ \\
\hline
$\Phi_2$ & $\mathsf{a}_2+2\Re\{\mathsf{a}_2^{\prime}\}$ \\
\hline
$\Phi_3$ & $\mathsf{a}_3+2\Re\{\mathsf{a}_3^{\prime}\}$\\
\hline
$\Psi_1$ & $\mathsf{b}_1+2\Re\{\mathsf{b}_1^{\prime}\}$\\
\hline
$\Psi_2$ & $\mathsf{b}_2+\mathsf{b}_2^{\prime}$\\
\hline
$\Psi_3$ & $\mathsf{b}_3^{\prime}$ \\
\hline
$\Psi_4$ & $\mathsf{b}_4+2\Re\{\mathsf{b}_4^{\prime}\}$\\
\hline
$\Lambda_1$ & $\mathsf{c}_1+\mathsf{c}_1^{\prime}$\\
\hline
$\Lambda_2$ & $\mathsf{c}_6^{\prime}$\\
\hline
$\Lambda_3$ & $\mathsf{c}_2+2\Re\{\mathsf{c}_2^{\prime}\}$\\
\hline
$\Lambda_4$ & $\mathsf{c}_3+\mathsf{c}_3^{\prime}$\\
\hline
$\Lambda_5$ & $\mathsf{c}_5^{\prime}$\\
\hline
$\Lambda_6$ & $\mathsf{c}_4+2\Re\{\mathsf{c}_4^{\prime}\}$\\
\hline
$\Xi_1$ & $\mathsf{d}_1+2\Re\{\mathsf{d}_1^{\prime}\}$\\
\hline\hline
\end{tabular}
}
\label{tab:delta_prod}
\parbox{.1\linewidth}{
\hspace{5ex}
}
\end{table}

\section{Final result}\label{sec:final result}
Eq.~\eqref{eq:psd_periodic} expresses the NLI PSD for a periodic signal of period $T=1/\Delta_f$ as a function of the statistical moments and cross-polarisation correlations of a generic 4D modulation format. To generalise this result to aperiodic signals we take the same approach in \cite{Poggiolini2012, Carena2012}, i.e., we let the period $T$ go to infinity, or equivalently, $\Delta_f \rightarrow 0$ (see Fig.~\ref{fig:time_sketch}).

The limit of \eqref{eq:psd_periodic} for $\Delta_f\rightarrow 0$ is a limit of a distribution (a Dirac's delta comb) which is parametric in $\Delta_f$. To rigorously evaluate such a limit we use Lemma \ref{lm:T_dim} and Theorem \ref{th:keyresult} presented in the following. In particular, Theorem \ref{th:keyresult} presents the \emph{final result} of this work.

\begin{lemma}[Dimensionality of the sets $\mathcal{T}_{l,i}$] The sets $\mathcal{T}_{l,i}$, for $l=1,2,3$, for $l=4,...,10$, and for $l=11$, have dimensionality 2, 3 and 4, respectively, $\forall i\in \mathbb{Z}$.  

\label{lm:T_dim}
\end{lemma}
\begin{proof}
See Appendix \ref{app:E}.
\end{proof}

\begin{theorem}[Limit of the distribution $S_x(f,N_s,L_s)$]\label{th:keyresult}
For an arbitrary aperiodic transmitted signal, the PSD \\$\bar{S}_{x}(f,N_s,L_s)\triangleq \lim_{\Delta_f\rightarrow 0}S_x(f,N_s,L_s)$, where $S_x(f,N_s,L_s)$ is given in \eqref{eq:psd_periodic}, is

\begin{align}
\begin{split}
\bar{S}_{x}(f,N_s,L_s)&=\left(\frac{8}{9}\right)^2\gamma^2\left[R_s^3\left(\Phi_{1}\rchi_1(f)+\Phi_{2}\rchi_2(f)+\Phi_{3} \rchi_3(f)\right)+R_s^2\left(\Psi_1\rchi_4(f)+2\Re\{\Psi_2\rchi_5(f)+\Psi_3\rchi_5^*(f)\}\right.\right.\\
&\left.+\Psi_4\rchi_6(f)
+2\Re\{\Lambda_1\rchi_7(f)+\Lambda_2\rchi_7^*(f)\}+\Lambda_3\rchi_8(f)+2\Re\{\Lambda_4\rchi_9(f)+\Lambda_5\rchi_9^*(f)\}+\Lambda_6\rchi_{10}(f)\right)\\
&\left.+R_s\Xi_1\rchi_{11}(f)\right],
\label{eq:final_result}
\end{split}
\end{align}
\\
where the coefficients $\Phi_i$, $i=1,2,3$, $\Psi_i$, $i=1,2,\ldots,4$, $\Lambda_i$, $i=1,2,\ldots,6$, and $\Xi_1$ as well as the integrals $\rchi_i(f)$, $i=1,2,\ldots,11$ are given in Table \ref{tab:final_result}. As discussed at the end of Sec.~\ref{sec:PSD_periodic}, $\bar{S}_y(f)$ can be obtained applying the transformation $x\rightarrow y$, $y\rightarrow x$ to \eqref{eq:final_result}.

The NLI power vector $\boldsymbol{\Sigma}_{\text{NLI}}$ can be obtained from the PSDs in $x$ and $y$ as
\begin{equation}
\boldsymbol{\Sigma}_{\text{NLI}}\triangleq [\sigma^2_{\text{NLI},x},\sigma^2_{\text{NLI},y}]^{T}=\left[\int_{-\infty}^{\infty}\bar{S}_x(f,N_s,L_s)|P(f)|^2df,\int_{-\infty}^{\infty}\bar{S}_y(f,N_s,L_s)|P(f)|^2df\right]^{T},
\label{eq:Sigma_NLI}
\end{equation}
where $P(f)$ is the transmitted pulse spectrum.

\end{theorem}
\begin{proof}
See Appendix~\ref{app:F}.
\end{proof}

The dependency on $N_s$ and $L_s$ of the $\rchi_l(f)$ was removed to make \eqref{eq:final_result} more compact. To derive the expressions for $\rchi_l(f)$ in Table~\ref{tab:final_result}, we used the definition in \eqref{eq:Chi_def} and the property $P(-f)=P^*(f)$, which stems from the fact that $p(t)$ is assumed to be real valued (see Sec.~\ref{sec:system_model}).

{\small
\begin{table}[tbp]
\caption{Table of high-order moments, correlation coefficients, and integrals appearing in \eqref{eq:final_result}. The function $\eta(f_1,f_2,f)$ is defined in \eqref{eq:fwm_efficiency}. }\label{tab:final_result}
\begin{tabular}{c|l}
\hline\hline
$\Phi_1$ & \makecell[l]{$2\mathbb{E}^3\{|a_x|^2\}+4\mathbb{E}\{|a_x|^2\}|\mathbb{E}\{a_xa_y^*\}|^2+\mathbb{E}\{|a_x|^2\}\mathbb{E}^2\{|a_y|^2\}+|\mathbb{E}\{a_xa_y^*\}|^2\mathbb{E}\{|a_y|^2\}$} \\ 
\hline
$\Phi_2$ &\makecell[l]{$4\mathbb{E}\{|a_x|^2\}|\mathbb{E}\{a_x^2\}|^2+\mathbb{E}\{|a_x|^2\}|\mathbb{E}\{a_y^2\}|^2+4\mathbb{E}\{|a_x|^2\}|\mathbb{E}\{a_xa_y\}|^2+|\mathbb{E}\{a_xa_y\}|^2\mathbb{E}\{|a_y|^2\}+2\Re\{\mathbb{E}\{a_xa_y\}\mathbb{E}\{a_x^*a_y\}\mathbb{E}^*\{a_y^2\}$\\$+2\mathbb{E}^*\{a_x^2\}\mathbb{E}\{a_xa_y\}\mathbb{E}\{a_xa_y^*\}\}$}\\ 
\hline
$\Phi_3$ & \makecell[l]{$\mathbb{E}\{|a_x|^2\}|\mathbb{E}\{a_x^2\}|^2+|\mathbb{E}\{a_xa_y\}|^2\mathbb{E}\{|a_y|^2\}+2\Re\{\mathbb{E}\{a_x^2\}\mathbb{E}^*\{a_xa_y\}\mathbb{E}\{a_x^*a_y\}\}$}\\ 
\hline
$\Psi_1$ & \makecell[l]{$4|\mathbb{E}\{a_x|a_x|^2\}|^2+4|\mathbb{E}\{|a_x|^2a_y\}|^2+\mathbb{E}\{|a_x|^2a_y\}\mathbb{E}\{a_y^*|a_y|^2\}+\mathbb{E}\{|a_x|^2a_y^*\}\mathbb{E}\{a_y|a_y|^2\}+|\mathbb{E}\{a_x|a_y|^2\}|^2+|\mathbb{E}\{a_x^*a_y^2\}|^2$\\$+2\Re\{\mathbb{E}\{a_x^*|a_x|^2\}\mathbb{E}\{a_x|a_y|^2\}\}$}   \\
\hline
$\Psi_2$ & \makecell[l]{$2|\mathbb{E}\{a_x|a_x|^2\}|^2+2|\mathbb{E}\{|a_x|^2a_y\}|^2+\mathbb{E}\{|a_x|^2a_y^*\}\mathbb{E}\{a_y|a_y|^2\}+|\mathbb{E}\{a_x|a_y|^2\}|^2$}  \\
\hline
$\Psi_3$ & $\mathbb{E}\{a_x^*|a_x|^2\}\mathbb{E}\{a_x|a_y|^2\}+|\mathbb{E}\{a_x^2a_y^*\}|^2$  \\
\hline
$\Psi_4$ & $|\mathbb{E}\{a_x^3\}|^2+2|\mathbb{E}\{a_x^2a_y\}|^2+|\mathbb{E}\{a_xa_y^2\}|^2
$ \\
\hline
$\Lambda_1$ & \makecell[l]{$-3\mathbb{E}\{|a_x|^2\}|\mathbb{E}\{a_x^2\}|^2+\mathbb{E}^*\{a_x^2|a_x|^2\}\mathbb{E}\{a_x^2\}-|\mathbb{E}\{a_x^2\}|^2\mathbb{E}\{|a_y|^2\}-2|\mathbb{E}\{a_xa_y\}|^2\mathbb{E}\{|a_y|^2\}+\mathbb{E}\{a_x^2\}\mathbb{E}^*\{a_x^2|a_y|^2\}$\\$-2\mathbb{E}\{a_x^2\}\mathbb{E}^*\{a_xa_y\}\mathbb{E}\{a_x^*a_y\}+\mathbb{E}\{a_xa_y\}\mathbb{E}^*\{a_xa_y|a_y|^2\}-\mathbb{E}\{a_xa_y\}\mathbb{E}\{a_x^*a_y\}\mathbb{E}^*\{a_y^2\}$}\\
\hline
$\Lambda_2$ & \makecell[l]{$-2\mathbb{E}\{|a_x|^2\}|\mathbb{E}\{a_xa_y\}|^2+\mathbb{E}\{a_xa_y\}\mathbb{E}^*\{a_xa_y|a_x|^2\} -\mathbb{E}\{a_x^2\}\mathbb{E}^*\{a_xa_y\}\mathbb{E}\{a_x^*a_y\}$}  \\
\hline
$\Lambda_3$ & \makecell[l]{$4\mathbb{E}\{|a_x|^4\}\mathbb{E}\{|a_x|^2\}-4\mathbb{E}\{|a_x|^2\}|\mathbb{E}\{a_x^2\}|^2-8\mathbb{E}^3\{|a_x|^2\}+4\mathbb{E}\{|a_x|^2\}\mathbb{E}\{|a_x|^2|a_y|^2\}-12\mathbb{E}\{|a_x|^2\}|\mathbb{E}\{a_xa_y^*\}|^2$\\$-4\mathbb{E}\{|a_x|^2\}|\mathbb{E}\{a_xa_y\}|^2-4\mathbb{E}^2\{|a_x|^2\}\mathbb{E}\{|a_y|^2\}-3\mathbb{E}\{|a_x|^2\}\mathbb{E}^2\{|a_y|^2\}-\mathbb{E}\{|a_x|^2\}|\mathbb{E}\{a_y^2\}|^2+\mathbb{E}\{|a_x|^2|a_y|^2\}\mathbb{E}\{|a_y|^2\}$\\$+\mathbb{E}\{|a_x|^2\}\mathbb{E}\{|a_y|^4\}-5|\mathbb{E}\{a_xa_y^*\}|^2\mathbb{E}\{|a_y|^2\}-|\mathbb{E}\{a_xa_y\}|^2\mathbb{E}\{|a_y|^2\}-\mathbb{E}^*\{a_x^2\}\mathbb{E}\{a_xa_y\}\mathbb{E}\{a_xa_y^*\}\}$\\$+2\Re\{2\mathbb{E}\{a_xa_y^*\}\mathbb{E}\{a_x^*a_y|a_x|^2\}-\mathbb{E}\{a_xa_y\}\mathbb{E}\{a_x^*a_y\}\mathbb{E}^*\{a_y^2\}+\mathbb{E}\{a_x^*a_y\}\mathbb{E}\{a_xa_y^*|a_y|^2\}\}$}  \\
\hline
$\Lambda_4$ & \makecell[l]{$-6\mathbb{E}\{|a_x|^2\}|\mathbb{E}\{a_x^2\}|^2+2\mathbb{E}^*\{a_x^2|a_x|^2\}\mathbb{E}\{a_x^2\}+4\mathbb{E}\{|a_x|^2\}|\mathbb{E}\{a_xa_y\}|^2
-\mathbb{E}\{|a_x|^2\}|\mathbb{E}\{a_y^2\}|^2+\mathbb{E}^*\{|a_x|^2a_y^2\}\mathbb{E}\{a_y^2\}$\\$+2\mathbb{E}\{a_xa_y\}\mathbb{E}^*\{a_x|a_x|^2a_y\}-2|\mathbb{E}\{a_xa_y\}|^2\mathbb{E}\{|a_y|^2\}-2\mathbb{E}^*\{a_x^2\}\mathbb{E}\{a_xa_y\}\mathbb{E}\{a_xa_y^*\}+\mathbb{E}\{a_xa_y\}\mathbb{E}^*\{a_xa_y|a_y|^2\}$\\$-\mathbb{E}^*\{a_xa_y\}\mathbb{E}\{a_xa^*_y\}\mathbb{E}\{a_y^2\}-2\Re\{\mathbb{E}^*\{a_xa_y\}\mathbb{E}\{a_xa^*_y\}\mathbb{E}\{a_y^2\}\}$}\\
\hline
$\Lambda_5$ & \makecell[l]{$2\mathbb{E}\{|a_x|^2\}|\mathbb{E}\{a_xa_y\}|^2+\mathbb{E}\{a_xa_y\}\mathbb{E}^*\{a_xa_y|a_x|^2\}-|\mathbb{E}\{a_x^2\}|^2\mathbb{E}\{|a_y|^2\}-\mathbb{E}^*\{a_x^2\}\mathbb{E}\{a_xa_y\}\mathbb{E}\{a_xa_y^*\}$\\$-2\Re\{\mathbb{E}\{a_x^2\}\mathbb{E}^*\{a_xa_y\}\mathbb{E}\{a_x^*a_y\}\}$}\\
\hline
$\Lambda_6$ & \makecell[l]{$2\mathbb{E}^3\{|a_x|^2\}+\mathbb{E}\{|a_x|^4\}\mathbb{E}\{|a_x|^2\}-\mathbb{E}\{|a_x|^2\}|\mathbb{E}\{a_x^2\}|^2-4\mathbb{E}\{|a_x|^2\}|\mathbb{E}\{a_xa_y^*\}|^2-\mathbb{E}\{|a_x|^2\}\mathbb{E}^2\{|a_y|^2\}$\\$+\mathbb{E}\{|a_x|^2|a_y|^2\}\mathbb{E}\{|a_y|^2\}-|\mathbb{E}\{a_xa_y^*\}|^2\mathbb{E}\{|a_y|^2\}-|\mathbb{E}\{a_xa_y\}|^2\mathbb{E}\{|a_y|^2\}-\mathbb{E}\{a_x^2\}\mathbb{E}^*\{a_xa_y\}\mathbb{E}\{a_x^*a_y\}$\\$+2\Re\{\mathbb{E}\{a_xa_y^*\}\mathbb{E}\{a_x^*a_y|a_x|^2\}$} \\
\hline
$\Xi_1$ & \makecell[l]{$\mathbb{E}\{|a_x|^6\}-9\mathbb{E}\{|a_x|^4\}\mathbb{E}\{|a_x|^2\}+12\mathbb{E}^3\{|a_x|^2\}-2\mathbb{E}\{|a_x|^4\}\mathbb{E}\{|a_y|^2\}+\mathbb{E}\{|a_x|^2|a_y|^4\}-8\mathbb{E}\{|a_x|^2\}\mathbb{E}\{|a_x|^2 |a_y|^2\}$\\$-4\mathbb{E}\{|a_x|^2|a_y|^2\}\mathbb{E}\{|a_y|^2\}+2\mathbb{E}\{|a_x|^4 |a_y|^2\}-\mathbb{E}\{|a_x|^2\}\mathbb{E}\{|a_y|^4\}+4\mathbb{E}\{|a_x|^2\}\mathbb{E}^2\{|a_y|^2\}+8\mathbb{E}^2\{|a_x|^2\}\mathbb{E}\{|a_y|^2\}$\\$+18\mathbb{E}\{|a_x|^2\}|\mathbb{E}\{a_x^2\}|^2-|\mathbb{E}\{a_x^3\}|^2-9|\mathbb{E}\{a_x|a_x|^2\}|^2+2\mathbb{E}\{|a_x|^2\}|\mathbb{E}\{a_y^2\}|^2-4|\mathbb{E}\{a_x|a_y|^2\}|^2+8|\mathbb{E}\{a_xa_y^*\}|^{2}\mathbb{E}\{|a_y|^2\}$\\$+8|\mathbb{E}\{a_xa_y\}|^2\mathbb{E}\{|a_y|^2\}-|\mathbb{E}\{a_xa_y^2\}|^2-|\mathbb{E}\{a_x^*a_y^2\}|^2+16\mathbb{E}\{|a_x|^2\}|\mathbb{E}\{a_xa_y^*\}|^2-2|\mathbb{E}\{a_x^2a_y^*\}|^2+16\mathbb{E}\{|a_x|^2\}|\mathbb{E}\{a_xa_y\}|^2$\\$-\mathbb{E}\{a_xa_y\}\mathbb{E}^*\{a_xa_y|a_x|^2\}\}+4|\mathbb{E}\{a_x^2\}|^2\mathbb{E}\{|a_y|^2\}-2|\mathbb{E}\{a_x^2a_y\}|^2-2\Re\{3\mathbb{E}\{a_x^2|a_x|^2\}\mathbb{E}^*\{a_x^2\}+2\mathbb{E}\{|a_x|^2a_y\}\mathbb{E}\{a_y^* |a_y|^2\}$\\$+\mathbb{E}\{|a_x|^2 a_y^2\}\mathbb{E}^*\{ a_y^2\}-4\mathbb{E}\{a_xa_y\}\mathbb{E}\{a_x^*a_y\}\mathbb{E}^*\{a_y^2\}+2\mathbb{E}\{a_xa_y\}\mathbb{E}^*\{a_xa_y|a_y|^2\}+\mathbb{E}\{a_x a_y^*\}\mathbb{E}\{a_x^*a_y|a_y|^2\}$\\$+2\mathbb{E}\{a_x^* |a_x|^2\}\mathbb{E}\{a_x |a_y|^2\}+2\mathbb{E}\{a_x^2\}\mathbb{E}^*\{a_x^2 |a_y|^2\}+\mathbb{E}\{a_x |a_x|^2\}\mathbb{E}\{a_x |a_y|^2\}+4\mathbb{E}\{a_x a_y^*\}\mathbb{E}\{a_x^*a_y|a_x|^2\}$\\$-8\mathbb{E}\{a_x^2\}\mathbb{E}^*\{a_xa_y\}\mathbb{E}\{a_x^*a_y\}\}$} \\
\hline\hline

$\rchi_1(f)$ & $\int_{-R_s/2}^{R_s/2}\int_{-R_s/2}^{R_s/2}|P(f_1)|^2|P(f_2)|^2|P(f-f_1+f_2)|^2|\eta(f_1,f_2,f)|^2 df_1\,df_2$\\
\hline
$\rchi_2(f)$ & $\int_{-R_s/2}^{R_s/2}\int_{-R_s/2}^{R_s/2}|P(f_1)|^2|P(f_2)|^2|P(f-f_1+f_2)|^2\eta(f_1,f_2,f)\eta^*(f_1,-f+f_1-f_2,f) df_1\,df_2$ \\
\hline
$\rchi_3(f)$ & $|P(f)|^2\int_{-R_s/2}^{R_s/2}\int_{-R_s/2}^{R_s/2}|P(f_1)|^2|P(f_2)|^2\eta(f_1,-f,f)\eta^*(f_2,-f,f) df_1\,df_2$ \\
\hline
$\rchi_4(f)$ & \makecell[l]{$\int_{-R_s/2}^{R_s/2}\int_{-R_s/2}^{R_s/2}\int_{-R_s/2}^{R_s/2}P(f_1)P^*(f_2)P(f-f_1+f_2)P^*(f_1-f_2)P(f_3)P^*(f-f_1+f_2+f_3)\eta(f_1,f_2,f)$\\$\cdot\eta^*(f_1-f_2,f_3,f) df_1\,df_2\,df_3$} \\
\hline
$\rchi_5(f)$ & \makecell[l]{$\int_{-R_s/2}^{R_s/2}\int_{-R_s/2}^{R_s/2}\int_{-R_s/2}^{R_s/2}P(f_1)P^*(f_2)P(f-f_1+f_2)P^*(f_3)P(f_2-f_1)P^*(f-f_1+f_2-f_3)\eta(f_1,f_2,f)$\\$\cdot\eta^*(f_3,f_2-f_1,f)df_1\,df_2\,df_3$}  \\
\hline
$\rchi_6(f)$ & \makecell[l]{$\int_{-R_s/2}^{R_s/2}\int_{-R_s/2}^{R_s/2}\int_{-R_s/2}^{R_s/2}P(f_1)P^*(f_2)P(f-f_1+f_2)P^*(f_3)P(-f-f_2)P^*(f_2+f_3)\eta(f_1,f_2,f)$\\$\cdot\eta^*(f_3,-f-f_2,f)df_1\,df_2\,df_3$}  \\
\hline
$\rchi_7(f)$ 
&$P(f)\int_{-R_s/2}^{R_s/2}\int_{-R_s/2}^{R_s/2}\int_{-R_s/2}^{R_s/2}|P(f_1)|^2P^*(f_2)P(f_3)P^*(f-f_2+f_3)\eta(f_1,f,f)\eta^*(f_2,f_3,f)df_1\,df_2\,df_3$
\\
\hline
$\rchi_8(f)$ &\makecell[l]{$\int_{-R_s/2}^{R_s/2}\int_{-R_s/2}^{R_s/2}\int_{-R_s/2}^{R_s/2}|P(f_1)|^2P^*(f_2)P(f-f_1+f_2)P(f_3)P^*(f-f_1+f_3)\eta(f_1,f_2,f)\eta^*(f_1,f_3,f)df_1\,df_2\,df_3$}\\
\hline
$\rchi_9(f)$ &  \makecell[l]{$\int_{-R_s/2}^{R_s/2}\int_{-R_s/2}^{R_s/2}\int_{-R_s/2}^{R_s/2}|P(f_1)|^2P^*(f_2)P(f-f_1+f_2)P^*(f_3)P^*(f-f_1-f_3)\eta(f_1,f_2,f)\eta^*(f_3,-f_1,f)df_1\,df_2\,df_3$}\\
\hline
$\rchi_{10}(f)$ & \makecell[l]{$\int_{-R_s/2}^{R_s/2}\int_{-R_s/2}^{R_s/2}\int_{-R_s/2}^{R_s/2}P(f_1)|P(f_2)|^2P(f-f_1+f_2)P^*(f_3)P^*(f+f_2-f_3)\eta(f_1,f_2,f)\eta^*(f_3,f_2,f)df_1\,df_2\,df_3$} \\
\hline
$\rchi_{11}(f)$ & \makecell[l]{$\int_{-R_s/2}^{R_s/2}\int_{-R_s/2}^{R_s/2}\int_{-R_s/2}^{R_s/2}\int_{-R_s/2}^{R_s/2}P(f_1)P^*(f_2)P(f-f_1+f_2)P^*(f_3)P(f_4)P^*(f-f_3+f_4)\eta(f_1,f_2,f)$\\$\cdot\eta^*(f_3,f_4,f)df_1\,df_2\,df_3\,df_4$}\\
\hline\hline
\end{tabular}
\end{table}
}

\section{Discussion and Conclusions}
In this work, we have derived a comprehensive analytical expression for the NLI power when a general dual-polarisation 4D modulation format is transmitted. The transmitted format is only assumed to be zero-mean. This result extends the model in \cite{Carena2014} by accounting for any constellation geometry and statistic in four dimensions. This is done by lifting two underlying assumptions in \cite{Carena2014} (and in other existing models): i) the transmitted formats are PM versions of a 2D format; ii) some high-order moments of the transmitted modulation format, such as $\mathbb{E}\{a^2_x\}$ and $\mathbb{E}\{a^3_x\}$, are implicitly assumed to be equal to zero.

The presented results are derived in a single-channel transmission scenario. However,  as it can be inferred from previous works, extending the expressions to the wavelength-division multiplexing (WDM) case does not lead to a different set of statistical moments of the transmitted constellation in the NLI power expression. An extension of this work to the WDM transmission scenario will be addressed in future versions of this manuscript.

Future work will also focus on comparing the presented model with possible heuristic extensions of existing PM-2D models to the general 4D case, e.g., by using the 4D constellation normalised fourth-order moment (or so-called kurtosis). Lastly, 4D constellation shaping in the optical fibre channel arguably represents the most attractive application for the model derived in this paper. 
\appendices

\section{Proof of Theorem \ref{The:RP_NLSE}}\label{app:theoRP}

The Manakov equation \eqref{eq:Manakov} can be written in frequency domain as
\begin{align}
    \frac{\partial \boldsymbol{E}(f,z)}{\partial z}&=-\frac{\alpha}{2}\boldsymbol{E}(f,z)+j4\pi^2f^2 \frac{\beta_2}{2} \boldsymbol{E}(f,z)+j\gamma \frac{8}{9}\mathcal{F}\{|\tilde{\boldsymbol{E}}(t,z)|^2 \tilde{\boldsymbol{E}}(t,z)\}\nonumber \\
    &=\left(-\frac{\alpha}{2}+j4\pi^2f^2 \frac{\beta_2}{2}\right)\boldsymbol{E}(f,z)+j\gamma\frac{8}{9}\mathcal{F}\{|\tilde{\boldsymbol{E}}(t,z)|^2\}*\boldsymbol{E}(f,z),
\label{eq:Manakov_FD}
\end{align}
where $*$ denotes a modified convolution operator between a scalar function and a vector function\footnote{For a scalar function $\alpha$, and a vector function $\boldsymbol{B}=[B_x,B_y]^{T}$, the operator $\alpha*\boldsymbol{B}$ is here defined as $\alpha*\boldsymbol{B}\triangleq\left[\alpha*B_x,\alpha*B_y\right]^T$.}. Expanding the nonlinear term in \eqref{eq:Manakov_FD}, we have 
\begin{equation*}
\mathcal{F}\{|\tilde{\boldsymbol{E}}(t,z)|^2\}*\boldsymbol{E}(f,z)=\left[\mathcal{F}\{\tilde{E}_{x}(t,z)\tilde{E}^*_{x}(t,z)\}+\mathcal{F}\{\tilde{E}_{y}(t,z)\tilde{E}^*_{y}(t,z)\}\right]*[E_x(f,z),E_y(f,z)]^T,
\label{eq:convolution}
\end{equation*}
which, for instance for the $x$ component, becomes 
\begin{equation}
E_{x}(f,z)*E^*_{x}(-f,z)*E_x(f,z)+E_{y}(f,z)*E^*_{y}(-f,z)*E_{x}(f,z).
\label{eq:convolution_x}
\end{equation}
Expanding the first term in \eqref{eq:convolution_x} we obtain
\begin{equation}
E_{x}(f,z)*E^*_{x}(-f,z)*E_x(f,z)=\int_{-\infty}^{\infty}\int_{-\infty}^{\infty} E_x(f_1,z)E^*_x(f_1-f_2,z)E_x(f-f_2,z)df_1df_2,
\label{eq:convolution_x2_a}
\end{equation}
which by substitution $f_1-f_2=\tilde{f}_2$
becomes\footnote{For notation's simplicity, the integration variable $\tilde{f}_2$ is relabelled as $f_2$.}
\begin{equation}
E_{x}(f,z)*E^*_{x}(-f,z)*E_x(f,z)=-\int_{-\infty}^{\infty}\int_{-\infty}^{\infty} E_x(f_1,z)E^*_x(f_2,z)E_x(f-f_1+f_2,z)df_1df_2.
\label{eq:convolution_x2_b}
\end{equation}
Similarly to the steps in \eqref{eq:convolution_x2_a} and \eqref{eq:convolution_x2_b}, the second term in \eqref{eq:convolution_x} can be found as 
\begin{equation*}
E_{y}(f,z)*E^*_{y}(-f,z)*E_x(f,z)=-\int_{-\infty}^{\infty}\int_{-\infty}^{\infty} E_y(f_1,z)E^*_y(f_2,z)E_x(f-f_1+f_2,z)df_1df_2.
\label{eq:convolution_yyx}
\end{equation*}
The $x$ component in \eqref{eq:Manakov_FD} can be then rewritten as 
\begin{equation}
\begin{split}
\frac{\partial E_x(f,z)}{\partial z}&=\left(-\frac{\alpha}{2}+j2\pi^2f^2\beta_2\right)E_x(f,z)-j\frac{8}{9}\gamma\int_{-\infty}^{\infty}\int_{-\infty}^{\infty}\left[ E_x(f_1,z)E^*_x(f_2,z)E_x(f-f_1+f_2,z)\right.\\
&\left.+E_y(f_1,z)E^*_y(f_2,z)E_x(f-f_1+f_2,z)\right]df_1df_2.
\end{split}
\label{eq:Manakov_expanded}
\end{equation}

Following the first-order RP approach to finding the solution to the Manakov equation \cite{Vannucci2002}, we replace the $x$ component of the first-order expansion in \eqref{eq:RP_general} into \eqref{eq:Manakov_expanded} and equate terms with the same power of $\gamma$. After some algebra and after substituting the $\boldsymbol{A}_n$ terms with the corresponding $\boldsymbol{E}_n$ using \eqref{eq:AtoE}, we find the following set of differential equations
\begin{align}
\frac{\partial E_{0,x}(f,z)}{\partial z}&=\left(-\frac{\alpha}{2}+j2\pi^2f^2\beta_2\right)E_{0,x}(f,z),
\label{eq:zero_order_de}
\\
\begin{split}
\frac{\partial E_{1,x}(f,z)}{\partial z}&=-j\frac{8}{9}\gamma\int_{-\infty}^{\infty}\int_{-\infty}^{\infty}\left[ E_{0,x}(f_1,z)E^*_{0,x}(f_2,z)E_{0,x}(f-f_1+f_2,z)\right.\\
&\left.+E_{0,y}(f_1,z)E^*_{0,y}(f_2,z)E_{0,x}(f-f_1+f_2,z)\right]df_1df_2.
\label{eq:first_order_de}
\end{split}
\end{align}

The zeroth-order term for a single fibre span of length $z$ is given by
\begin{equation}
E_{0,x}(f,z)= E(f,0)e^{(-\alpha/2+j2\pi^2\beta_2f^2)z}.  
 \label{eq:zero_th_order}
\end{equation}
On the other hand, the first-order term (for the $x$ component)  $E_{1,x}(f,z)$, with initial conditions given by the transmitted signal $\boldsymbol{E}(f,0)$, can be found solving the following differential equation
\begin{equation}
\begin{split}
\frac{\partial E_{1,x}(f,z)}{\partial z}=&\left(-\frac{\alpha}{2}+j2\pi^2\beta_2f^2\right)E_{1,x}(f,z)-j\frac{8}{9}\gamma\int_{-\infty}^{\infty}\int_{-\infty}^{\infty} \left[E_{0,x}(f_1,z)E^*_{0,x}(f_2,z)E_{0,x}(f-f_1+f_2,z)\right.\\
&\left.+E_{0,y}(f_1,z)E^*_{0,y}(f_2,z)E_{0,y}(f-f_1+f_2,z)\right]df_1df_2.
\end{split}
\label{eq:1st_order_diff_equ}    
\end{equation}
The solution to \eqref{eq:1st_order_diff_equ} with initial condition $E_{1,x}(f,0)=0$ is given by 
\begin{subequations}
\begin{align}
\begin{split}
E_{1,x}(f,z)&=-j\frac{8}{9}\gamma e^{\left(-\alpha z+j2\beta_2\pi^2f^2 z\right)}\int_{0}^{z} e^{\left(\frac{\alpha}{2}-j2\beta_2\pi^2f^2\right)z^{\prime}}\int_{-\infty}^{\infty}\int_{-\infty}^{\infty}\left[E_{0,x}(f_1,z^{\prime})E^*_{0,x}(f_2,z^{\prime})E_{0,x}(f-f_1+f_2,z^{\prime})\right.\\
& \left.+E_{0,y}(f_1,z^{\prime})E^*_{0,y}(f_2,z^{\prime})E_{0,x}(f-f_1+f_2,z^{\prime})\right]df_1df_2 d z^{\prime}\label{eq:1st_order_RP_solution1}
\end{split}
\\
\begin{split}
&=-j\frac{8}{9}\gamma e^{\left(-\alpha z+j2\beta_2\pi^2f^2 z\right)}\int_{0}^{z} e^{\left(\frac{\alpha}{2}-j2\beta_2\pi^2f^2\right)z^{\prime}}\int_{-\infty}^{\infty}\int_{-\infty}^{\infty}\left[E_{x}(f_1,0)E^*_{x}(f_2,0)E_{x}(f-f_1+f_2,0)\right.
\\
&\left.+E_{y}(f_1,0)E^*_{y}(f_2,0)E_{x}(f-f_1+f_2,0)\right]e^{-\frac{3}{2}\alpha z^{\prime}}e^{j4\pi^2\frac{\beta_2}{2}(f_1^2-f_2^2+(f-f_1+f_2)^2)z^{\prime}}   df_1df_2 d z^{\prime},
\end{split}
\label{eq:1st_order_RP_solution}
\end{align}
\end{subequations}
where \eqref{eq:zero_th_order} was used in the step from \eqref{eq:1st_order_RP_solution1} to \eqref{eq:1st_order_RP_solution}.

The power profile assumed in Sec.~\ref{sec:system_model} for the multi-span optical link is exponentially decaying with a lumped amplification at the end of each span which brings it back to the transmitted power level. This leads to a discontinuity in the function $\alpha(z)$ across the interface where an amplifier is located. For such a power profile, we can solve the differential equations \eqref{eq:zero_order_de}, \eqref{eq:first_order_de} by exploiting the continuity of their coefficients within each span, and imposing the initial conditions at the input of each new fibre span $E_{0,x}(f,lL_s^{+})=e^{\alpha L_s/2}E_{0,x}(f,lL_s^{-})$ and $E_{1,x}(f,lL_s^{+})=e^{\alpha L_s/2}E_{1,x}(f,lL_s^{-})$, for $l=1,2,...,N_s$, where $z=lL_s^{-}$ and $z=lL_s^{+}$ indicate the sections at the input and at the output of the $l$-th amplifier, respectively. Thus, we obtain that the zeroth and first-order term after $N_s$ fibre spans are given by
\begin{align}
\begin{split}
E_{0,x}(f,N_s,L_s)&=E(f,0)e^{j2\pi^2\beta_2f^2N_sL_s},
\label{eq:1st_order_RP_solution31}
\end{split}
\\
\begin{split}
E_{1,x}(f,N_s,L_s)&=-j\frac{8}{9}\gamma e^{j2\beta_2\pi^2f^2 N_sL_s}\sum_{l=1}^{N_s}\int_{(l-1)L_s}^{lL_s}  e^{\left(\frac{\alpha}{2}-j2\beta_2\pi^2f^2\right)z^{\prime}}  \\
&\cdot\int_{-\infty}^{\infty}\int_{-\infty}^{\infty}\left[E_{0,x}(f_1,z^{\prime})E^*_{0,x}(f_2,z^{\prime})E_{0,x}(f-f_1+f_2,z^{\prime}) \right.  \\
&\left.+E_{0,y}(f_1,z^{\prime})E^*_{0,y}(f_2,z^{\prime})E_{0,x}(f-f_1+f_2,z^{\prime})\right]df_1df_2 dz^{\prime}.
\label{eq:1st_order_RP_solution32}
\end{split}
\end{align}
Using \eqref{eq:1st_order_RP_solution31} in \eqref{eq:1st_order_RP_solution32}, and swapping the integral in $z^{\prime}$ with the double integral in $df_1df_2$, we obtain
\begin{align*}
\begin{split}
E_{1,x}(f,N_s,L_s)&=-j\frac{8}{9}\gamma e^{j2\beta_2\pi^2f^2 N_sL_s}\int_{-\infty}^{\infty}\int_{-\infty}^{\infty}\left[E_{x}(f_1,0)E^*_{x}(f_2,0)E_{x}(f-f_1+f_2,0)\right.\\
&\left.+E_{y}(f_1,0)E^*_{y}(f_2,0)E_{x}(f-f_1+f_2,0)\right]\sum_{l=1}^{N_s}\int_{(l-1)L_s}^{lL_s} e^{\left[-\alpha+j\beta_2(f-f_1)(f_2-f_1)\right]z^{\prime}}d z^{\prime}df_1df_2\\
=-j\frac{8}{9}\gamma e^{j2\beta_2\pi^2f^2 N_sL_s}\int_{-\infty}^{\infty}&\int_{-\infty}^{\infty}\left[E_{x}(f_1,0)E^*_{x}(f_2,0)E_{x}(f-f_1+f_2,0)+E_{y}(f_1,0)E^*_{y}(f_2,0)E_{x}(f-f_1+f_2,0)\right]\\
&\cdot\frac{1-e^{-\alpha z}e^{j\beta_2(f-f_1)(f_2-f_1)z}}{\alpha-j\beta_2(f-f_1)(f_2-f_1)}\sum_{l=1}^{N_s}e^{-j4\pi^2\beta_2 (l-1)(f-f_1)(f_2-f_1)L_s}df_1df_2.
\end{split}
\end{align*}

The $y$-component of the zeroth order term $E_{0,y}(f,z)$ and first-order term $E_{1,y}(f,z)$ can be found using the transformation $x\rightarrow y$, $y\rightarrow x$ in \eqref{eq:1st_order_RP_solution31} and \eqref{eq:1st_order_RP_solution32}, respectively. Finally, bringing together the $x$ and $y$ components, we have
\begin{align*}
\begin{split}
\boldsymbol{E}_0(f,N_s,L_s)=&\boldsymbol{E}_0(f,0)e^{j2\pi^2f^2\beta_2N_sL_s},   
\end{split}
\\
\begin{split}
\boldsymbol{E}_{1}(f,N_s,L_s)=-j\frac{8}{9}\gamma e^{j2\beta_2\pi^2f^2N_sL_s}\int_{-\infty}^{\infty}\int_{-\infty}^{\infty}& \boldsymbol{E}^{T}(f_1,0)\boldsymbol{E}^*(f_2,0)\boldsymbol{E}(f-f_1+f_2,0)\eta(f_1,f_2,f,z)df_1df_2,
\end{split}
\end{align*}
where $\eta(f,f_1,f_2,z)$ is defined in \eqref{eq:fwm_efficiency}, which proves the theorem.

\section{Proof of Proposition \ref{prop:conj}}\label{app:B}
Applying the variable transformation $\tilde{k}=k^{\prime},  \tilde{m}=m^{\prime}, \tilde{n}=n^{\prime}, \tilde{k}^{\prime}=k, \tilde{m}^{\prime}=m, \tilde{n}^{\prime}=n$, to the left-hand side of \eqref{eq:PSDconj} we obtain
\begin{subequations}
\begin{align}
&\sum_{\substack{(k,m,n) \in \mathcal{S}_i\\ (k^{\prime},m^{\prime},n^{\prime}) \in \mathcal{S}_i}}\mathcal{P}_{k,m,n,k^{\prime},m^{\prime},n^{\prime}}\nu_{x,k}\nu^*_{x,m}\nu_{x,n}\nu^*_{y,k^{\prime}}\nu_{y,m^{\prime}}\nu^*_{x,n^{\prime}}\eta_{k,n,m}\eta^*_{k^{\prime},n^{\prime},m^{\prime}}\\
&=\sum_{\substack{(\tilde{k}^{\prime},\tilde{m}^{\prime},\tilde{n}^{\prime}) \in \mathcal{S}_i\\ (\tilde{k},\tilde{m},\tilde{n}) \in \mathcal{S}_i}}
\mathcal{P}_{\tilde{k}^{\prime},\tilde{m}^{\prime},\tilde{n}^{\prime},\tilde{k},\tilde{m},\tilde{n}}\nu_{x,\tilde{k}^{\prime}}\nu^*_{x,\tilde{m}^{\prime}}\nu_{x,\tilde{n}^{\prime}}\nu^*_{y,\tilde{k}}\nu_{y,\tilde{m}}\nu^*_{x,\tilde{n}}\eta_{\tilde{k}^{\prime},\tilde{m}^{\prime},\tilde{n}^{\prime}}\eta^*_{\tilde{k},\tilde{m},\tilde{n}}\label{eq:PSDconjproof_b}\\
&=\sum_{\substack{(\tilde{k},\tilde{m},\tilde{n}) \in \mathcal{S}_i\\
(\tilde{k}^{\prime},\tilde{m}^{\prime},\tilde{n}^{\prime}) \in \mathcal{S}_i
}}\mathcal{P}^*_{\tilde{k},\tilde{m},\tilde{n},\tilde{k}^{\prime},\tilde{m}^{\prime},\tilde{n}^{\prime}}(\nu^*_{y,\tilde{k}}\nu_{y,\tilde{m}}\nu^*_{x,\tilde{n}}\nu_{x,\tilde{k}^{\prime}}\nu^*_{x,\tilde{m}^{\prime}}\nu_{x,\tilde{n}^{\prime}}\eta_{\tilde{k},\tilde{m},\tilde{n}}\eta^*_{\tilde{k}^{\prime},\tilde{m}^{\prime},\tilde{n}^{\prime}})^*\label{eq:PSDconjproof_c}\\
&=\biggl(\sum_{\substack{(\tilde{k},\tilde{m},\tilde{n}) \in \mathcal{S}_i\\
(\tilde{k}^{\prime},\tilde{m}^{\prime},\tilde{n}^{\prime}) \in \mathcal{S}_i
}}\mathcal{P}_{\tilde{k},\tilde{m},\tilde{n},\tilde{k}^{\prime},\tilde{m}^{\prime},\tilde{n}^{\prime}}\nu_{y,\tilde{k}}\nu^*_{y,\tilde{m}}\nu_{x,\tilde{n}}\nu_{x,\tilde{k}^{\prime}}^*\nu_{x,\tilde{m}^{\prime}}\nu^*_{x,\tilde{n}^{\prime}}\eta^*_{\tilde{k},\tilde{m},\tilde{n}}\eta_{\tilde{k}^{\prime},\tilde{m}^{\prime},\tilde{n}^{\prime}}\biggl)^*,
\label{eq:PSDconjproof_d}
\end{align}
\end{subequations}
where in the step between \eqref{eq:PSDconjproof_b} and \eqref{eq:PSDconjproof_c} we have used the property
\begin{equation}
\mathcal{P}_{\tilde{k}^{\prime},\tilde{m}^{\prime},\tilde{n}^{\prime},\tilde{k},\tilde{m},\tilde{n}}=\mathcal{P}^*_{\tilde{k},\tilde{m},\tilde{n},\tilde{k}^{\prime},\tilde{m}^{\prime},\tilde{n}^{\prime}}\label{eq:PconjP}
\end{equation}
which can be easily verified based on definition \eqref{eq:P_def}. Using the relabelling $\tilde{k}\rightarrow k,  \tilde{m}\rightarrow m, \tilde{n}\rightarrow n, \tilde{k}^{\prime}\rightarrow k^{\prime},\tilde{m}^{\prime}\rightarrow m^{\prime},\tilde{n}^{\prime}\rightarrow n^{\prime}$ for \eqref{eq:PSDconjproof_d}
the proposition is proven.

\section{Proof of Proposition \ref{prop:1}}\label{app:C}
Applying the variable transformation $\tilde{k}=n, \tilde{m}=m, \tilde{n}=k, \tilde{k}^{\prime}=n^{\prime}, \tilde{m}^{\prime}=m^{\prime},\tilde{n}^{\prime}=k^{\prime}$ to the right-hand side of \eqref{eq:prop_1} we have

\begin{align}
\begin{split}
\sum_{\substack{(k,m,n) \in \mathcal{S}_i \\ (k^{\prime},m^{\prime},n^{\prime}) \in \mathcal{S}_i}}\mathcal{P}_{k,m,n,k^{\prime},m^{\prime},n^{\prime}}\eta_{k,m,n}\eta^*_{k^{\prime},m^{\prime},n^{\prime}}D_2(k,m,n,k^{\prime},m^{\prime},n^{\prime})\\
=\sum_{\substack{(\tilde{n},\tilde{m},\tilde{k}) \in \mathcal{S}_i \\ (\tilde{n}^{\prime},\tilde{m}^{\prime},\tilde{k}^{\prime})\in \mathcal{S}_i}}\mathcal{P}_{\tilde{n},\tilde{m},\tilde{k},\tilde{n}^{\prime},\tilde{m}^{\prime},\tilde{k}^{\prime}}\eta_{\tilde{n},\tilde{m},\tilde{k}}\eta^*_{\tilde{n}^{\prime},\tilde{m}^{\prime},\tilde{k}^{\prime}}D_2(\tilde{n},\tilde{m},\tilde{k},\tilde{n}^{\prime},\tilde{m}^{\prime},\tilde{k}^{\prime}).
\end{split}\label{eq:km_transf}
\end{align}
From definitions \eqref{eq:eta_def} and \eqref{eq:P_def} it can be easily verified that $\mathcal{P}_{n,m,k,n^{\prime},m^{\prime},k^{\prime}}=\mathcal{P}_{k,m,n,k^{\prime},m^{\prime},n^{\prime}}$ and
$\eta_{n,m,k}=\eta_{k,m,n}$. Moreover, based on the definition of the set $\mathcal{S}_i$ in \eqref{eq:setSi}, it can be observed that the condition $(k,m,n)\in\mathcal{S_i}$ is equivalent to $(n,m,k)\in\mathcal{S_i}$, i.e. generates the same set of triplets $(k,m,n)$. We can, thus, write \eqref{eq:km_transf} as 
\begin{align*}
&\sum_{\substack{(k,m,n) \in \mathcal{S}_i \\ (k^{\prime},m^{\prime},n^{\prime}) \in \mathcal{S}_i}}\mathcal{P}_{k,m,n,k^{\prime},m^{\prime},n^{\prime}}\eta_{k,m,n}\eta^*_{k^{\prime},m^{\prime},n^{\prime}}D_2(k,m,n,k^{\prime},m^{\prime},n^{\prime})\\
&=\sum_{\substack{(\tilde{k},\tilde{m},\tilde{n}) \in \mathcal{S}_i \\ (\tilde{k}^{\prime},\tilde{m}^{\prime},\tilde{n}^{\prime})\in \mathcal{S}_i}}\mathcal{P}_{\tilde{k},\tilde{m},\tilde{n},\tilde{k}^{\prime},\tilde{m}^{\prime},\tilde{n}^{\prime}}\eta_{\tilde{k},\tilde{m},\tilde{n}}\eta^*_{\tilde{k}^{\prime},\tilde{m}^{\prime},\tilde{n}^{\prime}}D_2(\tilde{n},\tilde{m},\tilde{k},\tilde{n}^{\prime},\tilde{m}^{\prime},\tilde{k}^{\prime})\\
&=\sum_{\substack{(\tilde{k},\tilde{m},\tilde{n}) \in \mathcal{S}_i \\ (\tilde{k}^{\prime},\tilde{m}^{\prime},\tilde{n}^{\prime})\in \mathcal{S}_i}}{P}_{\tilde{k},\tilde{m},\tilde{n},\tilde{k}^{\prime},\tilde{m}^{\prime},\tilde{n}^{\prime}}\eta_{\tilde{k},\tilde{m},\tilde{n}}\eta^*_{\tilde{k}^{\prime},\tilde{m}^{\prime},\tilde{n}^{\prime}}D_1(\tilde{k},\tilde{m},\tilde{n},\tilde{k}^{\prime},\tilde{m}^{\prime},\tilde{n}^{\prime}),   
\end{align*}
which proves the proposition.

\section{Proof of Proposition  \ref{prop:4}}\label{app:D}
Since $D_1(k,m,n,k^{\prime},m^{\prime},n^{\prime})=D_2(k^{\prime},m^{\prime},n^{\prime},k,m,n)$, the left-hand side of \eqref{eq:prop2} can be written as  
\begin{align}
\begin{split}
&\sum_{\substack{(k,m,n) \in \mathcal{S}_i \\ (k^{\prime},m^{\prime},n^{\prime}) \in \mathcal{S}_i}}\mathcal{P}_{k,m,n,k^{\prime},m^{\prime},n^{\prime}}\eta_{k,m,n}\eta^*_{k^{\prime},m^{\prime},n^{\prime}}D_1(k,m,n,k^{\prime},m^{\prime},n^{\prime})\\
&=\sum_{\substack{(k,m,n) \in \mathcal{S}_i \\ (k^{\prime},m^{\prime},n^{\prime}) \in \mathcal{S}_i}}\mathcal{P}_{k,m,n,k^{\prime},m^{\prime},n^{\prime}}\eta_{k,m,n}\eta^*_{k^{\prime},m^{\prime},n^{\prime}}D_2(k^{\prime},m^{\prime},n^{\prime},k,m,n). 
\label{eq:d1_d2}
\end{split}
\end{align}
Using the change of variables $\tilde{k}=k^{\prime},  \tilde{m}=m^{\prime}, \tilde{n}=n^{\prime}, \tilde{k}^{\prime}=k, \tilde{m}^{\prime}=m, \tilde{n}^{\prime}=n$, the right-hand side of 
\eqref{eq:d1_d2} can be equivalently expressed as
\begin{subequations}
\begin{align}
&\sum_{\substack{(k,m,n) \in \mathcal{S}_i \\ (k^{\prime},m^{\prime},n^{\prime}) \in \mathcal{S}_i}}\mathcal{P}_{k,m,n,k^{\prime},m^{\prime},n^{\prime}}\eta_{k,m,n}\eta^*_{k^{\prime},m^{\prime},n^{\prime}}D_2(k^{\prime},m^{\prime},n^{\prime},k,m,n)\label{eq:93a}\\
&=\sum_{\substack{(\tilde{k}^{\prime},\tilde{m}^{\prime},\tilde{n}^{\prime}) \in \mathcal{S}_i \\ (\tilde{k},\tilde{m},\tilde{n})\in \mathcal{S}_i}}\mathcal{P}_{\tilde{k}^{\prime},\tilde{m}^{\prime},\tilde{n}^{\prime},\tilde{k},\tilde{m},\tilde{n}}\eta^*_{\tilde{k},\tilde{m},\tilde{n}}\eta_{\tilde{k}^{\prime},\tilde{m}^{\prime},\tilde{n}^{\prime}}D_2(\tilde{k},\tilde{m},\tilde{n},\tilde{k}^{\prime},\tilde{m}^{\prime},\tilde{n}^{\prime})\label{eq:93b}
\\
&=\sum_{\substack{(\tilde{k},\tilde{m},\tilde{n}) \in \mathcal{S}_i \\ (\tilde{k}^{\prime},\tilde{m}^{\prime},\tilde{n}^{\prime})\in \mathcal{S}_i}}\mathcal{P}^*_{\tilde{k},\tilde{m},\tilde{n},\tilde{k}^{\prime},\tilde{m}^{\prime},\tilde{n}^{\prime}}\eta^*_{\tilde{k},\tilde{m},\tilde{n}}\eta_{\tilde{k}^{\prime},\tilde{m}^{\prime},\tilde{n}^{\prime}}D_2(\tilde{k},\tilde{m},\tilde{n},\tilde{k}^{\prime},\tilde{m}^{\prime},\tilde{n}^{\prime})\label{eq:93c}\\
&=\biggl(\sum_{\substack{(\tilde{k},\tilde{m},\tilde{n}) \in \mathcal{S}_i \\ (\tilde{k}^{\prime},\tilde{m}^{\prime},\tilde{n}^{\prime})\in \mathcal{S}_i}}\mathcal{P}_{\tilde{k},\tilde{m},\tilde{n},\tilde{k}^{\prime},\tilde{m}^{\prime},\tilde{n}^{\prime}}\eta_{\tilde{k},\tilde{m},\tilde{n}}\eta^*_{\tilde{k}^{\prime},\tilde{m}^{\prime},\tilde{n}^{\prime}}D_2(\tilde{k},\tilde{m},\tilde{n},\tilde{k}^{\prime},\tilde{m}^{\prime},\tilde{n}^{\prime})\biggl)^*,\label{eq:93d}
\end{align}
\end{subequations}
where in the step from \eqref{eq:93b} to \eqref{eq:93c} we have used \eqref{eq:PconjP}. 
Eq.~\eqref{eq:93d} is identical to the right-hand side of \eqref{eq:prop2} up to the variable relabelling $\tilde{k}\rightarrow k,  \tilde{m}\rightarrow m, \tilde{n}\rightarrow n, \tilde{k}^{\prime}\rightarrow k^{\prime},\tilde{m}^{\prime}\rightarrow m^{\prime},\tilde{n}^{\prime}\rightarrow n^{\prime}$, which proves the proposition.

\section{Proof of Lemma \ref{lm:T_dim}}\label{app:E}
To prove the statement about the dimensionality of the sets $\mathcal{T}_{l,i}$ we take as an example the cases for $l=1,2,3$. In these instances, the sets $\mathcal{T}_{l,i}$ $\forall i\in \mathbb{Z}$ are identified by 5 linear constraints on the set of variables $(k,m,n,k^{\prime},m^{\prime},n^{\prime}) \in \{0,1,\ldots,W-1\}^6$ given by: i) the  2 linearly independent constraints, $(k,m,n) \in \mathcal{S}_i$ and $(k^{\prime},m^{\prime},n^{\prime}) \in \mathcal{S}_i$; ii) and the 3 linearly independent constraints induced by the condition $D^{(l)}=1$ for $l=1,2,3$ (see Table \ref{tab:delta_prod}). Let then $\boldsymbol{A}_{l}\triangleq[\boldsymbol{a}_1^{T}; \boldsymbol{a}_2^{T};\dots;\boldsymbol{a}_5^{T}]$ be a $5\times6$ matrix whose rows $\boldsymbol{a}_k$, $k=1,2,\dots,5,$ describe each of these 5 linear combinations,  $\boldsymbol{x}\triangleq [k,m,n,k^{\prime},m^{\prime},n^{\prime}]$, and $\boldsymbol{y}_i=[i,i,0,0,0]$. Thus the set $\mathcal{T}_{l,i}$ can be equivalently defined as 
\begin{equation}
\mathcal{T}_{l,i}=\{\boldsymbol{x} \in \{0,1,\dots,W-1\}^6 : \boldsymbol{A}_l\boldsymbol{x}=\boldsymbol{y}_i\}.
\label{eq:matrix_Tli}    
\end{equation}
From \eqref{eq:matrix_Tli} it can be seen that $\mathcal{T}_{l,i}$ is a vector space whose number of dimensions is given by  
\begin{equation}
\dim\{\mathcal{T}_{l,i}\}=6-\text{rank}(\boldsymbol{A}_l).
\label{eq:dim_Tli}    
\end{equation}
Due to the construction of the delta products $D^{(l)},\; l\in\{1,2,3\}$ it can be  shown that the rows of $\boldsymbol{A}_l$ are linearly dependent under the relationship $\boldsymbol{a}_1-\boldsymbol{a}_2=\pm\boldsymbol{a}_3\pm\boldsymbol{a}_4\pm\boldsymbol{a}_5$. Hence, $\forall\; l \in \{1,2,3\}$ and $i \in \mathbb{Z}$ we have  rank$(\boldsymbol{A}_l)=4$. As a result, from \eqref{eq:dim_Tli}, $\dim\{\mathcal{T}_{l,i}\}=2, \;\forall\; l \in \{1,2,3\}$ and $i \in \mathbb{Z}$.

For $l\in \{4,5,\dots,10\}$, we have that $\mathcal{T}_{l,i}$ is identified by 4 linear constraints, 2 of them related to the $\mathcal{S}_i$ set and 2 to the condition $D^{(l)}=1$. Furthermore, it can be seen that $\boldsymbol{a}_1-\boldsymbol{a}_2=\pm\boldsymbol{a}_3\pm\boldsymbol{a}_4$, hence leading to rank$(\boldsymbol{A}_l)=3,\; \forall\; l \in \{4,5,\dots,10\}$ and $\dim \{\mathcal{T}_{l,i}\}=3$. Finally, based on similar arguments one can show that rank$(\boldsymbol{A}_l)=2$ for $l=11$ and dim$\{\mathcal{T}_{l,i}\}=4$, which proves the lemma.

\section{Proof of Theorem \ref{th:keyresult}}\label{app:F}
The limit of a sequence of distributions $\mathsf{f}(\Delta_f)$ can be defined as the distribution $\tilde{\mathsf{f}}$ such that \cite[Sec.~2.2]{Strichartz1994}
\begin{equation}
 	\langle\tilde{\mathsf{f}},\psi \rangle=\lim_{\Delta_f\rightarrow 0} 	\langle\mathsf{f}(\Delta_f),\psi 	\rangle, \qquad \forall\; \psi,
\label{eq:dist_limit}
\end{equation}
where 
\begin{equation}
\langle\mathsf{f},\psi \rangle\triangleq\int_{-\infty}^{\infty}\mathsf{f}\psi\,df   
\label{eq:dist_def}
\end{equation} 
denotes the functional corresponding to the distribution $\mathsf{f}$ applied to a generic test function $\psi$.
In particular, the delta distribution centered in $f_0$ is defined as 
\begin{equation}
\langle\delta_{f_0},\psi\rangle\triangleq\int_{-\infty}^{\infty}\delta(f-f_0)\psi(f)df=\psi(f_0).   \label{eq:delta_def}
\end{equation}
Based on \eqref{eq:dist_def}, we have for the distribution $S_x(f,N_s,L_s)$ in \eqref{eq:psd_periodic},
\begin{subequations}
\begin{align}
\begin{split}
\langle S_x(f,N_s,L_s),\psi\rangle&=\left(\frac{8}{9}\right)^2\gamma^2\Delta_f\biggl[R_s^3\Delta_f^{2}\biggl(\Phi_{1}\int_{-\infty}^{\infty}\sum_{i=-\infty}^{\infty}\sum_{\mathcal{T}_{1,i}}\mathsf{P}\delta(f-i\Delta_f)\psi(f)df+...\\
+\Phi_{3}\int_{-\infty}^{\infty}\sum_{i=-\infty}^{\infty}\sum_{\mathcal{T}_{3,i}}&\mathsf{P}\delta(f-i\Delta_f)\psi(f)df\biggl)+R_s^2\Delta_f^{3}\biggl(\Psi_1\int_{-\infty}^{\infty}\sum_{i=-\infty}^{\infty}\sum_{\mathcal{T}_{4,i}}\mathsf{P}\delta(f-i\Delta_f)\psi(f)df+...\\
+\Lambda_6\int_{-\infty}^{\infty}\sum_{i=-\infty}^{\infty}\sum_{\mathcal{T}_{10,i}}&\mathsf{P}\delta(f-i\Delta_f)\psi(f)df\biggl)+R_s\Delta_f^{4}\Xi_1\int_{-\infty}^{\infty}\sum_{i=-\infty}^{\infty}\sum_{\mathcal{T}_{11,i}}\mathsf{P}\delta(f-i\Delta_f)\psi(f)df\biggl]
\label{eq:delta_train_functional1}
\end{split}\\
\begin{split}
&=\left(\frac{8}{9}\right)^2\gamma^2\biggl[R_s^3\Delta_f^{3}\biggl(\Phi_{1}\sum_{i=-\infty}^{\infty}\sum_{\mathcal{T}_{1,i}}\mathsf{P}\psi(i\Delta_f)+...+\Phi_{3}\sum_{i=-\infty}^{\infty}\sum_{\mathcal{T}_{3,i}}\mathsf{P}\psi(i\Delta_f)\biggl)\\
+R_s^2\Delta_f^{4}\biggl(\Psi_1\sum_{i=-\infty}^{\infty}\sum_{\mathcal{T}_{4,i}}&\mathsf{P}\psi(i\Delta_f)+...+\Lambda_6\sum_{i=-\infty}^{\infty}\sum_{\mathcal{T}_{10,i}}\mathsf{P}\psi(i\Delta_f)\biggl)+R_s\Delta_f^{5}\Xi_1\sum_{i=-\infty}^{\infty}\sum_{\mathcal{T}_{11,i}}\mathsf{P}\psi(i\Delta_f)\biggl],
\label{eq:delta_train_functional2}
\end{split}
\end{align}
\end{subequations}
and where we have used \eqref{eq:delta_def} in the step between \eqref{eq:delta_train_functional1} and \eqref{eq:delta_train_functional2}.

Now we want to show that all the terms in \eqref{eq:delta_train_functional2} are multidimensional Riemann sums, which then will converge to multidimensional integrals in the limit for $\Delta_f\rightarrow 0$. From \eqref{eq:eta_def}, \eqref{eq:P_def} and \eqref{eq:Psf_def}, it can be seen that the terms $\mathsf{P}\psi(i\Delta_f)$ are samples on a multidimensional grid of step $\Delta_f$ of the multivariate function 
\begin{align}
\begin{split}
\tilde{\mathsf{P}}(f_1,f_2,f_3,f_1^{\prime},f_2^{\prime},f_3^{\prime})&\triangleq P(f_1)P^*(f_2)P(f_3)P^*(f_1^{\prime})P(f_2^{\prime})P^*(f_3^{\prime})\\
&\cdot\eta(f_1,f_2,f_1-f_2+f_3,N_s,L_s)\eta^*(f_1^{\prime},f_2^{\prime},f_1^{\prime}-f_2^{\prime}+f_3^{\prime},N_s,L_s), \qquad f_1, f_2, f_3, f_1^{\prime},f_2^{\prime},f_3^{\prime} \in \mathbb{R}.
\end{split}
\label{eq:Pcont_fun}    
\end{align}
Moreover, $\Delta_f^{t(l)}$, which represents the power of $\Delta_f$ multiplying the $l$th element in \eqref{eq:delta_train_functional2}, where 
\begin{align}
t(l)\triangleq
\begin{dcases}
 3 &  \text{for} \;\;  l=1,2,3; \\
 4 & \text{for} \;\; l=4,...,10;\\
 5 & \text{for} \;\; l=11;
\end{dcases}
\label{eq:t_l}
\end{align}
is a measure of the $t(l)$th dimensional hypercube in $\mathbb{R}^{t(l)}$ whose side measures $\Delta_f$. Hence, to prove that each term in \eqref{eq:delta_train_functional2} converges to a sum of multiple integrals of the multivariate functions $\tilde{\mathsf{P}_{l}}\psi(f)$ we simply need to show that the dimensionality of the summation sets, i.e. $\mathbb{Z}\times \mathcal{T}_{l,i}$, is equal to $t(l)$, i.e., $\text{dim} \{\mathcal{T}_{l,i}\}=t(l)-1$, for $l=1,...,11$, and $\forall i \in \mathbb{Z}$. This can be easily verified comparing Lemma \ref{lm:T_dim} to \eqref{eq:t_l}. Defining the subspaces of $\mathbb{R}^6$
\begin{equation}
\mathcal{Q}_{l}(f) \triangleq \{(f_1,f_2,f_3,f_1^{\prime},f_2^{\prime},f_3^{\prime}) \in \mathbb{R}^6\cap\mathcal{G}_l: f_1-f_2+f_3=f,\; f_1^{\prime}-f_2^{\prime}+f_3^{\prime}=f\},
\label{eq:calQdef}
\end{equation}
where $\mathcal{G}_l$ is the set defined by the condition $D^{(l)}=1$ and the discrete variables $(k,m,n,k^{\prime},m^{\prime},n^{\prime})$ are replaced by the continuous ones $(f_1,f_2,f_3,f_1^{\prime},f_2^{\prime},f_3^{\prime}),$
we have
\begin{subequations}
\begin{align}
\begin{split}
\lim_{\Delta_f \rightarrow 0}\langle S_x(f,N_s,L_s),\psi\rangle&=\left(\frac{8}{9}\right)^2\gamma^2\biggl[R_s^3\biggl(\Phi_1\int_{-\infty}^{\infty}\idotsint\limits_{\mathcal{Q}_{1}(f)}\tilde{\mathsf{P}}(f_1,...,f_3^{\prime})\psi(f)df_1\,...\,df_3^{\prime}\,df+...\\
+\Phi_3\int_{-\infty}^{\infty}\idotsint\limits_{\mathcal{Q}_{3}(f)}\tilde{\mathsf{P}}(f_1,&...,f_3^{\prime})\psi(f)df_1\,...df_3^{\prime}\,df\biggl)+R_s^2\biggl(\Psi_1\int_{-\infty}^{\infty}\idotsint\limits_{\mathcal{Q}_{4}(f)}\tilde{\mathsf{P}}(f_1,...,f_3^{\prime})\psi(f)df_1\,...df_3^{\prime}\,df+...\\
+\Lambda_6\int_{-\infty}^{\infty}\idotsint\limits_{\mathcal{Q}_{10}(f)}&\tilde{\mathsf{P}}(f_1,...,f_3^{\prime})\psi(f)df_1\,...df_3^{\prime}\,df\biggl)+R_s\Xi_1\int_{-\infty}^{\infty}\idotsint\limits_{\mathcal{Q}_{11}(f)}\tilde{\mathsf{P}}(f_1,...,f_3^{\prime})\psi(f)df_1\,...df_3^{\prime}\,df\biggl]
\label{eq:Int_dist_a}
\end{split}
\\
\begin{split}
&=\left(\frac{8}{9}\right)^2\gamma^2\biggl[R_s^3\biggl(\Phi_1\int_{-\infty}^{\infty}\int_{-\infty}^{\infty}\int_{-\infty}^{\infty}\tilde{\mathsf{P}}_1(f_1,f_2,f)\psi(f)df_1\,df_2\,df+...\\
+\Phi_3\int_{-\infty}^{\infty}\int_{-\infty}^{\infty}\int_{-\infty}^{\infty}\tilde{\mathsf{P}}_{3}(f_1&,f_2,f)\psi(f)df_1\,df_2\,df\biggl)+R_s^2\biggl(\Psi_1\idotsint_{\mathbb{R}^4}\tilde{\mathsf{P}}_{4}(f_1,...,f_3,f)\psi(f)df_1\,...\,df_3\,df+...\\
+\Lambda_6\idotsint_{\mathbb{R}^4}\tilde{\mathsf{P}}_{10}(f_1,&...,f_3,f)\psi(f)df_1\,...\,df_3\,df\biggl)+R_s\Xi_1\idotsint_{\mathbb{R}^5}\tilde{\mathsf{P}}_{11}(f_1,...,f_4,f)\psi(f)df_1\,...\,df_4\,df\biggl].
\label{eq:Int_dist_b}
\end{split}
\end{align}
\end{subequations}
In the step from \eqref{eq:Int_dist_a} to \eqref{eq:Int_dist_b}, we have replaced in each integrand the function $\tilde{\mathsf{P}}$ in \eqref{eq:Pcont_fun} with its constrained version over $\mathcal{Q}_{l}(f)$   
\begin{equation}
\tilde{\mathsf{P}}_{l}\triangleq\tilde{\mathsf{P}}(f_1,f_2,f_3,f_1^{\prime},f_2^{\prime},f_3^{\prime})|_{(f_1,f_2,f_3,f_1^{\prime},f_2^{\prime},f_3^{\prime})\in \mathcal{Q}_{l}(f)},    
\end{equation}
and explicitly expressed the dimensionality of the integrals based on the dimension of their corresponding integration domains $\mathcal{Q}_{l}(f)$.~By construction (see \eqref{eq:calQdef}), $\text{dim}\{\mathcal{Q}_{l}(f)\}=t(l)-1$, for $l=1,...,11$.

Finally, using \eqref{eq:dist_limit} and comparing definition \eqref{eq:dist_def} with \eqref{eq:Int_dist_b}, we obtain    
\begin{align*}
\bar{S}_x(f,N_s,L_s)&=\lim_{\Delta_f \rightarrow 0}S_x(f,N_s,L_s)=\left(\frac{8}{9}\right)^2\gamma^2\biggl[R_s^3\biggl(\Phi_1\int_{-\infty}^{\infty}\int_{-\infty}^{\infty}\tilde{\mathsf{P}}_{1}(f_1,f_2,f)df_1\,df_2+...\\
+\Phi_3&\int_{-\infty}^{\infty}\int_{-\infty}^{\infty}\tilde{\mathsf{P}}_{3}(f_1,f_2,f)df_1\,df_2\biggl)+R_s^2\biggl(\Psi_1\int_{-\infty}^{\infty}\int_{-\infty}^{\infty}\int_{-\infty}^{\infty}\tilde{\mathsf{P}}_{4}(f_1,f_2,f_3,f)df_1\,df_2\,df_3\,+...\\
+\Lambda_6&\int_{-\infty}^{\infty}\int_{-\infty}^{\infty}\int_{-\infty}^{\infty}\tilde{\mathsf{P}}_{10}(f_1,f_2,f_3,f)df_1\,df_2\,df_3\biggl)+R_s\Xi_1\idotsint_{\mathbb{R}^4}\tilde{\mathsf{P}}_{11}(f_1,...,f_4,f)df_1...\,df_4\biggl],    
\end{align*}
which, defining
\begin{align}
    \rchi_l(f)\triangleq
    \begin{dcases}
     \int_{-\infty}^{\infty}\int_{-\infty}^{\infty}\tilde{\mathsf{P}}_{l}(f_1,f_2,f)df_1\,df_2=\int_{-\frac{R_s}{2}}^{\frac{R_s}{2}}\int_{-\frac{R_s}{2}}^{\frac{R_s}{2}}\tilde{\mathsf{P}}_{l}(f_1,f_2,f)df_1\,df_2, \qquad l=1,2,3;\\
     \int_{-\infty}^{\infty}\int_{-\infty}^{\infty}\tilde{\mathsf{P}}_{l}(f_1,f_2,f_3,f)df_1\,df_2\,df_3=\int_{-\frac{R_s}{2}}^{\frac{R_s}{2}}\int_{-\frac{R_s}{2}}^{\frac{R_s}{2}}\int_{-\frac{R_s}{2}}^{\frac{R_s}{2}}\tilde{\mathsf{P}}_{l}(f_1,f_2,f_3,f)df_1\,df_2\,df_3, \; l=4,\dots,10;\\
     \idotsint_{\mathbb{R}^4}\tilde{\mathsf{P}}_{11}(f_1,...,f_4,f)df_1...df_4=\idotsint_{\mathcal{R}_4}\tilde{\mathsf{P}}_{11}(f_1,...,f_4,f)df_1...df_4, \qquad l=11,
    \end{dcases}
\label{eq:Chi_def}    
\end{align}
with $\mathcal{R}_4\triangleq [-R_s/2,R_s/2]^4$, proves the theorem. The second equalities in \eqref{eq:Chi_def} are justified by the form of the functions $\tilde{\mathsf{P}}_{l}$ (see Table \ref{tab:final_result}) which, due to the assumption of strictly band-limited pulses, have limited support within the hybercube $[-R_s/2,R_s/2]^{t(l)-1}$. To derive the explicit expressions for the $\rchi_l(f)$ in Table~\ref{tab:final_result}, we used \eqref{eq:Chi_def} and the property $P(-f)=P^*(f)$, which stems from the fact that $p(t)$ is assumed to be real valued (see Sec.~\ref{sec:system_model}).

\bibliographystyle{IEEEtran}
\bibliography{references} 

\begin{thebibliography}{10}
\providecommand{\url}[1]{#1}
\csname url@samestyle\endcsname
\providecommand{\newblock}{\relax}
\providecommand{\bibinfo}[2]{#2}
\providecommand{\BIBentrySTDinterwordspacing}{\spaceskip=0pt\relax}
\providecommand{\BIBentryALTinterwordstretchfactor}{4}
\providecommand{\BIBentryALTinterwordspacing}{\spaceskip=\fontdimen2\font plus
\BIBentryALTinterwordstretchfactor\fontdimen3\font minus
  \fontdimen4\font\relax}
\providecommand{\BIBforeignlanguage}[2]{{%
\expandafter\ifx\csname l@#1\endcsname\relax
\typeout{** WARNING: IEEEtran.bst: No hyphenation pattern has been}%
\typeout{** loaded for the language `#1'. Using the pattern for}%
\typeout{** the default language instead.}%
\else
\language=\csname l@#1\endcsname
\fi
#2}}
\providecommand{\BIBdecl}{\relax}
\BIBdecl

\bibitem{Agrell09}
E.~Agrell and M.~Karlsson, ``Power-efficient modulation formats in coherent
  transmission systems,'' \emph{J. Lightwave Technol.}, vol.~27, no.~22, pp.
  5115--5126, Nov 2009.

\bibitem{Karlsson09}
M.~Karlsson and E.~Agrell, ``{Which is the most power-efficient modulation
  format in optical links?}'' \emph{Opt. Express}, vol.~17, no.~13, pp.
  10\,814--10\,819, jun 2009.

\bibitem{Alvarado2015}
A.~Alvarado and E.~Agrell, ``{Four-Dimensional Coded Modulation with Bit-Wise
  Decoders for Future Optical Communications},'' \emph{Journal of Lightwave
  Technology}, vol.~33, no.~10, pp. 1993--2003, 2015.

\bibitem{Eriksson2016}
T.~A. Eriksson, T.~Fehenberger, P.~A. Andrekson, M.~Karlsson, N.~Hanik, and
  E.~Agrell, ``{Impact of 4D Channel Distribution on the Achievable Rates in
  Coherent Optical Communication Experiments},'' \emph{Journal of Lightwave
  Technology}, vol.~34, no.~9, pp. 2256--2266, 2016.

\bibitem{Kojima2017}
K.~{Kojima}, T.~{Yoshida}, T.~{Koike-Akino}, D.~S. {Millar}, K.~{Parsons},
  M.~{Pajovic}, and V.~{Arlunno}, ``Nonlinearity-tolerant four-dimensional
  2a8psk family for 5–7 bits/symbol spectral efficiency,'' \emph{Journal of
  Lightwave Technology}, vol.~35, no.~8, pp. 1383--1391, 2017.

\bibitem{Chen2019}
B.~Chen, O.~Chigo, H.~Hafermann, and A.~Alvarado,
  ``{Polarization-ring-switching for nonlinearity-tolerant geometrically-shaped
  four-dimensional formats maximizing generalized mutual information},''
  \emph{Journal of Lightwave Technology}, vol.~37, no.~14, pp. 1--1, 2019.

\bibitem{Chen2020}
B.~Chen, A.~Alvarado, S.~van~der Heide, M.~van~den Hout, H.~Hafermann, and
  C.~Okonkwo, ``Analysis and experimental demonstration of orthant-symmetric
  four-dimensional 7 bit/4d-sym modulation for optical fiber communication,''
  2020.

\bibitem{Poggiolini2012}
\BIBentryALTinterwordspacing
P.~Poggiolini, G.~Bosco, A.~Carena, V.~Curri, Y.~Jiang, and F.~Forghieri, ``A
  detailed analytical derivation of the {GN} model of non-linear interference
  in coherent optical transmission systems,'' \emph{arXiv}, no. 1209.0394,
  2012. [Online]. Available: \url{http://arxiv.org/abs/1209.0394}
\BIBentrySTDinterwordspacing

\bibitem{Carena2014}
\BIBentryALTinterwordspacing
A.~Carena, G.~Bosco, V.~Curri, Y.~Jiang, P.~Poggiolini, and F.~Forghieri, ``On
  the accuracy of the {GN}-model and on analytical correction terms to improve
  it,'' \emph{arXiv}, no. 1401.6946v7, 2014. [Online]. Available:
  \url{http://arxiv.org/abs/1401.6946v7}
\BIBentrySTDinterwordspacing

\bibitem{Mecozzi2012}
A.~{Mecozzi} and R.~{Essiambre}, ``Nonlinear shannon limit in pseudolinear
  coherent systems,'' \emph{Journal of Lightwave Technology}, vol.~30, no.~12,
  pp. 2011--2024, 2012.

\bibitem{Dar2013}
R.~Dar, M.~Feder, A.~Mecozzi, and M.~Shtaif, ``Properties of nonlinear noise in
  long, dispersion-uncompensated fiber links,'' \emph{Opt. Express}, vol.~21,
  no.~22, pp. 25\,685--25\,699, Nov 2013.

\bibitem{Marcuse1997}
D.~Marcuse, C.~R. Menyuk, and P.~K.~A. Wai, ``Application of the {Manakov-PMD}
  equation to studies of signal propagation in optical fibers with randomly
  varying birefringence,'' \emph{Journal of Lightwave Technology}, vol.~15,
  no.~9, pp. 1735--1745, Sep. 1997.

\bibitem{Vannucci2002}
A.~Vannucci, P.~Serena, S.~Member, and A.~Bononi, ``The {RP} method: A new tool
  for the iterative solution of the nonlinear {S}chr{\"{o}}dinger equation,''
  \emph{Journal of Lightwave Technology}, vol.~20, no.~7, pp. 1102--1112, July
  2002.

\bibitem{Johannisson2013}
P.~Johannisson and M.~Karlsson, ``Perturbation analysis of nonlinear
  propagation in a strongly dispersive optical communication system,''
  \emph{Journal of Lightwave Technology}, vol.~31, no.~8, pp. 1273--1282, Apr.
  2013.

\bibitem{Colombeau1984}
J.~F. Colombeau, \emph{New Generalized Functions and Multiplication of
  Distributions}.\hskip 1em plus 0.5em minus 0.4em\relax North-Holland, 1984.

\bibitem{ProakisDSP3rdEd}
J.~G. Proakis and D.~G. Manolakis, \emph{Digital Signal Processing: Principles,
  Algorithms, and Applications}, 3rd~ed.\hskip 1em plus 0.5em minus 0.4em\relax
  USA: Prentice-Hall, Inc., 1996.

\bibitem{Carena2012}
A.~Carena, V.~Curri, G.~Bosco, P.~Poggiolini, and F.~Forghieri, ``Modeling of
  the impact of nonlinear propagation effects in uncompensated optical coherent
  transmission links,'' \emph{Journal of Lightwave Technology}, vol.~30,
  no.~10, pp. 1524--1539, May 2012.

\bibitem{Strichartz1994}
R.~Strichartz, \emph{A Guide to Distribution Theory and Fourier
  Transforms}.\hskip 1em plus 0.5em minus 0.4em\relax CRC-Press, 1994.

\end{thebibliography}

\end{document}